\documentclass[]{jair}

\setcopyright{cc}
\acmDOI{10.1613/jair.1.17921}   

\JAIRAE{Rafael Pe\~{n}aloza}
\JAIRTrack{} 

\acmVolume{87}
\acmArticle{12}
\acmMonth{10}
\acmYear{2026}
\copyrightyear{2026}

\RequirePackage[
  datamodel=acmdatamodel,
  style=acmauthoryear,
  backend=biber,
  giveninits=true,
  uniquename=init
  ]{biblatex}

\usepackage{graphicx}
\usepackage{amsmath} 
\usepackage{amsthm} 
\usepackage{stmaryrd}
\usepackage[all,cmtip]{xy}
\usepackage{arydshln}
\usepackage{stackengine}
\usepackage{tensor}
\usepackage{rotating}
\usepackage{tikz}

\newtheorem{definition}{Definition}
\newtheorem{postulate}{Postulate}
\newtheorem{theorem}{Theorem}
\newtheorem{corollary}{Corollary}
\newtheorem{lemma}{Lemma}

\def\T{{\small\textsf{T}}}
\def\F{{\small\textsf{F}}}
\def\TF{{\{\T,\F\}}}

\def\Y{{\small\textsf{Y}}}
\def\N{{\small\textsf{N}}} 
\def\YN{{\{\Y,\N\}}}

\def\xor{{\ \textsf{xor}\ }}

\def\botSet{{\{\bot_{01},\bot_{10},\bot_{11}\}}}

\def\EQ{\textsf{EQ}}
\def\DR{\textsf{DR}}
\def\PO{\textsf{PO}}
\def\PP{\textsf{PP}}
\def\PPi{\textsf{PPi}}
\def\botOneZero{{\bot_{10}}}
\def\botZeroOne{{\bot_{01}}}
\def\botOneOne{{\bot_{11}}}

\def\dia{\textsf{dia}}
\def\DIA{\textsf{Dia}}

\def\RccFive{{\textsf{R}^\textsf{5}}}
\def\RccFiveZeroBit{{\textsf{R}_\textsf{0-bit}}}
\def\RccFiveOneBit{{\textsf{R}_\textsf{$\parallel$-1bit}}}
\def\RccFiveTwoBit{{\textsf{R}_\textsf{$\parallel$-2bit}}}
\def\RccFiveThreeBit{{\textsf{R}_\textsf{$\parallel$-3bit}}}

\def\RccFiveOneBitX{{\textsf{R}_\textsf{$\times$-1bit}}}
\def\RccFiveTwoBitX{{\textsf{R}_\textsf{$\times$-2bit}}}

\def\RccFiveTwoBitXX{{\textsf{R}_\textsf{$\times\times$-2bit}}}

\def\RccFiveQA{{\textsf{P}_{\scriptscriptstyle\,\i\;\j\ \k}^{\scriptscriptstyle\alpha\beta\gamma}}}

\def\ORccFive{{{\textsf{R}}}}
\def\IdxSet{\textsf{I}}
\def\IdxSetX{{\textsf{IX}}}
\def\IdxSetNu{{\textsf{I}^\nu}}

\def\RccFiveTwoBitFms{{\PL_\textsf{$\parallel$-2-bit}}}

\def\RccFiveTwoBitFmsX{{\PL_\textsf{$\times$-2-bit}}}

\def\ThreeBit{{\textsf{3bit}\ }}

\def\rccFive{{\textsf{RCC5}\ }}
\def\rccFivePlus{{\textsf{RCC5}^+}}

\def\rccfive{\textsf{rcc5}^+}

\def\oQ{{\overline{Q}}}

\def\Reg{\textsf{REG}}
\def\OReg{{\textsf{Reg}}}

\def\OO{\mathcal{O}}

\def\cReg{{\textsf{cReg}}}
\def\vReg{{\textsf{vReg}}}
\def\eReg{{\textsf{eReg}}}
\def\zReg{\textsf{zReg}}

\def\ccReg{{\textsf{ccReg}}}
\def\ceReg{{\textsf{ceReg}}}
\def\cvReg{{\textsf{cvReg}}}
\def\czReg{{\textsf{czReg}}}
\def\ecReg{{\textsf{ecReg}}}
\def\eeReg{{\textsf{eeReg}}}
\def\evReg{{\textsf{evReg}}}
\def\ezReg{{\textsf{ezReg}}}
\def\vcReg{{\textsf{vcReg}}}
\def\veReg{{\textsf{veReg}}}
\def\vvReg{{\textsf{vvReg}}}
\def\vzReg{{\textsf{vzReg}}}
\def\zzReg{{\textsf{zzReg}}}
\def\zvReg{{\textsf{zvReg}}}
\def\zcReg{{\textsf{zcReg}}}
\def\zeReg{{\textsf{zeReg}}}

\def\starEReg{\textsf{*eReg}}
\def\EstarReg{\textsf{e*Reg}}

\def\xReg{{\textsf{xReg}}}
\def\yReg{{\textsf{yReg}}}
\def\xyReg{{\textsf{xyReg}}}

\def\ESPR{{{{\textsf{Reg}}}^2_\parallel}}
\def\ESET{{{{\textsf{Reg}}}^2_\times}}

\def\RegThreeBit{{\OReg^2_\textsf{$\parallel$-3-bit}}}
\def\RegOneBit{{\OReg^2_\textsf{1-bit}}}
\def\RegTwoBit{{\OReg^2_\textsf{2-bit}}}

\def\RegOneBitPR{{\OReg^2_\textsf{$\parallel$-1-bit}}}
\def\RegOneBitET{{\OReg^2_\textsf{$\times$-1-bit}}}
\def\RegTwoBitPR{{\OReg^2_\textsf{$\parallel$-2-bit}}}
\def\RegTwoBitET{{\OReg^2_\textsf{$\times$-2-bit}}}

\def\RegThreeBitET{{\OReg^2_\textsf{$\times\times$-2-bit}}}

\def\IdxZeroBit{{\textsf{0-bit}}}
\def\IdxOneBit{{\textsf{1-bit}}}
\def\IdxTwoBit{{\textsf{2-bit}}}
\def\IdxThreeBit{{\textsf{3-bit}}}
\def\IdxOneBitPR{{\textsf{$\parallel$-1-bit}}}
\def\IdxOneBitET{{\textsf{$\times$-1-bit}}}
\def\IdxTwoBitPR{{\textsf{$\parallel$-2-bit}}}
\def\IdxTwoBitET{{\textsf{$\times$-2-bit}}}
\def\IdxThreeBitET{{\textsf{$\times\times$-2-bit}}}
\def\IdxThreeBitPR{{\textsf{$\parallel$-3-bit}}}

\def\IdxOneToThreeBit{{\textsf{1-3-bit}}}

\def\meet{\sqcap}
\def\M{{\textsf{M}}}

\def\pr{\textsf{pr}}

\def\a{{\textsf{a}}}
\def\b{{\textsf{b}}}
\def\c{{\textsf{c}}}
\def\i{{\textsf{i}}}
\def\j{{\textsf{j}}}
\def\k{{\textsf{k}}}

\def\Q{{\textsf{Q}}}
\def\QQ{{\mathcal{Q}}}

\def\A{{\textsf{A}}}

\def\id{{\textsf{id}}}

\def\Models{{\ {\mid\models}\ }}

\def\PL{{\textsf{P}}}

\def\rel{{\textsf{rel}}}

\def\Jud{{\textsf{Jud}}}

\def\iform{\textsf{iform}}
\def\iforms{\textsf{iforms}}
\def\sharpDet{\textsf{Det}}

\def\simNu{{\stackanchor[1pt]{$\sim$}{$\scriptstyle \nu$}}}

\def\is{{\textsf{s}}}

\def\istate{{\textsf{istate}}}
\def\istates{{\textsf{istates}}}

\def\IntQ{{\textsf{I}_\textsf{Q}}}
\def\IntJ{{\textsf{I}_\textsf{J}}}

\def\AR{\text{AR}}
\def\FH{\text{FH}}
\def\NPB{\text{NPB}}
\def\MPF{\text{MPF}}

\newcommand{\lrceil}[1]{{\lceil {#1} \rceil}}

\newtheorem{example}{Example}
\newtheorem{remark}{Remark}

\begin{document}

\title[Quantized Mereology]{Quantized Mereology}

\author{Thomas Bittner}
\orcid{0000-0003-4204-1532}
\email{bittner3@buffalo.edu}
\affiliation{%
  \institution{State University of New York at Buffalo}
  \city{Buffalo}
  \state{New York}
  \country{USA}
}

\begin{abstract}
This paper investigates quantized mereology, which merges quantum information theory with classical mereology, the formal analysis of the part-whole relationship. It emphasizes that information about crisp mereological relations is captured by classical bits, whereas mereological vagueness involves the use of quantum bits, or qubits, to capture mereological information. Transitioning from classical bits to qubits offers a new understanding of part-whole relations among entities that are subject to mereological vagueness. This paper methodically studies the quantized mereology that arises from the classical \rccFive formalism, one of the most widely studied formalisms in Qualitative Reasoning (QR), a subfield of Artificial Intelligence. The argument is made that this method can be extended to quantize more sophisticated QR formalisms, mirroring how  strategies are used in physics to transition from classical to quantum systems. The study explores: (i) the conversion of bits to qubits while understanding that these (qu)bits convey mereological information; (ii) the representation of qubit states via set partitions, acknowledging that set theory aligns more closely with classical logic than with the complex vector spaces employed in quantum physics; (iii) demonstrating that classical QR formalisms can be seen as special cases within quantized formalisms; and (iv) showing that qubits more effectively encapsulate information about vague phenomena.
\end{abstract}

\received{22 December 2024}
\received[accepted]{7 July 2025}

\maketitle

\section{Introduction}
\label{sec:Introduction}

Profound challenges are encountered by classical mereology, the theory of part-whole relations as formalized by \citet{simons:parts}, when dealing with vagueness and indeterminacy \citep{sep-mereology,hawley:vagueness,sider:FourDimensionalism, Donnelly2014,Barnes2014}. While a precise language for crisp, well-defined entities is provided by formalisms like the Region Connection Calculus (RCC) \citep{Cohn01a}, the real world is often presented with entities like clouds, forests, or ecological habitats whose boundaries are irreducibly vague \citep{frank:IndeterminateBoundaries}. Paradoxes and conceptual difficulties are created when the sharp distinctions of classical mereology are applied to such cases.

In this paper, \textit{Quantized Mereology}, a new methodological framework that addresses this challenge by reframing the problem, is introduced and defended. Instead of vagueness being treated as a flaw of language \citep{fine:Vagueness:75} or knowledge \citep{williamson:vagueness}, it is posited by the framework that vagueness is an objective, ontological feature of a system \citep{Barnes2014}. The central thesis is that (mereological) vagueness is manifested as a fundamental, information-theoretic limit: less information about its relations can be carried by an elementary system composed of vague entities than by a system of crisp entities \citep{bittner:InformationMereologyAndVagueness, Bittner:VMQI}. (See Remark \ref{remark:VoidOfInformation} for an exception.) A shift in the formal tools that are used is forced by this information-theoretic limitation: a move from the classical truth values and bits of traditional logic to the richer structure of quantum bits, or qubits, corresponding to superpositions of truth values or correlation-based entanglement. 

This shift can be visualized using a geometric 'information cube' (Fig.~\ref{fig:rccFive}) where crisp relations are vertices and vagueness corresponds to higher-dimensional features like edges and faces. These features represent \textit{information states} ($\istates$) where one or more of the fundamental questions about the relationship (e.g., Overlap?, Part Of?, Inverse Part Of?) remain indeterminate. This framework is formally developed in Sections~\ref{sec:CrispMereology} and~\ref{sec:PartitionLattice}.

\begin{example}\normalfont
\label{ex:RunningExample}
To render this formalism accessible, a running example is introduced. The basis for the example is provided by an ecological system with two vague regions, an Acidic Region (\AR) and a Fern Habitat (\FH). Quantized mereology posits that, due to vagueness, less than 3 bits of mereological information (1 or 2 bits, depending on the observer) is possible for the ordered pair (\AR,\FH). The determination of the mereological relation between \AR\ and \FH\ depends on how the system is interacted with by an observer seeking answers to three \textit{fundamental questions}: ($\emptyset$) Do they overlap?  (1) Is \AR\ part of \FH? (2) Is \FH\ part of \AR? The \textit{information state}, or $\istate$, of the system with respect to the observer is the set of all mereological possibilities consistent with all the available answers.

\paragraph{The Soil Chemist ($O_1$), Orthogonal Vagueness:}
A "Yes" answer to "Overlap?" is obtained when a fern is found by a soil chemist growing in a highly acidic soil sample. The resulting $\istate$ is the set of all relations where overlap is true: $\{\PO, \PP, \PPi, \EQ\}$ (cf. Fig.~\ref{fig:rccFive}). This leaves the system in a state of \textit{orthogonal vagueness}; the question of overlap is settled (1 bit of information is obtained), but the possible answers to the two orthogonal questions about parthood remain completely indeterminate (not unknown).

\paragraph{The Remote Sensing Specialist ($O_2$), Non-Orthogonal Vagueness:}
A remote sensing specialist finds that the answers to the two parthood questions must be the same, but what those answers are is not determined. This establishes a state of \textit{non-orthogonal vagueness}, or entanglement. Here, the individual parthood questions are indeterminate, but their relationship (equality) is fixed. This state also corresponds to one bit of information, as it resolves the binary question of whether the two parthood relations are equal or unequal. Possible answers to the overlap question remain indeterminate.

\paragraph{The Entomologist ($O_3$), Void of Information:}
An entomologist investigates a Mosquito Population Feeding ground ($\MPF$) whose extent is impossible to determine. An entity void of mereological information is $\MPF$, which is represented by the empty region ($\emptyset$) (cf. Remark \ref{remark:VoidOfInformation}). The relation between $\FH$ and $\MPF$ is the crisp, determinate relation $\botZeroOne$ (cf. Fig.~\ref{fig:rccFive}), which signifies that the second entity is void of information.
\qed
\end{example}

\begin{remark}\normalfont
\label{remark:VoidOfInformation}
While the framework's central thesis is that vagueness corresponds to an information limit, a special case arises for entities treated as being 'void of mereological information'. Such entities, which lack any discernible mereological features to be vague \textit{about}, are formally mapped to the null region ($\emptyset$). Consequently, any relationship between a 'void' entity and a crisp region becomes fully determinate, as the relations between the null region and any crisp region are uniquely defined by the underlying lattice structure. Therefore, this specific pairing does not exhibit an information limit; instead, it represents a state of complete mereological information. The paper's core argument about information limits applies to entities with fuzzy or indeterminate features, not this unique case. \qed
\end{remark}

\begin{figure}[h!]
\begin{minipage}{9.5cm}
\centering
\begin{tikzpicture}[scale=0.75, every node/.style={scale=0.75}]

\draw[dotted] (0.85, 2) -- (0.85, -3);
\node at (-2,-3) {Overlap? = No };
\node at (3.5,-3) {Overlap? = Yes };
    \begin{scope}[shift={(-5.10,0)}]
        \node at (0,-0.6) {$r_1=\emptyset$};
        \node at (0,0.6) {$r_2=\emptyset$};
        \node[below] at (0,-1.3) {$\begin{array}{cc}\botOneOne\\(r_1,r_2)\end{array}$};
    \end{scope}

    \begin{scope}[shift={(-3.40,0)}]
        \draw[dashed] (0,-0.6) circle (0.5);
        \node at (0,-0.6) {$r_1$};
        \node at (0,0.6) {$r_2=\emptyset$};
        \node[below] at (0,-1.3) {$\begin{array}{cc}\botZeroOne\\(r_1,r_2)\end{array}$};
    \end{scope}

    \begin{scope}[shift={(-1.70,0)}]
        \node at (0,-0.6) {$r_1=\emptyset$};
        \draw (0,0.6) circle (0.5);
        \node at (0,0.6) {$r_2$};
        \node[below] at (0,-1.3) {$\begin{array}{cc}\botOneZero\\(r_1,r_2)\end{array}$};
    \end{scope}

    \begin{scope}[shift={(0,0)}]
        \draw[dashed] (0,-0.6) circle (0.5);
        \node at (0,-0.6) {$r_1$};
        \draw (0,0.6) circle (0.5);
        \node at (0,0.6) {$r_2$};
        \node[below] at (0,-1.3) {$\begin{array}{cc}\textsf{DR}\\(r_1,r_2)\end{array}$};
    \end{scope}

    \begin{scope}[shift={(1.7,0)}]
        \draw[dashed] (0,-0.3) circle (0.5);
        \node at (0,-0.5) {$r_1$};
        \draw (0,0.3) circle (0.5);
        \node at (0,0.5) {$r_2$};
        \node[below] at (0,-1.3) {$\begin{array}{cc}\textsf{PO}\\(r_1,r_2)\end{array}$};
    \end{scope}

    \begin{scope}[shift={(3.5, 1)}]
        \draw (0,0) circle (0.7);
        \node at (0,0.4) {$r_2$};
        \draw[dashed] (0,0) circle (0.3);
        \node at (0,0) {$r_1$};
        \node[above] at (0,0.8) {$\textsf{PP}(r_1,r_2)$};
    \end{scope}

    \begin{scope}[shift={(3.5,-1)}]
        \draw[dashed] (0,0) circle (0.7);
        \node at (0,0.4) {$r_1$};
        \draw (0,0) circle (0.3);
        \node at (0,0) {$r_2$};
        \node[below] at (0,-0.8) {$\textsf{PPi}(r_1,r_2)$};
    \end{scope}

    \begin{scope}[shift={(5.5,0)}]
        \draw (0,0) circle (0.6);
        \node at (0,0) {$r_1=r_2$};
        \node[below] at (0,-0.8) {$\begin{array}{cc}\EQ\\(r_1,r_2)\end{array}$};
    \end{scope}

\end{tikzpicture}
\end{minipage}
			\hspace{1cm}\begin{minipage}{4cm}
				\small
				\xymatrixcolsep{2.5pc}
				\xymatrixrowsep{2.5pc}
				\xymatrix@!0{
 & \txt{\small[011]\\$\botOneOne$} \ar@{->}[rr]\ar@{<-}'[d][dd]
   & & \txt{\small[111]\\\EQ} \ar@{<-}[dd]
\\
\txt{\small[001]\\$\botZeroOne$} \ar@{->}[ur]\ar@{->}[rr]\ar@{<-}[dd]^{\overrightarrow{z}}
 & & \txt{\small[101]\\\PPi} \ar@{->}[ur]\ar@{<-}[dd]
\\
 & \txt{\small[010]\\$\botOneZero$} \ar@{->}'[r][rr]
   & & \txt{\small[110]\\\PP}
\\
 \txt{\small[000]\\\DR} \ar@{->}[rr]_{\overrightarrow{x}}\ar@{->}[ur]_{\overrightarrow{y}}
 & & \txt{\small[100]\\$\PO$} \ar@{->}[ur]
}
\end{minipage}
\caption{\label{fig:rccFive}(left) Visualization of prototypical instances of the  $\rccFivePlus$ relations; (right) The \ThreeBit information cube (displayed in the $\emptyset12$ base) whose vertices correspond to rows of the table and whose (non-directed) edges arise from the graph of the \ThreeBit lattice. \cite{bittner:InformationMereologyAndVagueness,Bittner:VMQI}}
\end{figure}

In Quantized Mereology, this process of resolving indeterminacy is modeled as a two-part process of 'information-supported judgment'. The first stage, \textit{Information Acquisition ($\IntQ$)}, is the observer-world interaction by which a relational $\istate$ is established. The second stage, \textit{Cognitive Judgment ($\IntJ$)}, is a subsequent act where a specific classical relation is selected from the possibilities constrained by the $\istate$. The second stage is developed in \citet{bittner:Hilbert}. The construction of the rigorous, formal engine for that first stage, $\IntQ$, is the primary focus of this paper.

The paper proceeds by methodically "quantizing" the well-known $\rccFivePlus$ relations ($\rccFivePlus$ = \rccFive with the 'empty' region) to create a more expressive framework for handling vagueness. This approach is inspired by quantum physics in that it uses a similar formalism but its domain of application is very different: macroscopic regions vs. microscopic particles. A key aspect of this methodology is that qubit states are represented not with complex vector spaces, but with an accessible set-partition-based formalism, following  \citet{Ellerman2022}.

The internal coherence of this framework is established by the \textit{Quantized Information Semantics (QIS) Principle}, the central claim this paper will prove. The QIS principle demonstrates that two distinct pathways for determining a mereological information state always yield the same result. The first is a top-down \textit{direct semantic path}, where an answered question identifies an information state by specifying the set of all classical possibilities consistent with its meaning. The second is a bottom-up \textit{indirect quantization path}, which formally constructs the very same state by quantizing the classical bit-level information obtained through an observer-system interaction. This core principle of equivalence can be expressed in plain language as follows:
\begin{quote}
    \textit{The information state of a system derived by formally constructing it from the bottom up through an observer's interaction is equivalent to the information state derived from the top down by the semantic meaning of the question the observer asked.}
\end{quote}

The remainder of this paper will build the formal machinery necessary to prove this principle, showing how classical mereological relations emerge as a special, information-complete case of this more general framework.

\begin{remark}[A Note on the Term 'Quantization']
\normalfont\label{remark:Quantization}
The title of this paper, \textit{Quantized Mereology}, intentionally invokes the language of quantum theory, and it is crucial to clarify the nature of this connection at the outset. The framework presented here is \textit{not} a claim that macroscopic entities like clouds or habitats are governed by quantum physical laws. Rather, it leverages the powerful \textit{information-theoretic formalism} of quantum theory as a mathematical tool for modeling objective vagueness.

Specifically, while this work adopts core concepts like qubits, superposition, and entanglement, it diverges from the formalism used in physics. Instead of representing states as vectors in complex Hilbert spaces, this paper follows the work of \citet{Ellerman2022} by representing qubit states in a more accessible, set-theoretic manner: as cells within partitions of a set of classical outcomes. This approach allows the logic of quantum information to be applied to qualitative reasoning in a way that aligns more naturally with classical logic and set theory. Therefore, "quantization" in this context refers to the process of transitioning from the discrete, binary logic of classical bits to the richer, superpositional logic of qubits to better capture information about vague phenomena. \qed
\end{remark}

\subsection*{A Note to the Reader on Navigating This Paper}

This paper is structured to be read in two ways, both of which build towards understanding the paper's central idea: the \textit{Quantized Information Semantics (QIS) Principle}. You are invited to choose the path that best suits your interests. For a comprehensive list of all major notations used in this paper, their meanings, and pointers to their formal definitions, please consult the Glossary of Notation in Table~\ref{tab:Glossary} in the Appendix.

\paragraph{Path 1: The Conceptual Track}
If your goal is to understand the core argument and conclusions, reading the following sections in order is recommended:
\begin{itemize}
   \item {Section \ref{sec:Introduction} (Introduction):} Introduces the core problem, the running example, and the central QIS Principle.
   \item {Section \ref{sec:FromBinaryToProps} (From Binary Propositions to a Quantized Information Space):} Explains the key geometric intuition of the "information cube".
   \item {Section \ref{sec:PartitionLattice} (The Partition Lattice), focusing on Sec~\ref{sec:IndexNotationGuide} ("An informal guide"):} Introduces the set of all possible information states and decodes the index notation.
   \item {Section \ref{sec:ElementaryQuestions} (Elementary Mereological Questions), focusing on Sec.~\ref{sec:MereologicalRelations} and Sec.~\ref{sec:LinkingIstates}:} Develops the key components needed to understand the QIS principle.
   \item {Section \ref{sec:CompositionalityRoadmap} (The QIS Framework):} Formally presents the QIS Principle; its proof is completed in Sections \ref{sec:QuantizationMap} and \ref{sec:CombiningQuestions}.
   \item {Sections \ref{sec:Synthesis} \& \ref{sec:Conclusion} (Synthesis \& Conclusion):} Summarize the paper's contributions and significance.
\end{itemize}
Reading these sections alone will provide a self-contained and comprehensive understanding of the framework's essential features.

\paragraph{Path 2: The Formal Track}
If you are interested in the detailed proofs and mathematical machinery, the following sections provide the complete formal engine that underpins the conceptual claims:
\begin{itemize}
   \item {Section \ref{sec:CrispMereology} (Mereological Relations):} A review of the classical $\rccFivePlus$ formalism.
   \item {Sections \ref{sec:InteractionMaps}, \ref{sec:QuantizationMap}, \& \ref{sec:CombiningQuestions} (Interaction Maps, Quantization Map, \& Combining Questions):} These sections provide the complete, rigorous construction of the indirect quantization path and the direct semantic path of the QIS principle.
\end{itemize}

\section{From Binary Propositions to a Quantized Information Space}
\label{sec:FromBinaryToProps}

Many qualitative representation and reasoning techniques in Artificial Intelligence are based on sets of crisp binary relations that are jointly exhaustive and pair-wise disjoint (JEPD) \citep{Bennett00ecai,galton:Change:00,186,weld:QualitativeReasoning,cohn:TaxonomySpatialRelations,cohn:QSR}. The foundation of such a system can be understood as starting with a set of N fundamental binary propositions $\{Q_1, \dots, Q_N\}$, that can be evaluated for any pair of entities $(r_1, r_2)$ from a given domain.

For any pair, the evaluation yields a sequence of $N$ Boolean values, forming a "spectrum"\footnote{The author would like to thank Reviewer B who suggested this way of presenting the material.} that measures the relationship between the entities with respect to these propositions. This process naturally gives rise to $2^N$ possible outcomes, where each outcome corresponds to a unique sequence of truth values. Each of these sequences defines a specific, crisp relation. In this work, it is assumed that all $2^N$ combinations are possible, meaning none of the resulting relations are empty. These $2^N$ distinct crisp relations can be visualized as the vertices of an N-dimensional information cube, where the axes correspond to the fundamental propositions. This kind of construction is critical for interpreting truth values of (fundamental) propositions as bits of information.

\begin{example}[$\rccFivePlus$]\normalfont
\label{ex:rccFivePlus}
The $\rccFivePlus$ formalism, which is central to this paper, provides a concrete example where $N=3$. The three fundamental propositions about the meet operation ($Q_\emptyset, Q_1, Q_2$) (Tab.~\ref{tab:rccFive}~top) define the three axes of a 3-dimensional cube (the \ThreeBit information cube, Fig.~\ref{fig:rccFive}, right). The truth values of these propositions correspond to the Yes/No answers to the three fundamental questions of Example~\ref{ex:RunningExample}. The eight   mereological relations include the five relations \DR, \PO, \PP, \PPi, \EQ\ between `non-empty' regions and the three relations $\botOneZero$, $\botZeroOne$, $\botOneOne$ where at least one of the relata is the `empty' region (Fig.~\ref{fig:rccFive}~left). Jointly they correspond precisely to the eight vertices of the \ThreeBit information cube, where the three truth-values are interpreted as coordinates as illustrated in the figure.

The index notation employed in this paper provides a compact way to represent subspaces of the \ThreeBit cube. The vertices are zero-dimensional subspace and notated as $\PL_{\emptyset 1 2}^{\a\b\c}$, where the lower indices name the questions and the upper indices represent the corresponding answer, e.g.  $\{\PP\} = \PL_{\emptyset12}^{\T\T\F}$, etc.  The Entomologist's observation of Example~\ref{ex:RunningExample} corresponds to the crisp relation $\botZeroOne$, which in this notation is the complete 3-bit pattern $\PL_{\emptyset 1 2}^{\F\F\T}=\{\botZeroOne\}$.
\qed
\end{example}

The core thesis of Quantized Mereology is that vagueness is an objective limit on the information an elementary system (here an ordered pair of domain entities) can carry. This can be understood intuitively as a situation where not all $N$ truth values of the fundamental predicates are determinate. For instance, if the value of one proposition, say $Q_\emptyset$, is  indeterminate for a given ordered pair, two  crisp possibilities remain: one where $Q_\emptyset$ is true and one where it is false. Geometrically, these two possibilities correspond to an edge of the information cube connecting two adjacent vertices. Similarly, if two propositions are indeterminate and uncorrelated, four possibilities remain, corresponding to a face of the cube. These geometrically significant subsets of the cube, namely vertices (0D), edges (1D), and planes (2D), are precisely the \textit{information states} (\istates) that are used to model different kinds and degrees of mereological vagueness.

\begin{example}[Geometric Interpretation of an \istate]
\normalfont
\label{ex:GeometricInterpretation}
Recall from the running example (Example~\ref{ex:RunningExample}) that the elementary system (\AR, \FH) carries less than 3 bits of mereological information. The Soil Chemist's "Yes" answer to the overlap question resolves the system into the possibilities represented by the set $\PL_\emptyset^\T=\{\PO, \PP, \PPi, \EQ\}$. This  is not just a set; in the geometric framework, it corresponds to a 2D face on the information cube. The resulting state, denoted $\PL_{\, \emptyset\, 1\, 2}^{\T--}$, encapsulates a determinate answer for overlap while the parthood relations remain indeterminate. This state carries exactly 1 bit of information, as selecting this face is equivalent to choosing one cell from the underlying two-cell partition $\{\PL_\emptyset^\F, \PL_\emptyset^\T\}$, in which $\PL_\emptyset^\F$ is formed by the four remaining possibilities.  The Remote Sensing Specialist's entangled state is written as $\PL_{\emptyset 1 2}^{-\overline{\b\b}}$, signifying an indeterminate answer for overlap but a determinate correlation between the two (individually indeterminate) parthood questions. \qed
\end{example}

This geometric picture motivates the transition from classical bits to quantum bits (qubits). A crisp relation (a single vertex) is fully described by $N$ classical bits. However, a vague state (an edge or a face) represents a set of possibilities and cannot be described by a single $N$-bit string. In this framework, such a state is modeled as a qubit. Following the work of \citet{Ellerman2022}, these qubits are not represented as vectors in a  Hilbert space, but more accessibly as cells in a partition of the set of $2^N$ classical vertices. The remainder of this paper will formalize this structure, showing how these partitions form a lattice (Figure \ref{fig:RccFivePartitionLattice}) that serves as the foundation for a robust logic of vague relations.

\begin{remark}[A Primer on the Index Notation]\normalfont
\label{remark:IndexNotation}
To formally describe these information states, this paper uses a compact index notation of the form $\PL_{\emptyset12}^{\alpha\beta\gamma}$. This notation can be understood through a few core principles introduced here and detailed in Section~\ref{sec:PartitionLattice}. The three lower indices, $\emptyset, 1, 2$, serve as an "address space" corresponding to the three fundamental mereological questions: Overlap?, Part Of?, and Inverse Part Of?. The upper indices, $\alpha, \beta, \gamma$, represent the information content held in these slots, as follows: (i) A variable or a \textit{truth value} (e.g., $\a$, $\T$ or $\F$) indicating a determinate, classical answer to the corresponding question.
  (ii)   A \textit{hyphen} (e.g., $-$) signifies complete indeterminacy, where the answer is in a state of equal superposition.
  (iii)   An \textit{overline} (e.g., $\overline{\a\b}$) signifies entanglement, where the individual answers are indeterminate but are linked by a determinate correlation. This notational system provides a precise way to represent the various kinds of mereological information, from crisp classical facts to the orthogonal and non-orthogonal vagueness illustrated in the example above. \qed
\end{remark}

\section{Mereological Relations and Triples of Boolean Values}
\label{sec:CrispMereology}

This section reviews material from \cite{bittner:InformationMereologyAndVagueness,Bittner:VMQI}. Consider the set $\Reg$ that possesses a complete lattice structure, ensuring the existence of both the greatest lower bound and the least upper bound for any of its subsets. For any two elements in $\Reg$, the operations meet and join ($\meet,\sqcup:\Reg \times \Reg \to \Reg$) are, respectively, determined by the greatest lower bound and the least upper bound of the set containing those two elements. Furthermore, within $\Reg$, it is postulated that the operations meet and join are distributive. The set $\OReg$ is the set of mereological entities/regions. The sets $\OReg$ and $\Reg$ share common elements, although not all elements of $\OReg$ are included in $\Reg$. In particular, $\emptyset$ the smallest element in the lattice structure of $\Reg$, often called the empty region, is important as an abstract mathematical object but is not recognized as a mereological entity/region.

\subsection{$\rccFivePlus$ Relations in Terms of the Meet Operation} \label{sec:RccFiveRelations}

The meet operation, $\meet$, is used to define mereological relations in a way that allows for an information-theoretic analysis. The $\rccFivePlus$ relations between two elements $r_1,r_2\in\Reg$ are determined by three indexed Boolean values, as shown in the top of Tab. \ref{tab:rccFive}. The mapping $\rccfive$ of members of $\Reg\times\Reg$ (abbreviated as $\Reg^2$) to indexed Boolean values induces an equivalence relation $\sim$ that partitions the set $\Reg^2$.\footnote{The symbol $\sim$ will be used in three senses: as a binary equivalence relation between pairs of regions (i.e. $(r_1,r_2) \sim (s_1,s_2)$ with $(r_1,r_2),(s_1,s_2) \in \Reg^2$), as the Boolean negation (i.e. $\T=\sim\F$) and as the `negation' of answers (i.e. $\Y=\sim\N$).} The symbol `$\bot$' is used to denote the patterns of Boolean values that hold for pairs of members of $\Reg$ when at least one of the components is the minimal lattice element / the empty region. The subscripts of $\bot$ indicate when the first element, the second element, or both elements of the ordered pair are the empty region. The set $\RccFive=\{\DR,\PO, \PP,\PPi,\EQ\}$ corresponds to the usual set of \rccFive\ relations between (non-empty) regions in $\Reg$.

\begin{table}
\caption{\label{tab:rccFive}(top) Bit pattern (called \ThreeBit pattern) for $\rccFivePlus$ relations ($\RccFive = \ORccFive \setminus \botSet$) defined in terms of the mapping $\rccfive  :  \Reg^2 \to  (\IdxSet \to \TF)$ with $\IdxSet=\{\emptyset,1,2\}$; (bottom) The map $\rccfive$ induces the equivalence relation $(r_1,r_2)\sim(s_1,s_2)$ and the respective equivalence classes $[(r_1,r_2)]_\sim$  \cite{BittnerStellAMAI}.}
	\begin{center}
		\begin{minipage}{13cm}
				$$\begin{array}{cc|ccc|c}
				&& \multicolumn{3}{c|}{\rccfive   : (r_1,r_2) \mapsto } &\\
					r_1& r_2 & \emptyset \mapsto & 1 \mapsto & 2 \mapsto & \\
					= & = & r_1 \meet r_2, & r_1 \meet r_2, & r_1 \meet r_2 & \ORccFive \\
					\emptyset & \emptyset & \neq \emptyset & = r_1 & = r_2 \\
					\hline
					\F & \F & \F                  &         \F       &         \F                      & \DR \\
					\F & \T & \F                  &         \F       &         \T                      & \botZeroOne \\
					\T & \F & \F                  &         \T       &         \F                      & \botOneZero \\
					\T & \T & \F                  &         \T       &         \T                      & \botOneOne \\
					\F & \F & \T                  &         \F       &         \F                      & \PO \\
					\F & \F & \T                  &         \F       &         \T                    & \PPi \\
					\F & \F & \T                  &         \T       &         \F                     & \PP \\
					\F & \F & \T                  &         \T       &         \T                      & \EQ \\\hline
						\multicolumn{2}{r|}{\text{faces}} & \PL_\emptyset^\T,\PL_\emptyset^\F & \PL_1^\T,\PL_1^\F & \PL_2^\T,\PL_2^\F & \PL_\i^\a\\
\multicolumn{6}{c}{}\\
	\multicolumn{6}{c}{\text{See corresponding definitions in Tab.~\ref{tab:OverviewSupportSemantics} for:}} \\
\multicolumn{2}{r|}{\text{propositions}} & Q_\emptyset & Q_1 & Q_2 & Q_i\\\hline
			\multicolumn{2}{r|}{\txt{\\questions}} & \text{Overlap?} & \text{Parthood?} & \txt{Inverse\\Parthood}?\\
			\multicolumn{2}{r|}{} & Q_\emptyset? & Q_1? & Q_2? & Q_i? \\\hline
				\end{array}$$
			\end{minipage}
			
			\begin{minipage}{13cm}
	$$\arraycolsep=0.5pt\begin{array}{rcl}
	\\
	\IdxSet & = & \{\emptyset,1,2\}\\ 
		\rccfive & : & \Reg\times\Reg \to (\IdxSet \to \TF )\\
		\rccfive(r_1,r_2)    & = & \{\emptyset \mapsto (r_1 \meet r_2 \neq \emptyset),
		1\mapsto (r_1 \meet r_2 = r_1),
		 2\mapsto (r_1 \meet r_2 = r_2)\}\\
		 &\equiv & \llbracket r_1,r_2\rrbracket\\	
		\rccfive_i (r_1,r_2) & \equiv  & (\rccfive\, (r_1,r_2)\, i) \\\hline
		(r_1,r_2)\sim(s_1,s_2) & \equiv & (\rccfive\ (r_1,r_2)) = (\rccfive\  (s_1,s_2))\\
		\text{[}(r_1,r_2)\text{]}_\sim & = & \{ (s_1,s_2) \in \Reg^2 \mid  (r_1,r_2)\sim(s_1,s_2) \}\\
		\Reg^2/_{\sim} & = & \{ \text{[}(r_1,r_2)\text{]}_{\sim} \mid (r_1,r_2)\in\Reg^2\} \\
		\textsf{[\rel]} & = & \{ (r_1,r_2)\in \Reg^2 \mid \rel = \rccfive (r_1,r_2) \}\\
		\Reg^2/_{\sim} & = & \{[\DR],[\PO], [\PP],[\PPi],[\EQ], [\botZeroOne],[\botOneZero],[\botOneOne]\} = \rccFivePlus\\
		\ORccFive & = & \{\DR,\PO, \PP,\PPi,\EQ,
		 \botZeroOne,\botOneZero,\botOneOne\} = (\TF_i)_{i\in\IdxSet}\\
		 \\
	\end{array}$$
			\end{minipage}
		
		
	\end{center}
\end{table}

Let $\IdxSet = \{\emptyset, 1, 2\}$ be a set of indexes. The members of the index set $\IdxSet$ are assumed to be totally ordered, with the particular order $\emptyset < 1 < 2$ stipulated in this paper.  Each index $i\in\IdxSet$ is associated with an indexed Boolean variable $a_i$ such that $a_i = (\rccfive\, (r_1,r_2)\, i)= \rccfive_i(r_1,r_2)$ for ordered pairs $(r_1,r_2)\in\Reg^2$ according to the definition of the function $\rccfive$ in the top  of Tab. \ref{tab:rccFive}. That is, for $(r_1,r_2)$ the variable $a_\emptyset$ takes the truth value of formula $(r_1 \meet r_2 \neq \emptyset)$, while the value of $a_1$ corresponds to the truth value of $(r_1 \meet r_2 = r_1)$ and the value of $a_2$ corresponds to the truth value of $(r_1 \meet r_2 = r_2)$. In this view, Tab. \ref{tab:rccFive}(top) can be considered as a set of indexed families where $\ORccFive = \{ (a_i)_{i\in\IdxSet}\mid a_i\in \TF\}$ with $(a_i)_{i\in\IdxSet}= \{a_\emptyset, a_1,a_2\}$. The set $\ORccFive$ is also called the set of \ThreeBit pattern.

\subsection{Equivalence and Ordering Relations}


Consider the definitions of the $\rccFivePlus$ relations in terms of indexed families of Boolean values (\ThreeBit pattern) in  Tab. \ref{tab:rccFive}(top). The members of the set $\ORccFive$ are organized by a family of equivalence relations $\cong_\alpha$ that are established by setting $\{a_\emptyset,a_1,a_2\} \cong_\alpha \{b_\emptyset,b_1,b_2\}$ if $a_\alpha = b_\alpha$. There is a relation $\cong_\alpha$ for each index $\alpha \in \IdxSet= {\{\emptyset,1,2\}}$.  The sets of elements of $\ORccFive$ that are equivalent with respect to $\cong_\alpha$ are denoted as $\PL_\alpha^\a = \{ s\in\ORccFive \mid s_\alpha = \a \}$ with $\alpha\in\IdxSet$ and $\a\in\TF$. The sets $\PL_\emptyset^\T$, $\PL_\emptyset^\F$, $\PL_1^\T$, $\PL_1^\F$, $\PL_2^\T$ and $\PL_2^\F$ are shown in Figure \ref{fig:PairsOfPlanes}. The set $\ORccFive$ is divided into equivalence classes based on sets with the same lower indexes; for example, $\{\PL_\alpha^\T,\PL_\alpha^\F \}$ is one of those systems of classes, which is represented by $\ORccFive /_{\cong_\alpha}$.


Boolean values are ordered by setting $\F$ to less than $\T$.  A set of preorder relations $\le_\alpha$ is created by setting that $\{a_\emptyset,a_1,a_2\} \le_\alpha \{b_\emptyset,b_1,b_2\}$ if $a_\alpha \le b_\alpha$.  The relation $\le_\alpha$ is a preorder rather than a partial order because if $x \le_\alpha y$ and $y \le_\alpha x$ then $x$ and $y$ are $\alpha$-equivalent in the sense of $\cong_\alpha$.  Using $\cong_\alpha$ and $\le_\alpha$ one can define a family of immediate $\alpha$-predecessor relations by setting $r \ll_\alpha s$ if $r \le_\alpha s$ and $r \not\cong_\alpha s$ and $r \cong_\beta s$ for all $\beta$ not equal to $\alpha$. The immediate $\alpha$-neighbors are signified as $r \ll\gg_\alpha s$ iff $r \ll_\alpha s$ or $r \gg_\alpha s$.

\begin{figure}
	\begin{minipage}{16cm}
		\begin{minipage}{5cm}
			\footnotesize
			\xymatrixcolsep{2.7pc}
			\xymatrixrowsep{2.7pc}
			\xymatrix@!0{
				& \txt{[\textbf{0}11]\\$\top_\alpha$\\$\botOneOne$} \ar@{..}[rr]\ar@{-}'[d][dd]
				& & \txt{[\textbf{1}11]\\$\top_\ORccFive$\\\EQ} \ar@{-}[dd]
				\\
				\txt{[\textbf{0}01]\\$\bot_{\alpha^\uparrow}$\\$\botZeroOne$} \ar@{-}[ur]\ar@{..}[rr]\ar@{-}[dd]
				& & \txt{[\textbf{1}01]\\$\top_\alpha^\downarrow$\\\PPi} \ar@{-}[ur]\ar@{-}[dd]
				\\
				& \txt{[\textbf{0}10]\\$\bot_{\alpha^\downarrow}$\\$\botOneZero$} \ar@{..}'[r][rr]
				& & \txt{[\textbf{1}10]\\$\top_\alpha^\uparrow$\\\PP}
				\\
				\txt{[\textbf{0}00]\\$\bot_\ORccFive$\\\DR} \ar@{..}[rr]\ar@{-}[ur]
				& & \txt{[\textbf{1}00]\\$\bot_\alpha$\\\PO} \ar@{-}[ur]
			}
			$$\arraycolsep=1pt
			\begin{array}{rl}
			\alpha & =\emptyset\\
\\\\				\\
				\multicolumn{2}{c}{\text{the partition}\ \{\PL_{\emptyset\,1\,2}^{\a--}\} }\\
			\end{array}
			$$
		\end{minipage}\hfill
		\begin{minipage}{5cm}
			\footnotesize
			\xymatrixcolsep{2.7pc}
			\xymatrixrowsep{2.7pc}
			\xymatrix@!0{
				& \txt{[0\textbf{1}1]\\$\top_\alpha^\downarrow$\\$\botOneOne$} \ar@{-}[rr]\ar@{-}'[d][dd]
				& & \txt{[1\textbf{1}1]\\$\top_\ORccFive$\\\EQ} \ar@{-}[dd]
				\\
				\txt{[0\textbf{0}1]\\$\bot_\alpha^\uparrow$\\$\botZeroOne$} \ar@{..}[ur]\ar@{-}[rr]\ar@{-}[dd]
				& & \txt{[1\textbf{0}1]\\$\top_\alpha$\\\PPi} \ar@{..}[ur]\ar@{-}[dd]
				\\
				& \txt{[0\textbf{1}0]\\$\bot_\alpha$\\$\botOneZero$} \ar@{-}'[r][rr]
				& & \txt{[1\textbf{1}0]\\$\top_\alpha^\uparrow$\\\PP}
				\\
				\txt{[0\textbf{0}0]\\$\bot_\ORccFive$\\\DR} \ar@{-}[rr]\ar@{..}[ur]
				& & \txt{[1\textbf{0}0]\\$\bot_\alpha^\downarrow$\\\PO} \ar@{..}[ur]
			}
			$$\arraycolsep=1pt
			\begin{array}{rl}
			\alpha & = 1 \\
			\PL_\alpha^\F & = \{\bot_\ORccFive,\bot_\alpha^\downarrow,\top_\alpha,\bot_\alpha^\uparrow\}\\
				\PL_\alpha^\T &= \{\bot_\alpha,\top_\alpha^\uparrow,\top_\ORccFive,\top_\alpha^\downarrow\}\\
				\\
				\multicolumn{2}{c}{\text{the partition}\ \{\PL_{\emptyset\,1\,2}^{-\b-}\}}\\
			\end{array}
			$$
		\end{minipage}\hfill
		\begin{minipage}{5cm}
			\footnotesize
			\xymatrixcolsep{2.7pc}
			\xymatrixrowsep{2.7pc}
			\xymatrix@!0{
				& \txt{[01\textbf{1}]\\$\top_\alpha^\downarrow$\\$\botOneOne$} \ar@{-}[rr]\ar@{..}'[d][dd]
				& & \txt{[11\textbf{1}]\\$\top_\ORccFive$\\\EQ} \ar@{..}[dd]
				\\
				\txt{[00\textbf{1}]\\$\bot_\alpha$\\$\botZeroOne$} \ar@{-}[ur]\ar@{-}[rr]\ar@{..}[dd]
				& & \txt{[10\textbf{1}]\\$\top_\alpha^\uparrow$\\\PPi} \ar@{-}[ur]\ar@{..}[dd]
				\\
				& \txt{[01\textbf{0}]\\$\bot_\alpha^\uparrow$\\$\botOneZero$} \ar@{-}'[r][rr]
				& & \txt{[11\textbf{0}]\\$\top_\alpha$\\\PP}
				\\
				\txt{[00\textbf{0}]\\$\bot_\ORccFive$\\\DR} \ar@{-}[rr]\ar@{-}[ur]
				& &  \txt{[10\textbf{0}]\\$\bot_\alpha^\downarrow$\\\PO} \ar@{-}[ur]
			}
			$$\arraycolsep=1pt
			\begin{array}{rl}
			\alpha & = 2\\
\\\\				\\
				\multicolumn{2}{c}{\text{the partition}\ \{\PL_{\emptyset\, 1\, 2}^{--\c}\}}\\
			\end{array}
			$$
		\end{minipage}
	\end{minipage}
	
\vspace{0.5cm}

\begin{minipage}{16cm}
\begin{minipage}{5cm}
\footnotesize
\xymatrixcolsep{2.7pc}
\xymatrixrowsep{2.7pc}
\xymatrix@!0{
 & \txt{[0{1}1]\\$\top_\alpha$\\$\botOneOne$} \ar@{-}[rr]\ar@{..}'[d][dd]
   & & \txt{[1{1}1]\\$\top_\ORccFive$\\\EQ} \ar@{..}[dd]
\\
\txt{[001]\\$\bot_\alpha^\uparrow$\\$\botZeroOne$} \ar@{..}[ur]\ar@{-}[rr]\ar@{..}[dd] \ar@{-}[rd]
 & & \txt{[101]\\$\top_\alpha^\downarrow$\\\PPi} \ar@{..}[ur]\ar@{..}[dd] \ar@{-}[rd]
\\
 & \txt{[010]\\$\bot_\alpha^\downarrow$\\$\botOneZero$} \ar@{-}'[r][rr]
   & & \txt{[110]\\$\top_\alpha^\uparrow$\\\PP}
\\
 \txt{[0{0}0]\\$\bot_\ORccFive$\\\DR} \ar@{-}[rr]\ar@{..}[ur] \ar@{-}[ruuu]
 & & \txt{[1{0}0]\\$\bot_\alpha$\\\PO} \ar@{..}[ur] \ar@{-}[ruuu]
}
$$\arraycolsep=1pt
\begin{array}{rl}
   & \alpha  = \emptyset\\
	\multicolumn{2}{c}{\text{the}\ \times\text{partition}\ \{\PL_{\emptyset\, 1\, 2}^{-\overline{\b\c}}\}}\\
\\
\PL_{\times\alpha}^\T & = \PL_\j^\a \vartriangle \PL_\k^\a\\
\PL_{\times\alpha}^\F & = \PL_\j^\a \vartriangle \PL_\k^{\sim\a}\\\end{array}
$$
\end{minipage}
\begin{minipage}{5cm}
\footnotesize
\xymatrixcolsep{2.7pc}
\xymatrixrowsep{2.7pc}
\xymatrix@!0{
 & \txt{[{0}11]\\$\top_\alpha^\downarrow$\\$\botOneOne$} \ar@{..}[rr]\ar@{..}'[d][dd]
   & & \txt{[{1}11]\\$\top_\ORccFive$\\\EQ} \ar@{..}[dd]
\\
\txt{[{0}01]\\$\bot_\alpha^\uparrow$\\$\botZeroOne$} \ar@{-}[ur]\ar@{..}[rr]\ar@{..}[dd]
 & & \txt{[{1}01]\\$\top_\alpha$\\\PPi} \ar@{-}[ur]\ar@{..}[dd]
\\
 & \txt{[{0}10]\\$\bot_\alpha$\\$\botOneZero$} \ar@{..}'[r][rr] \ar@{-}[ru]
   & & \txt{[{1}10]\\$\top_\alpha^\uparrow$\\\PP}  \ar@{-}[uull]
\\
 \txt{[{0}00]\\$\bot_\ORccFive$\\\DR} \ar@{..}[rr]\ar@{-}[ur]\ar@{-}[uurr]
 & & \txt{[{1}00]\\$\bot_\alpha^\downarrow$\\\PO} \ar@{-}[ur]  \ar@{-}[lluu]
}
$$\arraycolsep=1pt
\begin{array}{rl}
\alpha & = 1\\
\multicolumn{2}{c}{\text{the}\ \times\text{partition}\ \{\PL_{\emptyset\, 1\, 2}^{\overline{\a}-\overline{\c}}\}}\\\\

 \PL_{\times\alpha}^\T = &   \{\bot_\alpha^\downarrow,\top_\alpha^\uparrow,\top_\alpha^\downarrow,\bot_\alpha^\uparrow\}\\
\PL_{\times\alpha}^\F = &  \{\bot_\ORccFive,\top_\alpha,\top_\ORccFive,\bot_\alpha\} \\

\end{array} 
$$
\end{minipage}
\begin{minipage}{5cm}
\footnotesize
\xymatrixcolsep{2.7pc}
\xymatrixrowsep{2.7pc}
\xymatrix@!0{
 & \txt{[01{1}]\\$\top_\alpha^\downarrow$\\$\botOneOne$} \ar@{..}[rr]\ar@{-}'[d][dd]
   & & \txt{[11{1}]\\$\top_\ORccFive$\\\EQ} \ar@{-}[dd]^{\overrightarrow{z}}
\\
\txt{[001]\\$\bot_\alpha$\\$\botZeroOne$} \ar@{..}[ur]\ar@{..}[rr]\ar@{-}[dd] \ar@{-}[urrr]
 & & \txt{[101]\\$\top_\alpha^\uparrow$\\\PPi} \ar@{..}[ur]\ar@{-}[dd]\ar@{-}[lu]\\
 & \txt{[010]\\$\bot_\alpha^\uparrow$\\$\botOneZero$} \ar@{..}'[r][rr]
   & & \txt{[110]\\$\top_\alpha$\\\PP}
\\
 \txt{[00{0}]\\$\bot_\ORccFive$\\\DR} \ar@{..}[rr]_{\overrightarrow{x}} \ar@{..}[ur] \ar@{-}[urrr]
 & & \txt{[10{0}]\\$\bot_\alpha^\downarrow$\\\PO} \ar@{..}[ur]_{\overrightarrow{y}} \ar@{-}[ul]
}
$$\arraycolsep=1pt
\begin{array}{rl}
& \alpha  = 2\\
\multicolumn{2}{c}{\text{the}\ \times\text{partition}\ \{\PL_{\emptyset\, 1\, 2}^{\overline{\a\b}-}\}}\\\\
	 \PL_{\times\alpha}^\F = & \{ \{a_\alpha,a_j,a_k\}\in\ORccFive \mid a_\j = a_\k\} \\
	 \PL_{\times\alpha}^\T = & \{ \{a_\alpha,a_j,a_k\}\in\ORccFive \mid a_\j \neq a_\k\} \\

\end{array} 
$$
\end{minipage}
\end{minipage}

\caption{\label{fig:PairsOfPlanes}Planes of the \ThreeBit\ cube (displayed in the $\emptyset12$ base) and the corresponding 1-bit information states  (\istates) that can obtain for elementary systems $(r_1,r_2)\in\OReg^2$. Each figure shows a partition of the \ThreeBit\ cube. The three partitions $\{\PL_{\emptyset\,1\,2}^{\a--}\}$, $\{\PL_{\emptyset\,1\,2}^{-\b-}\}$ and $\{\PL_{\,\emptyset\,1\,2}^{--\c}\}$ form the family $\PL_\IdxOneBitPR$ \cite{bittner:InformationMereologyAndVagueness}. The three partitions $\{\PL_{\emptyset12}^{\overline{\a}-\overline{\c}}\}$, $\{\PL_{\emptyset12}^{-\overline{\b\c}}\}$ and $\{\PL_{\emptyset12}^{\overline{\a\b}-}\}$ form the family $\PL_\IdxOneBitET$ \cite{Bittner:VMQI}.  See also Tab. \ref{tab:VectorsInIndexNotation} for details regarding  the notation.}
\end{figure}

\subsection{The \ThreeBit Lattice}

In the cube-shaped graph in Fig. \ref{fig:rccFive} the set of vertices is formed by ordered triples whose slots are occupied by members of an indexed family, indicated as $\{(a_\emptyset, a_1,a_2) \mid a_\emptyset, a_1,a_2\in\TF \}=\ORccFive_{\emptyset12} = \TF_\emptyset \times \TF_1 \times \TF_2$.  For every ordered tuple $(a_\emptyset,a_1,a_2) \in \ORccFive_{\emptyset12}$, the set $\langle(a_\emptyset,a_1,a_2)\rangle = \{a_\emptyset,a_1,a_2\}$ is the permutation-invariant representation of the ordered tuple $(a_\emptyset,a_1,a_2)$ such that $\langle(a_\emptyset,a_1,a_2)\rangle = \langle(a_1,a_\emptyset,a_2)\rangle$. 

A partial ordering relation $\le$ on ordered tuples $a,b\in\ORccFive_{\emptyset12}$ can be introduced by setting that $a \le b $ if and only if $(\pr_j\ a) \le (\pr_j\ b)$ for $1 \le j \le 3$, where if $(1,2,3)$ is an ordered triple, then $\pr_j (1,2,3) = j$ for $1\le j\le3$ is the projection of the ordered triple on its $j$th slot. The created lattice $(\ORccFive_{\emptyset12},\le)$ is called the \ThreeBit lattice with the base $\emptyset12$. Its graph identifies a unique direction for each of the edges of the cube-shaped graph shown in Fig. \ref{fig:rccFive}. In general ordered tuples in $\ORccFive_{ijk}$ form lattices $(\ORccFive_{ijk},\le)$ for any given permutation $ijk$ of the index set $\{i,j,k\}=\IdxSet=\{\emptyset,1,2\}$.

All lattices $(\ORccFive_{\textit{ijk}},\le)$ with $\{i,j,k\}=\IdxSet$ have $\DR$ as their sole least element ($\langle\bot_{\textit{ijk}}\rangle=\DR$) and $\EQ$ as their sole greatest element ($\langle\top_{\textit{ijk}}\rangle=\EQ$). Relations of the form  $\ll_\alpha$ can be used to identify the immediate $\alpha$-successors of the bottom element and the immediate $\alpha$-predecessors of the top element of all lattices $(\ORccFive_{\textit{ijk}},\le)$ with  $\alpha\in\{i,j,k\}=\IdxSet$ as shown in the top of Fig. \ref{fig:isuccBot_iprecTop}, where the expression $(\min\ (\IdxSet \setminus \{\alpha\}))$ designates the smallest index in $\IdxSet$ that does not include the element $\alpha$ in the context of the order $\emptyset < 1 < 2$.\footnote{Expressions of the form $(\textsf{THE}\ x.\ \Phi\ x)$ is Russell's definite description operator which designates the unique $x$ such that $(\phi\ x)$ is true. If no such unique $x$ exists then the operator \textsf{THE}\ returns an arbitrary member of the domain of $x$ \cite{Paulson1994Isabelle}.} The descriptors $\bot_\alpha^{\downarrow}, \bot_\alpha^{\uparrow}$, and $\bot_\alpha$ as well as $\top_\alpha^{\downarrow}, \top_\alpha^{\uparrow}$, and $\top_\alpha$ of Fig. \ref{fig:isuccBot_iprecTop} (top) always pick out the three $\alpha$-successors of $\DR$ and the three $\alpha$-predecessors of $\EQ$ in a way determined by $\alpha$ only. As shown in the left diagram in the middle of the figure, it is possible to describe the lattice structure  $(\ORccFive_{\textit{ijk}},\le)$  with $\{i,j,k\}=\IdxSet$ using expressions of the form $\bot_\ORccFive$, $\top_\ORccFive$, $\bot_\alpha^{\downarrow}, \bot_\alpha^{\uparrow}$,  $\bot_\alpha$, $\top_\alpha^{\downarrow}, \top_\alpha^{\uparrow}$, and $\top_\alpha$.

\begin{figure}
    $$\begin{array}{rclcrcl}
    \bot_\ORccFive & = & \DR & \not\approxeq_\alpha & \top_\ORccFive & = & \EQ \\  
       \bot_\alpha^{\downarrow} &=& \textsf{THE}\ r \in \ORccFive.\ \DR \ll_{\min \IdxSet\setminus \{\alpha\}}\ r\ & 
       \ll_\alpha &  \top_\alpha^{\uparrow} &=& \textsf{THE}\ r\in \ORccFive.\ r \ll_{\max \IdxSet\setminus \{\alpha\}} \EQ  \\
\bot_\alpha^{\uparrow} &=& \textsf{THE}\  r \in \ORccFive.\ \DR \ll_{\max \IdxSet\setminus \{\alpha\}}\ r\ & \ll_\alpha &
\top_\alpha^{\downarrow} &=& \textsf{THE}\  r\in \ORccFive. \ r \ll_{\min \IdxSet\setminus \{\alpha\}} \EQ \\
\bot_\alpha &=& \textsf{THE}\ r \in \ORccFive.\ \DR \ll_\alpha\ r \ & \not\approxeq_\alpha &
\top_\alpha &=& \textsf{THE}\ r \in \ORccFive.\ r \ll_\alpha\ \EQ \\
\\
    \end{array}$$
\begin{minipage}{15cm}
\begin{minipage}{7.5cm}
$$
\xymatrixcolsep{1.2pc}
\xymatrixrowsep{1.2pc}
\xymatrix{
& \top_\alpha \ar@{..>}@/^1pc/[rrr]_{<} & & & \top_\ORccFive \\
\bot_\alpha^\uparrow \ar[ru]^{<} \ar@{..>}@/_1pc/[rrr]_(.8){<} &  \approxeq_\alpha & \bot_\alpha^\downarrow \ar[lu]_{<} \ar@{..>}@/^1pc/[rrr]^(.15){<} &   \top_\alpha^\downarrow \ar[ru]^{<} & \approxeq_\alpha & \top_\alpha^\uparrow \ar[lu]_{<} \\
& \bot_\ORccFive \ar[lu]^{<}\ar[ru]_{<} \ar@{..>}@/_1pc/[rrr]^{<} & & & \bot_\alpha \ar[lu]^{<}\ar[ru]_{<} \\
& \PL_\alpha^\F &&& \PL_\alpha^\T  \\
}
$$
\end{minipage}\hfill
\begin{minipage}{7.5cm}
$$
\xymatrixcolsep{1.2pc}
\xymatrixrowsep{1.2pc}
\xymatrix{
& \top_\ORccFive  & & & \top_\alpha^\downarrow \\
\top_\alpha \ar[ru]|{\ll_\alpha} \ar@{..>}[rr]|(.3){\dia_\alpha} &  & \bot_\alpha \ar@{--}[lu]_{\approxeq_\alpha}  &   \bot_\alpha^\uparrow \ar[ru]|{\ll_\alpha} \ar@{..>}[rr]|(.3){\dia_\alpha} &  & \top_\alpha^\uparrow \ar@{--}[lu]_{\approxeq_\alpha}  \\
&  \bot_\ORccFive \ar@{..}[uu]|(.3){\dia_\alpha} \ar@{--}[lu]^{\approxeq_\alpha} \ar[ru]|{\ll_\alpha} & & & \bot_\alpha^\downarrow\ar@{--}[lu]^{\approxeq_\alpha} \ar[ru]|{\ll_\alpha} \ar@{..>}[uu]|(.3){\dia_\alpha} \\
& \PL_{\times\alpha}^\F & & & \PL_{\times\alpha}^\T \\
}$$
\end{minipage}
\end{minipage}
\begin{minipage}{15cm}
\begin{minipage}{3.5cm}
		$$\begin{array}{c}
\\
		\footnotesize
		\xymatrixcolsep{2pc}
		\xymatrixrowsep{2pc}
		\xymatrix@!0{
			& \top_\alpha \ar@{..}[rr]\ar@{..}'[d][dd]
			& & \top_\ORccFive \ar@{..}[dd]
			\\
			\bot_\alpha^\uparrow \ar@{..}[ur]\ar@{..}[rr]\ar@{..}[dd]  
			& & \top_\alpha^\downarrow \ar@{..}[ur]\ar@{..}[dd]
			\\
			& \bot_\alpha^\downarrow \ar@{..}'[r][rr]  
			& & \top_\alpha^\uparrow
			\\
			\bot_\ORccFive \ar@{..}[rr]\ar@{..}[ur] \ar@{-}[rrruuu]
			& & \bot_\alpha \ar@{..}[ur]  
		}\\\\
		 \{\bot_\ORccFive, \top_\ORccFive\}\\
		\end{array}$$
		\end{minipage}
\begin{minipage}{3.5cm}
		$$\begin{array}{c}
\\
		\footnotesize
		\xymatrixcolsep{2pc}
		\xymatrixrowsep{2pc}
		\xymatrix@!0{
			& \top_\alpha \ar@{..}[rr]\ar@{..}'[d][dd]
			& & \top_\ORccFive \ar@{..}[dd]
			\\
			\bot_\alpha^\uparrow \ar@{..}[ur]\ar@{..}[rr]\ar@{..}[dd]  
			& & \top_\alpha^\downarrow \ar@{..}[ur]\ar@{..}[dd]
			\\
			& \bot_\alpha^\downarrow \ar@{..}'[r][rr]  
			& & \top_\alpha^\uparrow
			\\
			\bot_\ORccFive \ar@{..}[rr]\ar@{..}[ur] 
			& & \bot_\alpha \ar@{..}[ur]  \ar@{-}[luuu]
		}\\\\
		 \{\top_\alpha, \bot_\alpha\}\\
		\end{array}$$
		\end{minipage}
\begin{minipage}{3.5cm}
		$$\begin{array}{c}
\\
		\footnotesize
		\xymatrixcolsep{2pc}
		\xymatrixrowsep{2pc}
		\xymatrix@!0{
			& \top_\alpha \ar@{..}[rr]\ar@{..}'[d][dd]
			& & \top_\ORccFive \ar@{..}[dd]
			\\
			\bot_\alpha^\uparrow \ar@{..}[ur]\ar@{..}[rr]\ar@{..}[dd]  \ar@{-}[rrrd]
			& & \top_\alpha^\downarrow \ar@{..}[ur]\ar@{..}[dd]
			\\
			& \bot_\alpha^\downarrow \ar@{..}'[r][rr]  
			& & \top_\alpha^\uparrow
			\\
			\bot_\ORccFive \ar@{..}[rr]\ar@{..}[ur] 
			& & \bot_\alpha \ar@{..}[ur] 
		}\\\\ 
		\{\bot_\alpha^\uparrow, \top_\alpha^\uparrow\}\\
		\end{array}$$
		\end{minipage}
\begin{minipage}{3.5cm}
		$$\begin{array}{c}
\\
		\footnotesize
		\xymatrixcolsep{2pc}
		\xymatrixrowsep{2pc}
		\xymatrix@!0{
			& \top_\alpha \ar@{..}[rr]\ar@{..}'[d][dd]
			& & \top_\ORccFive \ar@{..}[dd]
			\\
			\bot_\alpha^\uparrow \ar@{..}[ur]\ar@{..}[rr]\ar@{..}[dd]  
			& & \top_\alpha^\downarrow \ar@{..}[ur]\ar@{..}[dd]
			\\
			& \bot_\alpha^\downarrow \ar@{..}'[r][rr]  \ar@{-}[ru]
			& & \top_\alpha^\uparrow
			\\
			\bot_\ORccFive \ar@{..}[rr]\ar@{..}[ur] 
			& & \bot_\alpha \ar@{..}[ur] 
		}\\ \\ 
		\{\bot_\alpha^\downarrow,\top_\alpha^\downarrow\}\\
		\end{array}$$
		\end{minipage}
$$
\begin{array}{l}
\DIA_\alpha=\{\{\bot_\ORccFive, \top_\ORccFive\}, \{\top_\alpha, \bot_\alpha\}, \{\bot_\alpha^\uparrow, \top_\alpha^\uparrow\}, \{\bot_\alpha^\downarrow,\top_\alpha^\downarrow\}\} = \{ \{d,(\dia_\alpha\, d)\}.\   d\in\PL_\alpha^\F \} \\
\PL_{\i\, \hat{\j}\, \k}^{\overline{\a\b\c}} \quad = \quad \PL_{\times\k}^{\a\xor\b}\cap\PL_{\times\i}^{\b\xor\c} \quad = \quad  (\PL_\i^\a\vartriangle\PL_\j^{\sim\b})\cap(\PL_\j^{\b}\vartriangle\PL_\k^{\sim\c})\quad \in\quad \DIA_j\\
\end{array}
$$
\end{minipage}
\caption{\label{fig:isuccBot_iprecTop}(top) The \ThreeBit lattice in terms of the immediate $\alpha$-successors and $\alpha$-predecessors of the bottom and top elements with $\alpha\in\IdxSet = \{\emptyset,1,2\}$. (middle left) Faces of the \ThreeBit cube and their relation to the graph of the \ThreeBit lattice. If the graph of the relation $<$ coincides with an edge of a face then this is represented by a solid arrow.  If the graph of the relation $<$ runs diagonally across a face then this is represented by a dotted arrow.  (middle right) Interior planes of the \ThreeBit cube and relations $\ll_\alpha$ (solid arrows) $\approxeq_\alpha$ (dashed lines; $\approxeq_\alpha$ stands for $\cong_\alpha$ throughout the figure) and the mapping $\dia$ (dotted arrows) of the \ThreeBit lattice. (bottom) Diagonals across the interior of the \ThreeBit cube (displayed in the $\emptyset12$-base for $\alpha=\emptyset$ and collected in the set $\DIA_\alpha$) and their representation in terms of the intersection of interior planes.}
\end{figure}

\begin{example}[Grounding the Lattice Descriptors]
\normalfont
\label{ex:LatticeDescriptors}
The abstract equivalence relations and lattice descriptors from this section can be grounded in the concrete mereological relations from Fig.~\ref{fig:rccFive} and Tab.~\ref{tab:rccFive}. Let the index be $\alpha = \emptyset$, corresponding to the fundamental question of overlap. The equivalence relation $\cong_\emptyset$ groups all relations that give the same answer to this question. For instance, the set $\{\PO, \PP, \PPi, \EQ\}$ are all $\cong_\emptyset$-equivalent because for each, the $\emptyset$-indexed value is $\T$. Likewise, $\{\DR, \botOneZero, \botOneOne, \botZeroOne\}$ are $\cong_\emptyset$-equivalent because their $\emptyset$-indexed value is $\F$. The descriptors for the immediate successors of the bottom element ($\bot_\ORccFive = \DR$) and predecessors of the top element ($\top_\ORccFive = \EQ$) also map to specific relations. For $\alpha = \emptyset$, as shown in Figure~\ref{fig:PairsOfPlanes} (top-left): The immediate $\emptyset$-successor of \DR\ is $\bot_\emptyset = \PO$. The other successors are $\bot_\emptyset^\downarrow = \botOneZero$ and $\bot_\emptyset^\uparrow = \botZeroOne$.  The immediate $\emptyset$-predecessor of \EQ\ is $\top_\emptyset = \botOneOne$. The other predecessors are $\top_\emptyset^\uparrow = \PP$ and $\top_\emptyset^\downarrow = \PPi$.
These mappings provide a concrete interpretation for the abstract lattice structure. \qed
\end{example}

\section{The \ThreeBit Partition Lattice}
\label{sec:PartitionLattice}

The previous sections established a geometric picture of mereological information, with crisp relations as vertices of a 3-bit cube. This section now formalizes the higher-dimensional features of that cube, namely its edges, faces, and internal planes, which represent states of vagueness. This is a crucial step in developing the \textit{Quantized Information Semantics (QIS) Principle}, as this section builds the formal state space of the theory. The elements of this space, called \textit{information states} ($\istates$), are the targets of both the top-down semantic path and the bottom-up quantization path.

Following the work of \citet{Ellerman2022}, these $\istates$ are not represented as vectors in a Hilbert space but as cells within partitions of the set of all eight classical outcomes, $\ORccFive$. These partitions and their cells form a non-distributive lattice that serves as the foundation for the logic of vague mereological relations. To begin,  the formal definition of a partition is provided.

Consider a set $U = \{u_1 , \ldots , u_n \}$. A partition $\pi = \{B_1 , \ldots, B_m \}$ consists of subsets $B_j \subseteq U$ (where $j = 1, ..., m$) that do not overlap and together cover $U$. Each subset $B_i$ in the partition $\pi$ is referred to as a cell. A partial order on partitions (refinement ordering) is defined as follows: Given a set $U = \{u_1 , \ldots , u_n \}$, a partition $\pi = \{B_1 , \ldots, B_m \}$ on $U$ and another partition $\sigma = \{C_1 , \ldots, C_{m^\prime} \}$ on $U$, then $\sigma$ is refined by $\pi$, written, $\sigma \lesssim \pi$, if for each cell of $\pi$, there is a cell of $\sigma$ that contains it:
\begin{equation}
\label{eq:lesssim_def}
    \begin{array}{l}
    \sigma \lesssim \pi  \equiv  \forall\ B_j \in \pi.\ \exists C_{j^\prime} \in \sigma.\ B_j \subseteq  C_{j^\prime}; \quad 
1_U  =  \{\{u_1 \} , \ldots, \{u_n \}\}; \quad 0_U = \{U \}. \\
    \end{array}
\end{equation}
\noindent For every two partitions $B,C$ of $\ORccFive$ the least upper bound in the refinement ordering is their common refinement, the partition whose cells are all the non-empty intersections $B_j \cap C_{j^\prime}$, and the greatest lower bound is their common coarsening, the partition obtained by merging overlapping cells until the cells are disjoint. These definitions follow the presentation of \citet{Ellerman2022}. The fifteen partitions selected below (Eq.~\ref{eq:PartitionLatticeCarrierSet}) are closed under common refinement but not under common coarsening; the lattice $\Pi(\ORccFive)$ of Fig.~\ref{fig:RccFivePartitionLattice} takes its greatest lower bounds within that family.

\subsection{Partitions of the Set $\ORccFive$ and Their Geometric Interpretation}

Functions of the form $f_i=\{a_\emptyset,a_1,a_2\}\in\ORccFive \mapsto  a_i\in\TF$ with $i\in\IdxSet$ generate partitions of the form $\pi_i=\{ f_i^{-1}(a) \mid a\in\TF\}$. The symbol $\{\PL_\i^\a\}$ denotes the partition $\pi_i$ if and only if $\pi_i$ coincides with the equivalence classes represented by $\ORccFive/_{\cong_i}$. That is, $\{\PL_\i^\a\}=\pi_i$ if $\pi_i = \ORccFive/_{\cong_i}$ and $\{\PL_\i^\a\}$ is undefined otherwise. Partition cells of $\{\PL_\i^\a\}$ are referenced using the notation $\PL_\i^\a$ with $\a\in\TF$. Each cell $\PL_\i^\a$ corresponds to a face, in the non-directed version of the cube-shaped graph of the \ThreeBit lattice, the \ThreeBit cube, as visualized in Fig. \ref{fig:PairsOfPlanes} (top).\footnote{The geometry of the \ThreeBit cube  arises from the graph of the \ThreeBit lattice $(\ORccFive_\textit{ijk},\le)$ (Fig. \ref{fig:rccFive}) because the non-directed version of this graph is a cubical graph and therefore has the symmetries (rotations and mirroring operations) that characterize the geometry of a cube. \cite{Johnson2018}}

Similarly, the functions $f_{ij}=\{a_\emptyset,a_1,a_2\}\in\ORccFive \mapsto  (a_i,a_j)\in\TF^2$ with $i\neq j\in\IdxSet$ generate the partitions $\pi_{ij}=\{ f_{ij}^{-1}(a,b)\}_{a,b\in\TF}$.  The notation $\{\PL_{\,\i\,\j}^{\a\b}\}$ designates partitions $\pi_{ij}$ that correspond to the intersection of the equivalence classes $\ORccFive/_{\cong_i}$ and $\ORccFive/_{\cong_j}$ such that $\{\PL_{\,\i\,\j}^{\a\b}\}=\pi_{ij}$ iff $\pi_{ij} = \{ r_i \cap r_j \mid r_i\in\ORccFive/_{\cong_{i}},r_j\in\ORccFive/_{\cong_{j}}\}$ and $\{\PL_{\,\i\,\j}^{\a\b}\}$ is undefined otherwise. Partition cells are designated using $\PL_{\,\i\,\j}^{\a\b}$ with $\a,\b\in\TF$. Each cell $\PL_{\,\i\,\j}^{\a\b}$ corresponds to an edge of the \ThreeBit cube.

\begin{example}[An Edge on the Information Cube]\normalfont
\label{ex:EdgeState}
An edge represents a state where two fundamental questions have been answered, leaving only one degree of indeterminacy. Consider an elementary system $(\AR^\prime,\FH^\prime)$ with an information capacity of at least two bits. An observer determines that the Acidic Region ($\AR^\prime$) and Fern Habitat ($\FH^\prime$) overlap ($Q_\emptyset=\T$) and that $\AR^\prime$ is not a part of $\FH^\prime$ ($Q_1=\F$). The resulting information state is the set of possibilities consistent with both answers, which is constructed from the intersection of the two corresponding 1-bit states (faces): $\PL_{\emptyset}^{\T} \cap \PL_{1}^{\F}$. This intersection yields the set $\{\PO, \PPi\}$, which corresponds to the edge on the \ThreeBit cube connecting the two vertices. In the full index notation developed in the next subsection, a state with one indeterminate question like this is formally written as $\PL_{\emptyset 1 2}^{\T\F-}$.  More examples are  visualized in Fig. \ref{fig:RccFiveEdges}.  \qed
\end{example}

Finally, the functions $f_{ijk}=\{a_\emptyset,a_1,a_2\}\in\ORccFive \mapsto  (a_i,a_j,a_k)\in\TF^3$ with $\{i,j,k\}=\IdxSet$ generate the partitions $\pi_\textit{ijk}=\{ f_\textit{ijk}^{-1}(a,b,c)\}_{a,b,c\, \in\TF} = \{ r_i \cap r_j \cap r_k \mid r_i\in\ORccFive/_{\cong_{i}},r_j\in\ORccFive/_{\cong_{j}},r_k\in\ORccFive/_{\cong_{k}}\} = \{\PL_{\,\i\,\j\,\k}^{\a\b\c}\}$.  Each cell $\PL_{\i\,\j\,\k}^{\a\b\c}$ corresponds to a vertex of the \ThreeBit cube.

In addition to partitions of the form $\{\PL_\i^\a\}=\{\PL_\i^\F,\PL_\i^\T\}$ in which cells are the parallel faces of the \ThreeBit cube, as shown in Fig. \ref{fig:PairsOfPlanes} (top), there are partitions in which cells are interior planes of the \ThreeBit cube that are pairwise orthogonal ($\times$ planes). These are designated as $\{\PL_{\times \i}^a\}=\{\PL_{\times \i}^\F,\PL_{\times \i}^\T\}$ and can be seen in Fig. \ref{fig:PairsOfPlanes} (bottom). The interior planes in the \ThreeBit cube are obtained from the faces using a disjunctive union. For example, $ \PL_{\times 1}^\F = \PL_\emptyset^\T \vartriangle \PL_2^\F = \PL_\emptyset^\F \vartriangle \PL_2^\T$ and $ \PL_{\times 1}^\T = \PL_\emptyset^\T \vartriangle \PL_2^\T = \PL_\emptyset^\F \vartriangle \PL_2^\F$. Equivalently, the interior planes are subsets of the form $\PL_{\times \i}^\F = \{ \{a_i,a_j,a_k\}\in\ORccFive \mid a_j = a_k\}$ and $\PL_{\times \i}^\T = \{ \{a_i,a_j,a_k\}\in\ORccFive \mid a_j \neq a_k\}$.  This can also be expressed as $\PL^\b_{\times \i} = \{ \{a_i,a_j,a_k\}\in\ORccFive \mid  a_j \xor\ a_k = \b\}$\footnote{In  expressions such as $\PL^\b_{\times \i} = \{ \{a_i,a_j,a_k\}\in\ORccFive \mid  a_j \xor\ a_k = \b\}$   the \textsf{SF}\ font is used for indexes that designate partition cells. The $CMM$ font is used for indexes that designate elements of indexed sets. Within an expression indexes with the same name take the same value, i.e. $\i = i$.}  where $\xor$ is the logical operation of `exclusive or'.

Whereas the 2-bit product states in Figure~\ref{fig:RccFiveEdges} correspond to the cube's edges, other 2-bit entangled states correspond to diagonals across its faces, as visualized in Figure~\ref{fig:RccFiveEdgesX}. They arise from the intersection of a partition in which cells are interior planes with a partition in which cells are faces, i.e. $\{\PL_{\times \i}^a\}\cap \{\PL_\i^\b\}=\{\PL_{\times \i}^\F\cap\PL_\i^\F,\PL_{\times \i}^\F\cap\PL_\i^\T,\PL_{\times \i}^\T\cap\PL_\i^\F,\PL_{\times \i}^\T\cap\PL_\i^\T\}$. Finally, there is a partition in which the cells are the line segments that  diagonally pass through the interior of the \ThreeBit cube (Fig. \ref{fig:isuccBot_iprecTop} (bottom)). They arise from the intersection of two partitions in which cells are interior planes, i.e. $\{\PL_{\times \i}^a\}\cap \{\PL_{\times \j}^\b\}=\{\PL_{\times \i}^\F\cap\PL_{\times \j}^\F,\PL_{\times \i}^\F\cap\PL_{\times \j}^\T,\PL_{\times \i}^\T\cap\PL_{\times \j}^\F,\PL_{\times \i}^\T\cap\PL_{\times \j}^\T\}$.

The line segments that are the diagonals across the interior of the \rccFive cube can also be defined as maps: for every index $i\in\IdxSet$ there is a bijective mapping between vertices in $\PL_\i^\F$ and $\PL_\i^\T$ given by $\dia_i: \PL_\i^\F \to \PL_\i^\T$ with $(\dia_i\ \{a_\emptyset,a_1,a_2\}) = \{\sim a_\emptyset,\sim a_1,\sim a_2\}$. Using this map the set of diagonals is specified as $\DIA_i=\{\{\bot_\ORccFive, \top_\ORccFive\}$, $\{\top_i, \bot_i\}$, $\{\bot_i^\uparrow, \top_i^\uparrow\}$, $\{\bot_i^\downarrow,\top_i^\downarrow\}\}$. The diagonals are embedded in  interior planes of the form $\PL_{\times \i}^\a$ such that  $\{\bot_\ORccFive, \top_\ORccFive\}\cup\{\top_i, \bot_i\} = \PL_{\times \i}^\F$ and  $\{\bot_i^\uparrow, \top_i^\uparrow\}\cup\{\bot_i^\downarrow,\top_i^\downarrow\} = \PL_{\times \i}^\T$ (Fig. \ref{fig:isuccBot_iprecTop} (bottom)).

By construction, partitions are formed by cells that are geometrically homogeneous. That is, all cells of a given partition are faces or interior planes or edges or line segments across faces or line segments across the interior of the \ThreeBit cube or vertices. Cells of the same geometric type are related to the \ThreeBit lattice in the same ways: Cells that, geometrically, are vertices and edges of the \ThreeBit cube emerge immediately from the nodes and directed edges of the \ThreeBit lattice. Cells that, geometrically, are faces are related to the ordering structure of the \ThreeBit lattice through equivalence classes of the form $\ORccFive /_{\cong_\alpha}$ and the diamond-shaped sub-lattices displayed in Fig. \ref{fig:isuccBot_iprecTop} (middle left). Cells that, geometrically, are interior planes are related to the \ThreeBit lattice through the top and bottom elements  and their respective immediate $\alpha$-predecessors and $\alpha$-successors as signified in the  equations in the left of Fig. \ref{fig:PairsOfPlanes}(bottom). Every interior plane of the form $\PL_{\times \i}^\a$ is characterized by a diamond-shaped ordering structure that is related to the \ThreeBit lattice as illustrated in the right diagram in the middle of Fig. \ref{fig:isuccBot_iprecTop}.

\subsection{The Index Notation}

To formally describe the partition cells that serve as the information states ($\istates$) of the framework, a compact and systematic index notation is required. This section introduces this notation, which takes the general form $\PL_{\, \i\ \j\ \k}^{\alpha\beta\gamma}$. This system provides a structured language to precisely represent the full spectrum of mereological information, encoding whether a fundamental proposition is determinate, in a state of superposition, or entangled. Based on its index structure this language syntactically facilitates representing the relations between questions, answers, and the resulting $\istates$.

\subsubsection{An Informal Guide}
\label{sec:IndexNotationGuide}

The index notation for \istates\ ($\PL$), questions ($\QQ$), and answers ($\A$) provides a compact and systematic way to describe the informational structure of a mereological system. This guide explains the notation by building from its core principles to its more complex applications.

\paragraph{The Three-Slot Foundation:} The notation is always built upon a three-slot foundation, represented by the lower indices \i\j\k\ which stand for an arbitrary permutation of the set $\{\emptyset,1,2\}$. These three slots act as an "address space" that corresponds directly to the three fundamental propositions of the system: $Q_\emptyset$ (Overlap), $Q_1$ (Part Of), and $Q_2$ (Inverse Part Of) (cf. Fig.~\ref{fig:rccFive}, Tab.~\ref{tab:rccFive}). The upper indices ($\alpha,\beta,\gamma$) represent the "data" or information content held in these slots.

\paragraph{Complete vs. Abbreviated Notation:} A key distinction exists between the full three-slot notation and its abbreviated forms.
\textit{The Complete Notation ($\PL_{\,\i\,\j\,\k}^{\a\b-}$)} makes an ontological claim about a system with a specific, limited information capacity. The hyphen (`-') explicitly marks a proposition as being in a state of "equal superposition". For example, $\PL_{\,\i\,\j\,\k}^{\a\b-}$ describes the objective state of a 2-bit system where the maximum information has been obtained for the first two slots, and the third slot represents  indeterminacy. \textit{The Abbreviated Notation ($\PL_{\i\,\j}^{\a\b}$)} is a mathematical convenience that represents the result of a set-theoretic operation, typically intersection ($\PL_\i^\a \cap \PL_\j^\b$). This notation is used when the focus is on the algebraic combination of states, not the total capacity of a specific system. While both $\PL_{\,\i\,\j\,\k}^{\a\b-}$ and $\PL_{\i\,\j}^{\a\b}$ are  extensionally identical (referring to the same set of classical relations), the complete notation makes a stronger, more specific claim about the system's objective state and its inherent informational limits.

\paragraph{Permutation Invariance:} A crucial property of the system is that the \istates\ are invariant under any joint permutation of the index pairs (i.e., swapping the \i\ and \j\ slots along with their corresponding upper indices $\alpha$ and $\beta$). This is formally established in Lemma~\ref{lemma:PermutationUpperAndLower}. The important consequence is that the visual adjacency of slots in the notation has no fundamental meaning. A state describing an entanglement between slots \i\ and \j\ is of the same fundamental \textit{type} as a state describing an entanglement between \i\ and \k.

\paragraph{The Building Blocks, the 1-Bit States:} The simplest information-incomplete states are those that carry exactly one bit of information. The index notation describes these states by showing how the three fundamental propositions are structured. It uses two core symbols to represent indeterminacy:

\begin{itemize}
    \item The hyphen $\binom{-}{\k}$ signifies a state of \textit{equal superposition}. A proposition in this state is wholly indeterminate (it could be \T\ or \F\ if made determinate). A hyphen in the upper index of the $\k$ column represents the partition cell $\PL_\k^\T \cup \PL_\k^\F=\ORccFive$ (no possibilities excluded, therefore no information).

    \item Pairs of overlined variables $\binom{\overline{\a}}{\i}\binom{\overline{\b}}{\j}$ written as $_{\,\i\,\j}^{\overline{\a\b}}$ signify \textit{entanglement}. The propositions for the covered slots are individually indeterminate but are linked by a determinate correlation. Pairs of overlined variables correspond to the disjunctive union of the cell identified by the first variable and the complement of the cell identified by the second, i.e. $\PL_{\,\i\,\j}^{\overline{\a\b}} = \PL_\i^\a \vartriangle \PL_\j^{\sim\b}$, the set of patterns whose two answers are $(\a,\b)$ or $(\sim\a,\sim\b)$
\end{itemize}

A 1-bit state arises when one bit of information is determinate while the rest is not. For example, in a \textit{product state} like $\PL_{\,\i\,\j\,\k}^{\a--}= \PL_\i^\a \cap \ORccFive \cap \ORccFive$, one proposition is determinate ($Q_i=\a$), while the other two are in equal superposition. In an \textit{entangled state} like $\PL_{\,\i\,\j\,\k}^{\overline{\a\b}-}=\PL_{\,\i\,\j}^{\overline{\a\b}}\cap\ORccFive$, a single correlation (between $Q_i$ and $Q_j$) is determinate, while the third proposition ($Q_k$) is in equal superposition.

\paragraph{The Unified System of Questions, Answers, and States:}
The notation for questions ($\QQ$) and answers ($\A$) directly mirrors the structure of the state notation ($\PL$) they produce. The link between them is the Quantized Information Semantics (QIS) principle, which states that an answered question corresponds to a specific information state: $\lceil\QQ...?=\A...\rceil = \PL...$ (formalized in Corollary~\ref{corollary:QuantificationEqSupport}).

The overline is the key to this unified system for entanglement. It signals that the covered slots are to be combined via `xor', and the expressive answer notation represents the pattern of the two underlying bits that produced the single `\Y' or `\N' result. The following table illustrates this for a 1-bit entangled state.

\begin{table}[h!]
\caption{Mapping Questions to Answers and States for a 1-bit Entangled System}
\label{tab:QIS_rosetta_stone}
\centering
\def\arraystretch{1.5}
\begin{tabular}{p{3.5cm} p{2cm} p{2.2cm} p{2cm} p{2cm}}
\hline
\textbf{Informal Question} & \textbf{Formal Question ($\QQ$)} & \textbf{Answer \newline (\texttt{xor} result)} & \textbf{Formal Answer Pattern ($\A$)} & \textbf{\istate\ Pattern ($\PL$)} \\ \hline
Is $Q_i$ indeterminate and are the truth values $Q_j$ and $Q_k$ different?   & \quad \newline $\QQ_{-\overline{\j\k}}?$ & \N\ ($\b=\c$)\newline\newline  \Y\ ($\b\neq\c$) & $\A^{-\overline{\b\b}}$\newline\newline $\A^{-\overline{\b\sim\b}}$ & $\PL_{\i\j\k}^{-\overline{\b\b}}$\newline\newline $\PL_{\i\j\k}^{-\overline{\b\sim\b}}$ \\
 \hline
\end{tabular}
\end{table}

This integrated system shows how a specific type of question elicits a structurally similar answer, which in turn determines a state with the exact same notational pattern. This ensures a coherent and consistent mapping from an observer's inquiry to the resulting state of the mereological system.

\paragraph{Building Higher-Order States:}
With the fundamental principles and 1-bit states established, the more complex 2-bit states can be understood as logical combinations of these building blocks. The formal definitions for all 2-bit states can be found in Table~\ref{tab:VectorsInIndexNotation} .

\textit{2-Bit Product States:} A 2-bit product state, such as $\PL_{\,\i\,\j\,\k}^{\a\b-}$, is formed when two propositions are determinate and one remains in equal superposition. This state is constructed by the cumulative combination (intersection) of two 1-bit product states:
$$ \PL_{\,\i\,\j\,\k}^{\a\b-} = \PL_{\,\i\,\j\,\k}^{\a--} \cap \PL_{\,\i\,\j\,\k}^{-\b-} \quad \text{and similarly for}\quad  \PL_{\,\i\,\j\,\k}^{\a-\c} = \PL_{\,\i\,\j\,\k}^{\a--} \cap \PL_{\,\i\,\j\,\k}^{--\c} $$
Geometrically, this corresponds to an edge of the information cube, as illustrated in Example~\ref{ex:EdgeState} .

\textit{2-Bit Entangled States:} A 2-bit entangled state combines determinate and correlational information. For instance, the state $\PL_{\i\j\k}^{\a\overline{\b\c}}$ describes a system where one proposition ($Q_i=\a$) is determinate, while the other two ($Q_j$ and $Q_k$) are entangled. This state is constructed by the intersection of a 1-bit product state and a 1-bit entangled state:
$$ \PL_{\,\i\,\j\,\k}^{\a\overline{\b\c}} = \PL_{\,\i\,\j\,\k}^{\a--} \cap \PL_{\,\i\,\j\,\k}^{-\overline{\b\c}}\quad \text{and similarly for}\quad  \PL_{\,\i\,\j\,\k}^{\overline{\a}\b\overline{\c}} = \PL_{\,\i\,\j\,\k}^{-\b-} \cap \PL_{\,\i\,\j\,\k}^{\overline{\a}-\overline{\c}}$$
This demonstrates the compositional power of the notation, allowing for the precise description of complex mixtures of classical and quantum-like information. Finally, maximally entangled states are constructed by the intersection of two 1-bit entangled states:
$$\PL_{\,\i\,\j\,\k}^{\overline{\a\b\c}} = \PL_{\,\i\,\j\,\k}^{\overline{\a\b}-} \cap \PL_{\,\i\,\j\,\k}^{-\overline{\b\c}}.$$
\noindent The state with zero bits of information is indicated by $\PL_{\, \i\ \j\ \k}^{---} = \ORccFive$. 

The framework extensively uses this index-based notation. The detailed definitions are collected in Table~\ref{tab:VectorsInIndexNotation}.

\begin{example}[A 1-bit Product State.]\normalfont
To illustrate, consider the Soil Chemist from the running example. Their direct sampling method establishes the information state $\PL_{\emptyset 1 2}^{\T--}$. This notation signifies that the first proposition ($Q_\emptyset$, for overlap) has a determinate value of True (\T), while the other two propositions remain fully indeterminate (indicated by the dashes `--`). Formally, this \istate\ is the set of four unresolved classical possibilities $\{\PO, \PP, \PPi, \EQ\}$. It represents a state of \textit{orthogonal vagueness}, where the question of overlap is settled and no information about the individual parthood questions is obtained. Geometrically, it corresponds to one of the faces of the 3-bit information cube (Figure \ref{fig:PairsOfPlanes}, top left). \qed
\end{example}

\begin{example}[A 1-bit Entangled State.]\normalfont
This structure is illustrated by a variant of the Remote Sensing Specialist's scenario in which the correlation concerns $Q_\emptyset$ and $Q_1$. The correlational analysis establishes the information state $\PL_{\emptyset 1 2}^{\overline{\a\a}-}$. The overbar notation signifies entanglement: while the individual values of $Q_\emptyset$ and $Q_1$ are indeterminate, their relationship (in this case, their equality) is fixed. This represents a state of \textit{non-orthogonal vagueness}. Formally, this \istate\ corresponds to the set $\{\DR, \botZeroOne, \PP, \EQ\}$. Geometrically, it does not correspond to a face of the cube, but to a diagonal plane slicing through its interior. (Figure \ref{fig:PairsOfPlanes}, bottom right) \qed
\end{example}

\begin{example}[A 2-bit Maximally Entangled State]\normalfont
\label{ex:MaxEntangledState}
A maximally entangled state contains two bits of information, but in a purely correlational, non-orthogonal form. Consider the state $\PL_{\emptyset \hat{1} 2}^{\overline{\T\T\T}}$. This notation represents the intersection of two non-orthogonal \istates: $(\PL_\emptyset^\T \vartriangle \PL_1^{\sim\T})$ and $(\PL_1^\T \vartriangle \PL_2^{\sim\T})$. The resulting state is the set $\{\DR, \EQ\}$ (cf. left-most column of Fig.~\ref{fig:RccFiveEdgesXX}). 

Mereologically, this is a highly non-intuitive state that asserts a definite relationship between the three fundamental propositions ($Q_\emptyset \oplus Q_1 = \F$ and $Q_1 \oplus Q_2 = \F$). Here, the $\oplus$ symbol signifies the logical operation \textit{exclusive or}, so the value $\F$ means that the truth values in each pair must be the same. This correlation forces the outcome to be one of two specific, non-adjacent vertices on the cube. Neither of the three fundamental questions can be answered determinately on its own, yet the system is precisely constrained to be either 'Disjoint' or 'Equal', highlighting the unique nature of non-orthogonal vagueness. Geometrically, this state corresponds to one of the four main diagonals passing through the interior of the cube. (Fig.~\ref{fig:RccFiveEdgesXX}) \qed
\end{example}

\begin{table}
\caption{\label{tab:VectorsInIndexNotation}Sub-spaces of the \ThreeBit cube / partition cells of the \ThreeBit partition lattice written in index notations of the form $\PL_{\, \i\ \j\ \k}^{\alpha\beta\gamma}$ where $\alpha\in\{\a,\overline{\a},-\}$, $\beta\in\{\b,\overline{\b},-\}$ and $\gamma\in\{\c,\overline{\c},-\}$ for a permutation \i\j\k\ of the set $\{\i,\j,\k\}=\IdxSet$ (adapted from \cite{Bittner:VMQI}). Expressions in index notation are characterized by index pattern and index pattern fall into seven groups. The sequence of operations that corresponds to an index pattern is signified by $\boxplus$. }
   $$\arraycolsep=0.1pc\begin{array}{c|c|ccccc|c|l|l|c|c}  
 \multicolumn{2}{c|}{\text{index}} &  \multicolumn{5}{c|}{\text{operator-based definition of}} & \text{op-} & \multicolumn{2}{c}{\text{examples of}}&\multicolumn{2}{c}{}\\
\multicolumn{2}{c|}{\text{pattern}} &   \multicolumn{5}{c|}{\text{index notation for partition cells}} &\text{pattern} & \multicolumn{2}{c|}{\text{partition cells (left) and partitions}} & \multicolumn{2}{c}{\text{bits of}} \\
 \iform &  \alpha\beta\gamma & \multicolumn{5}{c|}{\PL_{\, \i\ \j\ \k}^{\alpha\beta\gamma}\ = } & \boxplus & \multicolumn{2}{c|}{\text{(right) in index notation}\ (\i\j\k = \emptyset12)} & \text{info} & \txt{vague\\ness}\\\hline
 \Lambda_\IdxZeroBit &   --- & (\PL_\i^\T \cup \PL_\i^\F) & \cap & (\PL_\j^\T \cup \PL_\j^\F) & \cap & (\PL_\k^\T \cup \PL_\k^\F) & \scriptstyle\cup\cap\cup\cap\cup & \PL_{\, \emptyset\, 1\, 2}^{---}=\ORccFive & \{\PL_{\, \emptyset\, 1\, 2}^{---}\} & 0 & 3 \\  \hline
 &   \a-- &  \PL_\i^\a & \cap & (\PL_\j^\T \cup \PL_\j^\F) & \cap &  (\PL_\k^\T \cup \PL_\k^\F) & \scriptstyle\cap\cup\cap\cup & \PL_{\, \emptyset\, 1\, 2}^{\T--}=\{\PO,\PP,\PPi,\EQ \} & \{ \PL_{\, \emptyset\, 1\, 2}^{\a--}\} & 1 & 2\\
 \Lambda_\IdxOneBitPR &   -\b- &  (\PL_\i^\T \cup \PL_\i^\F) & \cap & \PL_\j^\b & \cap & (\PL_\k^\T \cup \PL_\k^\F) & \scriptstyle\cup\cap\cap\cup & \PL_{\, \emptyset\, 1\, 2}^{-\F-}  = \{\PO,\PPi,\botZeroOne,\DR\} &  \{\PL_{\, \emptyset\, 1\, 2}^{-\b-}\} & 1 & 2 \\
&    --\c &  (\PL_\i^\T \cup \PL_\i^\F) & \cap & (\PL_\j^\T \cup \PL_\j^\F) & \cap & \PL_\k^\c & \scriptstyle\cup\cap\cup\cap &  \PL_{\, \emptyset\, 1\, 2}^{--\F}= \{\PP,\PO,\botOneZero,\DR\} &  \{\PL_{\, \emptyset\, 1\, 2}^{--\c}\} & 1 & 2 \\\hline
&    \a\b- &  \PL_\i^\a & \cap & \PL_\j^\b & \cap & (\PL_\k^\T \cup \PL_\k^\F) & \scriptstyle\cap\cap\cup & \PL_{\, \emptyset 1 2}^{\T\F-}=\{\PO,\PPi\} &  \{\PL_{ \emptyset 1\, 2}^{\a\b-}\} & 2 & 1\\
\Lambda_\IdxTwoBitPR  &  \a-\c &  \PL_\i^\a & \cap & (\PL_\j^\T \cup \PL_\j^\F) & \cap & \PL_\k^\c  & \scriptstyle\cap\cup\cap &  \PL_{\emptyset\ 12}^{\T-\F}=\{\PO,\PP\} &  \{\PL_{ \emptyset\, 1 2}^{\a-\c}\} & 2 & 1 \\
 &   -\b\c &  (\PL_\i^\T \cup \PL_\i^\F)  & \cap & \PL_\j^\b & \cap & \PL_\k^\c & \scriptstyle\cup\cap\cap &  \PL_{\, \emptyset 1 2}^{-\F\F}=\{\PO,\DR\} &  \{\PL_{\, \emptyset 1 2}^{-\b\c}\} & 2 & 1 \\\hline
   \Lambda_\IdxThreeBitPR &   \a\b\c &  \PL_\i^\a & \cap & \PL_\j^\b & \cap & \PL_\k^\c & \scriptstyle\cap\cap &  \PL_{\, \emptyset 12}^{\T\F\F}=\{\PO\} &  \{\PL_{\emptyset12}^{\a\b\c}\} & 3 & 0 \\\hline
  &  \overline{\a\b}- &  (\PL_\i^\a & \vartriangle & \PL_\j^{\sim\b}) & \cap & (\PL_\k^\T \cup \PL_\k^\F) & \scriptstyle\vartriangle\cap\cup &  \PL_{\emptyset 1\,2}^{\overline{\T\F}-} =\{\PO,\PPi,\botOneZero,\botOneOne\} & \{\PL_{\emptyset 1\,2}^{\overline{\a\b}-}\} & 1 & 2 \\
 \Lambda_\IdxOneBitET &   \overline{\a}-\overline{\c} &   (\PL_\i^\a & \vartriangle & \PL_\k^{\sim\c}) & \cap & (\PL_\j^\T \cup \PL_\j^\F) & \scriptstyle\overline{\vartriangle\cap\cup}  & \PL_{\emptyset\, 1\,2}^{\overline{\T}-\overline{\F}} = \{\PO,\PP,\botZeroOne,\botOneOne\} & \{\PL_{\emptyset\, 1\,2}^{\overline{\a}-\overline{\c}}\} & 1 & 2 \\
    
  &  -\overline{\b\c} &  (\PL_\i^\T \cup \PL_\i^\F) & \cap & (\PL_\j^\b & \vartriangle & \PL_\k^{\sim\c}) & \scriptstyle\cup\cap\vartriangle &  \PL_{\emptyset\, 12}^{-\overline{\T\F}} = \{\PP,\PPi,\botZeroOne,\botOneZero\} & \{\PL_{\emptyset\, 12}^{-\overline{\b\c}}\} & 1 & 2 \\\hline
   &  \overline{\a\b}\c &  (\PL_\i^\a & \vartriangle & \PL_\j^{\sim\b}) & \cap & \PL_\k^\c & \scriptstyle\vartriangle\cap &  \PL_{\emptyset 12}^{\overline{\T\F}\T} = \{\PPi,\botOneOne\} & \{\PL_{\emptyset 12}^{\overline{\a\b}\c}\} & 2 & 1 \\
 \Lambda_\IdxTwoBitET  & \overline{\a}\b\overline{\c} &  (\PL_\i^\a & \vartriangle & \PL_\k^{\sim\c}) & \cap & \PL_\j^\b & \scriptstyle\overline{\vartriangle\cap} &  \PL_{\emptyset 12}^{\overline{\T}\F\overline{\T}} = \{\DR,\PPi\} & \{\PL_{\emptyset 12}^{\overline{\a}\b\overline{\c}}\} & 2 & 1 \\
  &  \a\overline{\b\c} &  \PL_\i^\a & \cap & (\PL_\j^\b & \vartriangle & \PL_\k^{\sim\c}) & \scriptstyle\cap\vartriangle &  \PL_{\emptyset 12}^{\T\overline{\F\T}} = \{\PP,\PPi\} & \{\PL_{\emptyset 12}^{\a\overline{\b\c}}\} & 2 & 1 \\\hline
 \Lambda_\IdxThreeBitET &  \overline{\a\b\c} &  (\PL_\i^\a & \vartriangle & \PL_\j^{\sim\b}) \cap  (\PL_\j^{\b} & \vartriangle & \PL_\k^{\sim\c}) & \scriptstyle\vartriangle\cap\vartriangle & \PL_{\emptyset 12}^{\overline{\T\T\T}} = \{\DR,\EQ\} & \{\PL_{\emptyset 12}^{\overline{\a\b\c}} \} & 2 & 1 \\
\end{array}$$
\end{table}

\subsubsection{Index Pattern, Cell Types and Permutation Invariance}

The arrangement of upper indices in the second column of Tab. \ref{tab:VectorsInIndexNotation} is categorized into seven  groups (information forms or \iforms). The set of \iforms\ is denoted by $\IdxSetX = \{\IdxZeroBit, \IdxOneBitPR$, $\IdxOneBitET$, $\IdxTwoBitPR, \IdxTwoBitET$, $\IdxThreeBitET, \IdxThreeBitPR\}$. For each of the seven groups of index patterns, there exists a class or type of partition cells. Cells of type (of \iform) $\nu$ are collected in the set $\ORccFive_\nu$.  The correspondence between the \iforms\ in the set $\IdxSetX$ of index patterns and the type $\ORccFive_\nu$ of sub-spaces is indicated by the index $\nu\in\IdxSetX$ as illustrated in Tab. \ref{tab:VectorsInIndexNotation} and Fig. \ref{fig:RccFivePartitionLattice}.

The sub-space of the \ThreeBit information cube / the partition cell that is identified by an expression of the form  $\PL_{\, \i\ \hat{\j}\ \k}^{\alpha\beta\gamma}\in\bigcup\Pi(\ORccFive)$  with  $\{\i,\j,\k\}=\IdxSet$ remains invariant under the rearrangements of triples of pairs of upper and lower indexes: 
\begin{lemma}
\label{lemma:PermutationUpperAndLower}
If $\PL_{\, \i\ \j\ \k}^{\alpha\beta\gamma}\ \mapsto\ \PL\rho\left(\binom{\alpha}{\i}\binom{\beta}{\j}\binom{\gamma}{\k}\right)$ signifies the rearrangement of triples of pairs of upper and lower indexes of a partition cell in index notation by a member $\rho$ of the symmetric group $\textsf{S}_3$, then this rearrangement does not affect the content of the cell, i.e. $\PL_{\, \i\ \j\ \k}^{\alpha\beta\gamma} = \PL\rho\left(\binom{\alpha}{\i}\binom{\beta}{\j}\binom{\gamma}{\k}\right)$ for all $\rho\in\textsf{S}_3$.\footnote{If $\alpha\beta\gamma\in\Lambda_\IdxThreeBitET$ then the implicit assumption that the lower index in the middle of the sequence designates the partition cell that occurs in both discrete unions (i.e. $\PL_{\i\, \j\, \k}^{\overline{\a\b\c}} =  (\PL_\i^\a\vartriangle\PL_\j^{\sim\b})\cap(\PL_\j^{\b}\vartriangle\PL_\k^{\sim\c})$) needs to be made explicit by using the hat-based notation (i.e. $\i\j\k$ becomes $\i\hat{\j}\k$) before permutations are performed.}
\end{lemma}
\begin{proof}
By cases for all permutations of $\IdxSet=\{i,j,k\}$ and all $\alpha\beta\gamma\in\Lambda_\nu$ and all $\nu\in\IdxSetX$. ($\i\j\k=\emptyset12$): ($\nu\in\IdxOneBitPR$) $\PL_{\emptyset12}^{\a--} = \PL_{\, 1\emptyset2}^{-\a-} = \PL_{\, 2\, 1\, \emptyset}^{--\a} = \PL_\emptyset^\a$; ($\nu\in\IdxOneBitET$) $\PL_{\emptyset12}^{\overline{\a\b}-} = \PL_{ 1\emptyset2}^{\overline{\b\a}-} = \PL_{\, 2\, 1\, \emptyset}^{-\overline{\b\a}} = \PL_{1\, 2\, \emptyset}^{\overline{\b}-\overline{\a}} = \PL_\emptyset^\a\vartriangle\PL_1^{\sim\b} = \PL_1^\b\vartriangle\PL_\emptyset^{\sim\a}, \ldots$; ($\nu\in\IdxTwoBitET$) $\PL_{\emptyset12}^{\overline{\a\b}\c} = \PL_{ 1\emptyset2}^{\overline{\b\a}\c} = \PL_{ 2 1 \emptyset}^{\c\overline{\b\a}} = \PL_{1 2 \emptyset}^{\overline{\b}\c\overline{\a}} = (\PL_\emptyset^\a\vartriangle\PL_1^{\sim\b})\cap\PL_2^\c = (\PL_1^\b\vartriangle\PL_\emptyset^{\sim\a})\cap\PL_2^\c, \ldots$;  \ldots;  ($\nu\in\IdxThreeBitET$) $\PL_{\emptyset 12}^{\overline{\a\b\c}} = \PL_{\emptyset \hat{1}2}^{\overline{\a\b\c}} = \PL_{ \hat{1}\emptyset2}^{\overline{\b\a\c}} = \PL_{ 2 \hat{1} \emptyset}^{\overline{\c\b\a}} = \PL_{\hat{1} 2 \emptyset}^{\overline{\b\c\a}} = (\PL_\emptyset^\a\vartriangle\PL_1^{\sim\b})\cap(\PL_1^\b\vartriangle\PL_2^{\sim\c}) = (\PL_1^\b\vartriangle\PL_\emptyset^{\sim\a})\cap(\PL_1^\b\vartriangle\PL_2^{\sim\c})=(\PL_2^\c\vartriangle\PL_1^{\sim\b})\cap(\PL_1^\b\vartriangle\PL_\emptyset^{\sim\a})= (\PL_1^\b\vartriangle\PL_2^{\sim\c})\cap(\PL_1^\b\vartriangle\PL_\emptyset^{\sim\a}), \ldots$
\end{proof}
\noindent  That is, permutations $\rho$ of the symmetric group $\textsf{S}_3$ (e.g., \cite{Cameron1999}) do not change the sub-spaces that are identified. This is because the sequence $\i\j\k$ of lower indexes already stands for an arbitrary permutation of the index set $\{\emptyset,1,2\}=\IdxSet$ and the joint re-arrangement of upper and lower indexes does not change the relation between upper and lower indexes. By contrast, permutations of only the upper indexes or only the lower indexes can change the identity of the designated cell but leave its geometric type invariant, i.e., 
\begin{equation}
\label{eq:PermutationUpperOrLower}
\PL_{\i\ \j\ \k}^{\alpha\beta\gamma}\in \ORccFive_\nu\ \text{iff}\ \PL_{\rho(\, \i\ \j\ \k)}^{\ (\alpha\beta\gamma)}\in \ORccFive_\nu\ \text{and}\ \PL_{\i\ \j\ \k}^{\alpha\beta\gamma}\in \ORccFive_\nu\ \text{iff}\ \PL_{\ (\, \i\ \j\ \k)}^{\rho  (\alpha\beta\gamma)}\in \ORccFive_\nu\ \text{for all}\ \rho\in\textsf{S}_3.
\end{equation}

\subsubsection{The \ThreeBit Partition Lattice in Index Notation}

The \ThreeBit partition lattice $\Pi (\ORccFive )$ is a structure whose carrier set is the set of the fifteen partitions of the set $\ORccFive$ of mereological relations of Fig. \ref{fig:rccFive} and Tab.~\ref{tab:rccFive}. The cells/blocks of these partitions are described using index notation (Tab. \ref{tab:VectorsInIndexNotation}) such that, for a fixed frame / arrangement of lower indexes $\i\j\k$, every upper index pattern $\alpha\beta\gamma$ identifies one member of the carrier set of $\Pi (\ORccFive )$: 
\begin{equation}
\label{eq:PartitionLatticeCarrierSet}
\Pi (\ORccFive ) = \{ \{\PL_{\i\ \j\ \k}^{---}\},\{\PL_{\, \i\ \j\ \k}^{\a--}\},\ldots, \{\PL_{\i\, \j\, \k}^{\a\overline{\b\c}}\}, \{\PL_{\i\, \j\, \k}^{\overline{\a\b\c}}\}\} = \{ \{\PL_{\i\ \j\ \k}^{\alpha\beta\gamma}\} \mid \exists \nu\in\IdxSetX.\ \alpha\beta\gamma \in \Lambda_\nu\ \}. \end{equation}

\noindent On the refinement order of Eq. \ref{eq:lesssim_def},  $0_{\ORccFive} = \{\ORccFive\}=\{\PL_{\, \i\ \j\ \k}^{---}\}$ is the most unrefined partition in $\Pi(\ORccFive)$ and $1_{\ORccFive} = \{\PL_{\i\, \j\, \k}^{\a\b\c}\}$ is the most refined partition in $\Pi(\ORccFive)$. Join and meet operations arise in the standard ways as least upper and greatest lower bounds. The resulting \ThreeBit partition lattice $((\Pi(\ORccFive),\sqcap,\sqcup,\{\PL_
    {\,\i\ \j\ \k}^{---}\},\{\PL_
    {\,\i\,\j\,\k}^{\a\b\c}\})$ is depicted in Fig. \ref{fig:RccFivePartitionLattice}.

\subsection{Partition Cells and Information States}
\label{sec:SubspacesIstates}

Information theory abstracts from the actual content of information, identifying  processes such as measurements and observations as interactions that establish correlations between the systems being observed and the observers. The amount of information that one system possesses about another is quantified by counting the possible alternatives that give rise to such correlations \citep{Shannon:1948jz,Rovelli:1995fv}. Following \citet{hintikka1962knowledge,Groenendijk:2003bc,Ciardelli2018}, the information state of a system is identified with the set of possible worlds that are compatible with a certain state of affairs. A non-empty set of information states indicates that the actual state of affairs is one of the possible worlds in the set. Information is encoded in information states (\istates) because an \istate\ may exclude as elements possible choices from a set of alternatives. The fewer possible worlds that are included in a (non-empty) $\istate$, the more information it contains. 

Specifying mereological information involves identifying sets of mereologically possible worlds or mereological information states. These possibilities are limited by the axioms of mereology, the formal theory that describes the relations between parts and wholes, and wholes and their parts \cite{simons:parts,varzi:Mereology:03,Cohn01a}. Following \citet{bittner:InformationMereologyAndVagueness,Bittner:VMQI}, this paper focuses on worlds that are mereologically distinguishable for ordered pairs of regions, i.e. the binary mereological ($\rccFivePlus$) relations of Fig. \ref{fig:rccFive} and Tab.~\ref{tab:rccFive}. The mereological information of ordered pairs of regions is encoded in \istates\ that exclude certain choices from the set of mereologically distinguishable alternatives. As illustrated in Example~\ref{ex:RunningExample} information states of  elementary systems \istates\ are always relational, i.e. in relation to some observer $o\in\OO$ with $\OO\cap\OReg=\{\}$. The term $\istate_o(r_1,r_2)$ signifies the maximum information state of the elementary system $(r_1,r_2)$ with respect to the observer $o$.

\begin{remark}\normalfont
The notion of an observer is philosophically difficult. The 'thin'  observer (every entity can be an observer) of Relational Quantum Mechanics has difficulties \citep{Adlam2023}. The notion of a 'thick' (information processing and possibly cognitive) observer in Quantized Mereology is not without difficulties either. The latter is defended in \citep{bittner:Hilbert}. \qed
\end{remark}

Mereological \istates\ are subsets of the set $\ORccFive$, where each element of $\ORccFive$ matches one of the eight unique mereological types of ordered region pairs in \Reg, as shown in Fig. \ref{fig:rccFive} and Tab.~\ref{tab:rccFive}.  Subsequent discussions will concentrate on those subsets of $\ORccFive$ that constitute the cells in the partitions, which together create $\Pi(\ORccFive)$ the carrier set for the $\ThreeBit$ partition lattice shown in Fig. \ref{fig:RccFivePartitionLattice}.  In this context, the members of the set $\bigcup\Pi(\ORccFive)$, the mereological \istates, are related to the ordered pairs of regions in $\OReg^2$ and observers $o\in\OO$ (signified as $\istate_o(r_1,r_2)\in\bigcup\Pi(\ORccFive)$) via the following postulate (adapted from \citep{Bittner:VMQI}):\footnote{Postulates are statements in the meta-language  expressing principles that are assumed to be true of quantized mereology. In the formalism they shape the statement of axioms and/or arise as theorems.}

\begin{postulate}[Information states, observers and elementary systems] \label{postulate:Istates} 
Let $o\in\OO$ be an  observer (`asking' questions, `receiving' answers). 
(i) A sub-space $Z$ of the \ThreeBit cube / a cell in the \ThreeBit partition lattice $\Pi(\ORccFive)$ / an $\istate$ wrt. $o$ ($\istate_o$) in the sense of information theory, carries \emph{mereological} information if and only if it is correlated with an ordered pair of regions $(r_1,r_2) \in\OReg^2$ and the correlation manifests itself in terms of the mereological properties that are listed in Tab.~\ref{tab:rccFive}; (ii) For a given observer $o\in\OO$ a pair of regions $(r_1,r_2)$ cannot be simultaneously correlated in this sense  with multiple maximum \istates/cells/sub-spaces wrt. $o$;  (iii) Every \istate/cell/sub-space of $\bigcup\Pi(\ORccFive)$ other than $\ORccFive$ itself is possible for some observer $o\in\OO$ and some elementary system $(r_1,r_2) \in\OReg^2$. 
\end{postulate} 
\noindent  In this context  one can define (adapted from \citep{Bittner:VMQI})  in the spirit of the relational interpretation of (quantum) information \cite{Rovelli:1995fv}:   
\begin{definition}
\label{definition:InfoCapacity}
If $Z\in \bigcup\Pi(\ORccFive)$ is the \istate/cell/sub-space that is correlated with the elementary system $(r_1,r_2)\in\OReg^2$ with respect to some observer $o\in\OO$ ($\istate_o(r_1,r_2)=Z$), then $Z$ is a cell in some partition of $\Pi(\ORccFive)$ (i.e. $Z\in \pi\in\Pi(\ORccFive)$) and  $|Z|$ and $|\pi|$ respectively signify the cardinality of the sets $Z$ and $\pi$. In this context one can define:  (i) the ordered pair $(o,(r_1,r_2))$ or $(r_1,r_2)_o$ is called a \emph{relational elementary system}.
(ii) The \istate\ $Z$ \emph{carries} $\log_2 |\pi|$  bits of mereological information about the relational elementary system $(r_1,r_2)_o$; (iii) The quantity $\log_2 |\pi|$ is the \emph{information capacity}   of $(r_1,r_2)_o$ wrt. $o$ and (iv) The relational elementary system $(r_1,r_2)_o$ is subject to the amount of $\log_2 |Z|$ bits of \emph{ontological vagueness}.
\end{definition}
\noindent An important consequence of the relational view underlying Quantized Mereology is that the notions of `information state', `information capacity' and `amount/degree of vagueness' apply to relational elementary systems (members of the set $\OO\times\OReg^2$) only. Examples of \istates\ that can obtain of relational elementary systems as well as the amount of information that can be carried  and the amount/degree of vagueness that can be captured are enumerated in the two right-most columns of Tab. \ref{tab:VectorsInIndexNotation}. 

The information states of relational elementary systems can be categorized into information forms (\iforms) that describe the different patterns of determinacy and indeterminacy in the bits describing the truth values of the predicates  $Q_\emptyset$, $Q_1$, and $Q_2$ that define mereological relations. This will be discussed in Sec.~\ref{sec:InteractionMaps}. The partition-theoretic correspondence between \istates\ and \iforms\ was enumerated in Tab.~\ref{tab:VectorsInIndexNotation} and Fig. \ref{fig:RccFivePartitionLattice}.

\section{Elementary Mereological Questions and Their Meaning}
\label{sec:ElementaryQuestions}

Having defined the state space of possible mereological outcomes in the previous section, the next step is to build the bridge between an observer's inquiry and these formal information states. This section introduces the crucial \textit{interaction map, $\sharp$}, and begins to detail its dual function within the framework. In this section, the primary focus is on its role as the  \textit{definitional anchor} for the top-down \textit{direct semantic path}, by providing the concrete classical fact that grounds the meaning of a mereological proposition.

This will also lay the groundwork for understanding the  second function of the interaction map  as the \textit{constructive first step} in the bottom-up \textit{indirect quantization path}. The full, formal details of this constructive role will be elaborated in Section \ref{sec:InteractionMaps}. Understanding both functions of this map is key to the coherence of the Quantized Information Semantics (QIS) principle, which unites these two pathways.

Information acquisition can be decomposed into the acquisition of individual pieces of information in response to elementary yes/no questions posed by an observer \citep{Wheeler1989-WHEIPQ}. This is because a yes/no answer to a question eliminates certain options from a set of alternatives, thereby conveying information. Logically possible pieces of mereological information for ordered pairs of regions with respect to some observer are encoded in the various cells of the \ThreeBit partition lattice. When responding with a yes or no answer to a  question about mereological properties, certain mereological possibilities are excluded, thereby mereological information is conveyed and correlations between observers $o\in\OO$, ordered pairs of entities/regions $(r_1,r_2)\in\OReg^2$ and information states in $\bigcup\Pi(\ORccFive)$ obtain (Postulate \ref{postulate:Istates}). The aim of what follows is to study the logic of the interactions that establish these correlations. For this purpose, the notion of  elementary questions and their meaning in the context of a support semantics from \citet{bittner:InformationMereologyAndVagueness,Bittner:VMQI} are reviewed and formally extended. For a comprehensive summary of the formal definitions that underpin the following discussion, please consult Table \ref{tab:OverviewSupportSemanticsTwo} in the Appendix.

\subsection{Mereological Structures and Mereological Systems}
\label{sec:MereologicalRelations}

A mereological structure $(\OReg \cup \{\emptyset\}, \sim, \sharp, \simNu, \natural)$ is an ordered quintuple where $\OReg$ is a set of regions and $\sim$ is the equivalence relation on the set $\Reg^2$ which is defined in Tab.~\ref{tab:rccFive}.  It is assumed that $\OReg$ is partitioned into four non-empty subsets: $\cReg$, $\vReg$, $\zReg$ and $\eReg$, as shown in Table \ref{tab:OverviewSupportSemantics}(1).\footnote{This four-part classification of regions provides a concrete model for grounding the theory's information theoretic analysis of relational elementary systems. The precise logical and philosophical role of these categories of regions is discussed in more detail in Section \ref{sec:RegOneBitAndRegTwoBitAndRegThreeBit}} $\cReg$ represents crisp regions that do not include the empty region, while $\vReg$, $\zReg$ and $\eReg$ are regions that are subject to increasing degrees of vagueness due to limitations on the amount of information that is possible. Elementary mereological structures are ordered pairs of regions $(r_1,r_2)\in\OReg^2$. 

The third component of a mereological structure $(\OReg \cup \{\emptyset\}, \sim, \sharp, \simNu, \natural)$ is a map $\sharp : \IdxSet \to \OO \to \OReg^2 \to  \Reg^2$. For a given observer $o\in\OO$ the map $\sharp_o$ assigns pairs of elements from $\OReg$ to pairs of elements from $\Reg$ based on postulated restrictions that determine how the elements of $\OReg^2$ can be described in terms of their relationship to the lattice structure of $\Reg$. The map $\sharp$ is called the \textit{interaction map} and is discussed in Sections \ref{sec:ElementaryQuestionsInteractionMaps} and \ref{sec:InteractionMaps}. 

The fourth component of the mereological structure $(\OReg\cup \{\emptyset\}, \sim, \sharp, \simNu, \natural)$ is a collection of equivalence relations $\simNu$ on subsets of $\OO\times\OReg^2$. These relations signify that relational elementary systems are information-theoretically equivalent if they are in (correlated with) the same information state in $\bigcup\Pi(\ORccFive)$ \cite{Bittner:VMQI}.  The last component of the mereological structure is a map $\natural : \OO\to\OReg^2 \to \OReg^2$ that takes relational elementary systems from one $\simNu-$  equivalence class (from one \istate) to another. The role of the map $\natural_o$ in the framework of quantized mereology is to capture situations in which relational elementary systems are changed (disturbed) by sequences of interactions via $\sharp_o$. This is discussed in \citet{bittner:InformationMereologyAndVagueness,Bittner:VMQI}.

Following the relational interpretation of quantum information \cite{Rovelli:1995fv} one can define a (relational) mereological system \M, as a pair $\M=(S,\OO)$, where $S$ is an observed system that is a mereological structure, i.e. ${S=(\OReg \cup \{\emptyset\},\sim,\sharp,(\simNu), \natural)}$ and $\OO$ with $\OO\cap\OReg=\{\}$ is a set of  observers of $S$. Members $o$ of $\OO$ obtain information about the mereological structure from $S$ by `asking' elementary yes/no questions \cite{Wheeler1989-WHEIPQ}. The `answers' to these questions are provided by $S$ and result in the information state $\istate_o$. The relational elementary systems $(r_1,r_2)_o$ of a mereological system $\M=(S,\OO)$ are the elementary structures $(r_1,r_2)\in\OReg^2$ of $S$ observed by $o$ (resulting in $\istate_o(r_1,r_2)$). The interactions of observers $o\in\OO$ with elementary mereological structures $(r_1,r_2)\in\OReg^2$ are modeled by the map $(r_1,r_2)_o \in \OO\times\OReg^2 \mapsto \sharp_o(r_1,r_2)\in\Reg^2$ and possibly affecting/disturbing the relational elementary system $(r_1,r_2)_o$ via the map $\natural$. These notations emphasize the relational foundations of Quantized Mereology.

\begin{table}[h!]
\caption{Overview of the Foundational Concepts of Quantized Mereology. (See also the corresponding rows of the more detailed Table \ref{tab:OverviewSupportSemanticsTwo} in the Appendix.)}
\label{tab:OverviewSupportSemantics}
	\centering
	\begin{tabular}{p{2cm}|p{11.5cm}|c}
	\hline
	\multicolumn{3}{c}{\textbf{Part 1: Foundational Ontology \& Axioms}} \\ \hline
	\textbf{Region Types} &
	The set of regions $\OReg$ is partitioned into four distinct, non-empty classes: (\cReg): Crisp regions with determinate boundaries.
		 (\vReg): Intermediate degree of vagueness.
		 (\zReg): maximally vague but still with  discernible mereological features.
		(\eReg): Degenerate regions with no discernible mereological features.
	& (1) \\ \hline
	\textbf{Interaction Map Axioms} &
	The behavior of the interaction map $\sharp$ is axiomatically constrained by the type of region pair. The table below defines the set of possible outcomes for both simple ($\sharp^i$) and correlational ($\sharp^\times$) interactions.
	\vspace{0.2cm}
	\centering
	\begin{tabular}{c|cccc}
	\txt{$r \in$\\$\xyReg$} & $\cReg$ & $\eReg$ & $\zReg$ & $\vReg$ \\\hline
	$\cReg$ & $\sharp_o^i(r) = r$ & $\sharp_o^i(r) \in [\botZeroOne]$ & \multicolumn{2}{c}{$\sharp_o^i(r) \in [\bot_\ORccFive]\cup[\top_\ORccFive]$}\\
	$\eReg$ & $\sharp_o^i(r) \in [\botOneZero]$ & $\sharp_o^i(r) \in [\botOneOne]$ & $\sharp_o^i(r) \in [\botOneZero]$ & $\sharp_o^i(r) \in [\botOneZero]$ \\
	$\zReg$ & \txt{$\sharp_o^i(r) \in$\\$[\bot_\ORccFive]\cup[\top_\ORccFive]$} & $\sharp_o^i(r) \in [\botZeroOne]$ & \txt{$\sharp_o^i(r) \in$\\$[\bot_\ORccFive]\cup[\top_\ORccFive],$\\$\sharp_o^{\times}$ in Eq. \ref{eq:InteractionConstraints}} & \txt{$\sharp_o^i(r) \in$\\$[\bot_\ORccFive]\cup[\top_\ORccFive],$\\$\sharp_o^{\times}$ in Eq. \ref{eq:InteractionConstraints}} \\
	$\vReg$ & \txt{$\sharp_o^i(r) \in$\\$[\bot_\ORccFive]\cup[\top_\ORccFive]$} & $\sharp_o^i(r) \in [\botZeroOne]$ & \txt{$\sharp_o^i(r) \in$\\$[\bot_\ORccFive]\cup[\top_\ORccFive],$\\$\sharp_o^{\times}$ in Eq. \ref{eq:InteractionConstraints}} & \txt{$\sharp_o^i(r) \in$\\$[\bot_\ORccFive]\cup[\top_\ORccFive],$\\$\sharp_o^{\times}$ in Eq. \ref{eq:InteractionConstraintsXX}}
	\end{tabular} & (2) \\ \hline
	\multicolumn{3}{c}{\textbf{Part 2: Support Semantics for Questions}} \\ \hline
	\textbf{Elementary Questions} &
	An observer interacts with a system by asking three elementary yes/no questions ($Q_i?$) about its mereological properties (Overlap?, Part Of?, Inverse Part Of?). Each question is linked to a proposition whose truth is defined via the interaction map:
	$Q_i(r_1,r_2)_o \equiv (\rccfive(\sharp_o^i(r_1,r_2))\, i)$  & (3) \\ \hline
	\textbf{Truthsets \& Support} &
	The meaning of an answer is given by its truthset: the set of all classical possibilities consistent with it (e.g., $|Q_i(r_o)| = \PL_i^\T$). An information state ($\is$) supports an answer if it logically entails it (i.e., $\is \subseteq |Q_i(r_o)|$). & (4,5) \\ \hline
	\textbf{Maximum \istate} &
	The maximum supporting \istate\ ($\lceil \dots \rceil$) is the full truthset for an answered question. It represents all the information gained from that answer and corresponds to a specific cell in the partition lattice (e.g., $\lceil Q_i?(r_o)=\Y \rceil = \PL_i^\T$). & (6) \\ \hline
	\end{tabular}
\end{table}

\subsection{Elementary Questions and Interaction Maps}
\label{sec:ElementaryQuestionsInteractionMaps}

Before detailing the formal properties of the framework's core components, it is crucial to understand the conceptual role of the \textit{interaction map}, denoted by $\sharp$. This map should \textit{not} be mistaken for a literal model of a physical measurement process. Rather, it serves as a formal bridge between the potentially vague entities in $\OReg^2$ and the crisp, lattice-structured domain of $\Reg^2$ where classical mereological operators like $\meet$ are well-defined. Its function is to model how an interaction between an observer and a vague system resolves its inherent indeterminacy into one of a limited set of possible classical outcomes. In short, the $\sharp$ map provides a determinate, classical proxy for an indeterminate pair, thereby enabling a truth-functional evaluation of a mereological proposition. This role is illustrated concretely in Example~\ref{ex:SharpMap}.

An observing system $o\in\OO$ can obtain information about the mereological relations for an ordered pair of regions $(r_1,r_2)\in\OReg^2$ by asking three elementary yes/no questions. These questions are represented by maps of the form $Q?: \IdxSet \to \OO \to \OReg^2 \to \YN$, with answers indicated by $\Y$ or $\N$ (Tab.~\ref{tab:OverviewSupportSemantics}(3) and Tab.~\ref{tab:OverviewSupportSemanticsTwo}(3)). Each question $Q_i?$ is formally linked to a proposition, $Q_i$, whose truth or falsity is determined via the interaction map. For example, the proposition for the overlap question is defined as:
$$ Q_\emptyset (r_1,r_2)_o \equiv \pr_1(\sharp_o^\emptyset (r_1,r_2)) \meet \pr_2(\sharp_o^\emptyset (r_1,r_2)) \neq \emptyset $$
As this definition shows, the interaction map is indexed by $i \in \IdxSet$, since the specific classical proxy generated, $\sharp_o^i(r_1,r_2)$, may depend on the question being asked. In the classical limit, for crisp regions in $\cReg^2$, this map simplifies to the identity for all questions, i.e., $\sharp^i_o = (\id,\id)$. The nature of this map for vague regions is discussed in detail in Section~\ref{sec:InteractionMaps}.

\begin{example}[The  Map $\sharp$ as a Formal Bridge]\normalfont
\label{ex:SharpMap}
The interaction map $\sharp$  serves as a formal bridge, providing a determinate, crisp proxy from $\Reg^2$ for an indeterminate pair from $\OReg^2$ so that classical mereological operators can apply. Its behavior is governed by axiomatic constraints (Tab.~\ref{tab:OverviewSupportSemantics}(2)). Consider a system formed by a crisp National Park boundary ($\NPB \in \cReg$) and a vague Fern Habitat ($\FH \in \vReg$), observed by $O_4$. For this type of system $(\NPB, \FH)_{O_4} \in \OO\times\cvReg$, the axioms state that the interaction's output, $\sharp^i_{O_4}(\NPB, \FH)$, must be a crisp pair whose relation is either Disjoint (\DR) or Equal (\EQ). The indeterminacy of the vague system means that before the interaction, it is unsettled which of these two outcomes will be produced; the act of asking a question forces a resolution. For instance, when $O_4$ asks the Overlap? question ($Q_\emptyset?$), then (i) if the interaction $\sharp^\emptyset_{O_4}(\NPB, \FH)$ yields a pair related by \DR, the proposition $Q_\emptyset$ evaluates to \F, and the answer is \N\ and (ii)  if it yields a pair related by \EQ, the proposition $Q_\emptyset$ evaluates to \T, and the answer is \Y.
\qed
\end{example}

It is important to clarify the role the $\sharp$-map plays in the definitions of Table~\ref{tab:OverviewSupportSemantics}. Here, it is used as a formal tool to ground the meaning of a mereological proposition. To give a precise meaning to a question about a vague pair $(r_1, r_2)$, the framework first uses the map $\sharp_o^i$ to generate a single, crisp proxy. The properties of this proxy then serve as the "gold standard" for defining the truth or falsehood of the proposition $Q_i(r_1,r_2)_o$. This function of the $\sharp$ map as a \textit{definitional anchor} is the cornerstone of what will later be called the direct semantic path, where the goal is to identify the set of all classical possibilities (the truth set) consistent with the fact established by this interaction.\footnote{This usage is not circular. The $\sharp$ map is used here to define the \textit{propositions} ($Q_i$), but its own behavior is not defined by them. Instead, it is independently governed by the axiomatic constraints detailed in Section~\ref{sec:InteractionMaps}.}

\subsection{Linking Elementary Systems to Maximum Information States}
\label{sec:LinkingIstates}
Yes / no answers to elementary questions of the form $Q_i?(r_1,r_2)_o$  are correlated with the truth and falsehood of propositions of the form $Q_i(r_1,r_2)_o$.  The truth of the proposition $Q_i (r_1,r_2)_o$ is symbolized in the context of a standard reference semantics using expressions of the form $\sharp_o^i(r_1,r_2) \models Q_i (r_1,r_2)_o$  such that $\sharp_o^i(r_1,r_2) \models Q_i (r_1,r_2)_o$ iff $(\rccfive\,(\sharp_o^i(r_1,r_2))\, i)=\T$ in the sense of Tab. \ref{tab:rccFive}.
 
\citet{Ciardelli2018} and others have shown that the meaning of a question (here yes/no questions of the form $Q_i?(r_1,r_2)_o$)  can be specified by outlining what information it takes to establish that the associated propositions (here $\oQ_i(r_1,r_2)_o$ with $\oQ_i \in \{Q_i,\neg Q_i\}$) are true. To specify  information  about a relational elementary mereological system in $(r_1,r_2)_o\in\OO\times\OReg^2$ is to identify a set of possible worlds / states of affairs / relations at/of which  propositions of the form $\oQ_i(r_1,r_2)_o$ are true. The collection of worlds in this context is the set  $\ORccFive$  of Tab.~\ref{tab:rccFive}, where it is assumed that $\ORccFive = \{ \rccfive(s_1,s_2) \mid (s_1,s_2) \in \Reg^2 \}$ as specified in Tab. \ref{tab:OverviewSupportSemanticsTwo}(1) (enough regions). 

The minimum amount of information that establishes that the proposition $\oQ_i(r_1,r_2)_o$ is true is contained in the subset of $\ORccFive$ where the proposition $\oQ_i(s_1,s_2)_o$ is true for some $(s_1,s_2)\in\OReg^2$ (including $(s_1,s_2)=(r_1,r_2)$). The respective subsets of $\ORccFive$ are called the truth sets of $\oQ_i(r_1,r_2)_o$ and are signified as $|\oQ_i(r_1,r_2)_o|$ (Tab.~\ref{tab:OverviewSupportSemantics}(4) and Tab.~\ref{tab:OverviewSupportSemanticsTwo}(4)). Proper, non-empty subsets of truth sets carry more information because they exclude more possibilities.

The meaning of a yes-answer to the question $Q_i?(r_1,r_2)_o$ is determined by the subsets of the truth set  $ |Q_i(r_1,r_2)_o|$  and the meaning of a no-answer is determined by the subsets of $|\neg Q_i(r_1,r_2)_o|$. In Table \ref{tab:OverviewSupportSemantics}(5) this is signified as $\is \Models (Q_i?(r_1,r_2)_o=\a)$ which means that the information state $\is\subseteq\ORccFive$ implies that the proposition $Q_i(r_1,r_2)_o$ is true if $\a=\Y$ and $\neg Q_i(r_1,r_2)_o$ is true if $\a=\N$. In short, $\is \Models (Q_i?$ $(r_1,r_2)_o=\a)$ means that the information state $\is$ settles $Q_i?(r_1,r_2)_o=\a$.   The maximum information state that settles $Q_i?(r_1,r_2)_o=\a$ is denoted $\lrceil{Q_i?(r_1,r_2)_o=\a}$ (Tab. \ref{tab:OverviewSupportSemantics}(6)) and $\lceil Q_i?(r_1,r_2)_o=\a\rceil \Models (Q_i?(r_1,r_2)_o=\a)$ means that $\lceil Q_i?(r_1,r_2)_o=\a\rceil$ is the maximum support set for $Q_i?(r_1,r_2)_o=\a$. The maximum support sets carry the minimum amount of information required to support the answer to the question at hand.

The non-empty maximum \istates\ that support answered elementary questions of the form $Q_i?(r_1,r_2)_o=\a$ are unique and identify cells in the \ThreeBit partition lattice of the form $\PL_{\,\i\ \j\ \k}^{\a--}=\PL_\i^\a=\lceil Q_i?(r_1,r_2)=\a \rceil$ (assuming $\lceil Q_i?(r_1,r_2)_o=\a\rceil\neq\{\}$).  To simplify the notation, it is assumed that, when used as indexes,  yes/no answers can be interchanged with truth values in the sense that there is a one-to-one correspondence of the form $\Y \leftrightarrow \T$ and $\N \leftrightarrow \F$ (Tab. \ref{tab:OverviewSupportSemantics} ((5) and (6))). For more details and background, see \cite{Ciardelli2018,bittner:InformationMereologyAndVagueness}.

\begin{example}[The Support Relation in Action]
\normalfont
\label{ex:SupportRelation}
The formal support semantics can be grounded using the Soil Chemist scenario from the running example (Example~\ref{ex:RunningExample}). When the chemist ($O_1$) asks $Q_\emptyset?(AR,FH)$ and gets the answer \Y, their finding is formally licensed by the support relation. The truth set for this answer, $|Q_\emptyset(AR,FH)_{O_1}|$, is the set of all classical relations where overlap is true. Any non-empty subset of this truth set, $\is$, will support the answer, which is written formally as $\is \Models (Q_\emptyset?(AR,FH)_{O_1}=\Y)$. Since the full truth set is the largest possible set that supports this answer, it constitutes the maximum support state, $\lceil Q_\emptyset?(AR,FH)_{O_1}=\Y \rceil$, which is equivalent to the partition cell $\PL_\emptyset^\T$. \qed
\end{example}

\section{The QIS Framework and the Compositionality of Questions, Answers, Information States and Qubits}
\label{sec:CompositionalityRoadmap}

As established in the introduction, the central thesis of this framework is encapsulated in the \textit{Quantized Information Semantics (QIS) principle}, which guarantees that the direct semantic path and the indirect quantization path for determining an information state yield the same result. The preceding sections have developed all the necessary components: the state space of the partition lattice, the interaction maps, and the support semantics for questions. This section now synthesizes these elements to formally detail the two pathways and prove their equivalence as stated in the QIS equation:
\begin{equation}
\label{eq:MeaningOfMereoQuantification}
(\overline{\lfloor\cdot\rfloor}\circ\overline{\rccfive}\circ\overline{\sharp})(r_o) = \lceil\QQ_{\i\j\k}?=\A^{\alpha\beta\gamma}\rceil(r_o)\qquad \text{(QIS)} 
\end{equation}
A key to this guaranteed equivalence is that both pathways are grounded in the same formal phenomenon of observer-system interaction, modeled by the $\sharp$-map, which functions in a dual capacity across the two paths. The remainder of this section will elaborate on each pathway and detail the compositional logic that governs the combination of questions, answers, and states for vague systems.

\paragraph{The Dual Role of $\sharp$ as the Linchpin of Coherence:}

The equivalence stated in the QIS principle is not a coincidence, but a necessary consequence of the dual role played by the interaction map, $\sharp$. How this dual role ensures the framework's internal coherence can be summarized as follows:

In the Indirect Quantization Path (left side of QIS) the map $\overline{\sharp}$ functions as a generative mechanism. It takes a single relational elementary system $(r_1,r_2)_o$ and produces an indexed family of crisp proxies. This constitutes the raw classical data that is subsequently evaluated and quantized into a final $\istate$.  In the Direct Semantic Path (right side of QIS) the map $\sharp$ functions as a definitional anchor. It provides the specific "gold standard" classical outcome that defines the proposition underlying a question. The final $\istate$ is then determined by identifying all classical relations consistent with this definition.
    
Because both pathways are rooted in the very same formal interaction, modeled by $\sharp$, their outcomes must align. The map provides the raw data for the bottom-up construction while simultaneously setting the standard for the semantic definition on the other. This shared foundation guarantees that the constructed state is identical to the state's semantic meaning, thus ensuring the coherence of the QIS principle.

It is crucial to differentiate the various notational variations used to describe interaction maps. The granular description of an interaction map, denoted $\sharp_o^i$, models a \textit{single} interaction corresponding to a single elementary question, yielding one crisp proxy from $\Reg^2$. In contrast, the composite interaction notation, $\overline{\sharp}$, represents the entire indexed \textit{family} of all such (simultaneously possible) interactions for a given relational elementary system. It is the complete set of classical data obtained from all (simultaneously possible) probes, which is the necessary input for the bottom-up quantization path. Finally, a benchmark map, $\sharp\sharp$, will be introduced later in Section~\ref{sec:DeterminismAndInformation} solely as a formal tool to define the determinacy of an interaction.

\subsection{The Direct Semantic Path: Maximum Supporting Information States}

The right side of the QIS equation, $\lceil\QQ_{\i\j\k}?=\A^{\alpha\beta\gamma}\rceil$, represents the direct semantic path to determining an information state. It formalizes the intuitive (classical) idea that answering a question narrows the field of possibilities. This expression should be read as "the maximum supporting information state for the answer $\A^{\alpha\beta\gamma}$ to the question $\QQ_{\i\j\k}?(r_o)$."  The term "$\QQ_{\i\j\k}?(r_o)$" denotes a potentially complex mereological question posed by an observer $o$ to an elementary system $r=(r_1,r_2)$.
 The term $\A^{\alpha\beta\gamma}$ represents the specific pattern of 'Yes'/'No' answers obtained. The notation $\lceil...\rceil$ formalizes the concept of a maximum supporting set within the support semantics reviewed in Section \ref{sec:ElementaryQuestions} \citep{Ciardelli2018}. In this view, the resulting information state is simply the complete set of classical mereological relations from $\ORccFive$ that remain consistent with the given answer. 

\begin{example}\label{ex:QIS_right}\normalfont
If an observer $O$ determines that regions $(r_1,r_2)=r$ overlap ($Q_\emptyset=\T$),  that $r_1$ is part of $r_2$ ($Q_1=\T$) and that inverse part of is indeterminate for $(r_1,r_2)$, the combined answered question is $\QQ_{\emptyset 1 -}?(r_O) = \A^{\Y\Y-}$. The resulting \istate\ is the intersection of the two individual maximum supporting states and the set $\ORccFive$ signaling indeterminacy:
$$ \lceil\QQ_{\emptyset 1 -}?(r_O) = \A^{\Y\Y-}\rceil = \lceil Q_\emptyset?(r_O)=\Y \rceil \cap \lceil Q_1?(r_O)=\Y \rceil\cap\ORccFive = \PL_\emptyset^\T \cap \PL_1^\T\cap\ORccFive = \{\PP,\EQ\} $$
This becomes the information state $\istate_O(r)=\PL_{\emptyset\,1\,2}^{\T\T-}$, which is the set of possibilities licensed by the observer's cumulative findings. \qed
\end{example}

The formula in Example~\ref{ex:QIS_right} includes the term `$\cap\ORccFive$` to make the logic of the third, indeterminate proposition explicit. The indeterminate answer to the question $Q_2?$ corresponds to the trivial information state $\lceil Q_2?(r_o)=\text{anything}\rceil = \PL_2^\T \cup \PL_2^\F = \ORccFive$. Thus, the full intersection $\PL_\emptyset^\T \cap \PL_1^\T \cap \ORccFive$ correctly shows that information was gained on the first two propositions but not on the third.

\subsection{The Indirect Quantization Path: Formal Construction of Qubits}

The left side of the QIS equation, $(\overline{\lfloor\cdot\rfloor}\circ\overline{\rccfive}\circ\overline{\sharp})(r_o)$, represents the indirect quantization path. This path describes a formal, bottom-up construction of an \istate\ through a sequence of three distinct operations visualized in Figure \ref{fig:MereologyVSinformation}.

\paragraph{Interaction ($\overline{\sharp}$)}: The process begins with the interaction map $\overline{\sharp}$. For a given observer $o\in\OO$ and elementary system $r \in \OReg^2$, this map produces not one, but an indexed family of crisp proxies from $\Reg^2$. 

\paragraph{Classical Evaluation ($\overline{\rccfive}$)}: The classical mereology map $\overline{\rccfive}$ is applied to each crisp proxy in the indexed family. This step yields an indexed family of  classical 3-bit patterns, $(\llbracket\sharp_{o,r}^\alpha\rrbracket, \sharpDet(\llbracket\sharp_{o,r}^\alpha\rrbracket))_{\alpha \in \IdxSetNu}$, where $\sharpDet(\llbracket\sharp_{o,r}^\alpha\rrbracket)$ signifies whether the outcome produced by the map $\llbracket\sharp_{o,r}^\alpha\rrbracket$ is the result of a determinate or an  indeterminate mapping. 
At this stage, the information is still a collection of classical facts, with the outcomes for each way the system was probed and meta-information about the nature of these outcomes.\footnote{While expressions like $(\llbracket\sharp_{o,r}^\alpha\rrbracket, \sharpDet(\llbracket\sharp_{o,r}^\alpha\rrbracket)_{\alpha \in \IdxSetNu})$ may look complicated, their components are straightforward. The subscripts and superscripts simply act as labels to keep track of which question/bit (`$\alpha$`) is being considered for which specific observer-system pair (`$_{o,r}$`). The notation ensures that the formal bookkeeping for all the classical bits that are collected is precise.}

\paragraph{Quantization ($\overline{\lfloor\cdot\rfloor}$)}: This is the final and crucial step that constructs the qubit. The quantization map $\overline{\lfloor\cdot\rfloor}$ takes the indexed family of pairs of bit patterns and performs two key functions: First, for each index $\alpha$, if the interaction is determinate ($\sharpDet(\llbracket\sharp_{o,r}^\alpha\rrbracket)$), then it yields a corresponding 1-bit \istate\ (like a face $\PL_\i^\a$ or a diagonal plane $\PL_{\times\i}^\a$). If  the interaction is indeterminate, it yields the trivial state $\ORccFive$.  Second, it combines these resulting 1-bit \istates\ by intersection ($\cap$) to construct the single, final \istate; the disjunctive union ($\vartriangle$) that forms a diagonal plane is already contained in the first step. The result of this two-stage quantization is a single cell in the \ThreeBit\ partition lattice (a vertex, edge, face, etc.) that represents the final qubit state of the system.

The QIS principle asserts that this formal, step-by-step construction will always result in the exact same partition cell $\PL_{\,\i\,\j\,\k}^{\alpha\beta\gamma}$ that is identified by the direct semantic path. This equivalence demonstrates the internal consistency of the framework, bridging the intuitive meaning of an observation with the formal machinery of quantization.\footnote{The point of this model of quantization is not that this \textit{is} the way interactions lead to sets of classical bits that then are 'quantized'. The model is intended to illustrate one way in which the process could be spelled out formally. See also  Remark~\ref{remark:CategoriesOfRegs}.}

\begin{figure}
$$\xymatrixcolsep{2pc}\xymatrix{
\Reg^2 \ar[d]_-{\txt{$\scriptstyle\rccfive$\\$\scriptstyle =\, \overline{\rccfive}$}} & \ar[l]_{\id\times\id\, =\, \overline{\sharp}} \cReg^2 \ar[d]^{\txt{$\scriptscriptstyle\lceil\QQ_{\i\j\k}?=\A^{\a\b\c}\rceil$}}&&
(({\Reg^2)}_i)_{i\in\IdxSetNu} \ar[d]_-{\txt{$\scriptstyle({\rccfive_i})_{i\in\IdxSetNu}$\\$\scriptstyle =\, \overline{\rccfive}$}}&& 
\ar@/_1pc/[ll]_-{(\sharp^i)_{i\in\IdxSetNu}}
\ar@{}[ll]|-{\overline{\sharp}} \OReg^2_\nu \ar[d]_{\QQ_{\i\j\k}?} \ar@/^3pc/[dd]|{\txt{$\scriptstyle\lceil\QQ_{\i\j\k}?$\\$\scriptstyle =$\\$\scriptstyle\A^{\alpha\beta\gamma}\rceil$}}
 \\
\RccFive \ar[r]_{\{\cdot\} = \overline{\lfloor\cdot\rfloor}} &  \ORccFive_{\IdxThreeBitPR} &&
 ((\ORccFive \times \TF)_i)_{i\in\IdxSetNu} \ar[d]_{\boxtimes}  \ar[drr]|{\overline{\lfloor\,\cdot\,\rfloor}} \ar@{<.>}[rr]^{\asymp}_{\txt{\footnotesize Tab. \ref{tab:OverviewSupportSemantics}(3)}}
& & \scriptstyle \YN^{\QQ_{\i\j\k}?} \ar[d]^{\scriptstyle \lceil\,\cdot\, \rceil}_{\txt{\footnotesize Tabs.\\ \footnotesize \ref{tab:OverviewSupportSemantics}(6),\ref{tab:CombinedQuestionsIformsIstatesCum},\\ \footnotesize \ref{tab:CombinedQuestionsIformsIstatesNonCum},\ref{tab:CombinedQuestionsIformsIstatesThree},\ref{tab:defQuestionPattern}}}\\
 & &&({\scriptstyle((\ORccFive_{\IdxOneBit}\cup\{\ORccFive\})_i)_{i\in\IdxSetNu}} \ar[rr]_-{\txt{$\boxplus$\\ \tiny(Tabs. \ref{tab:VectorsInIndexNotation},\ref{tab:Quantization})}}& \qquad\qquad\quad & \ORccFive_\nu\\
 \ar@{}[rrr]_{\txt{${\scriptstyle(\overline{\lfloor\cdot\rfloor}\circ\overline{\rccfive}\circ\overline{\sharp}) = \lceil\QQ_{\i\j\k}?=\A^{\alpha\beta\gamma}\rceil}$\\ \tiny\txt{(Fig. \ref{fig:CoarseOverview})}}} & \ar@{}[l]_{\nu=\IdxThreeBitPR} &\ar@{.}[uul]|{\OReg^2=\cReg^2} \ar@{.}[ru]|{\OReg^2\supset\, \cReg^2}& \ar@{}[rr]^{\nu \in \{\IdxZeroBit, \IdxOneBitPR, \IdxOneBitET, \IdxTwoBitPR, \IdxTwoBitET, \IdxThreeBitET, \IdxThreeBitPR\}=\IdxSetX}_{\tiny\txt{Every $\nu\in\IdxSetX$ corresponds to an index pattern for $\QQ_{\i\j\k}?=\A^{\alpha\beta\gamma}$}} &&\\
}$$
\caption{\label{fig:MereologyVSinformation}Mereology vs. information in the equation $(\overline{\lfloor\cdot\rfloor}\circ\overline{\rccfive}\circ\overline{\sharp})= \lceil\QQ_{\i\j\k}?=\A^{\alpha\beta\gamma}\rceil$ (observer is understood). Diagram based representation of the crisp/classical case (left) and of the general quantum case that captures vagueness (right and Sec. \ref{sec:CompositionalityRoadmap} to \ref{sec:CombiningQuestions}). This figure is a refinement of the upper rectangle in the diagram of Fig. \ref{fig:CoarseOverview}. The map $(\overline{\rccfive}\circ\overline{\sharp})$ is specified in more detail in the upper part of Tab. \ref{tab:sharpInducesIstate}. The two-sided dotted arrow labeled $\asymp$ relates the \ThreeBit patterns of truth values to yes/no answers of elementary questions as specified in Tab. \ref{tab:OverviewSupportSemantics}(3) and represents the correlation between truth values and yes/no answers. The quantization map $\overline{\lfloor\cdot\rfloor}$ is specified in more detail in Tab. \ref{tab:Quantization}. The substitution instances for indexed expressions of the form $\QQ_{\i\j\k}?$, $\A^{\alpha\beta\gamma}$ and $\lceil\QQ_{\i\j\k}?=\A^{\alpha\beta\gamma}\rceil$  are summarized in Tab. \ref{tab:defQuestionPattern}.)}
\end{figure}

\subsection{The Classical Limit: Crisp Elementary Systems in $\cReg^2$}
\label{sec:QIS_cReg}

The power of the QIS framework is that it generalizes classical mereology. One can see this by examining how it simplifies for crisp elementary systems in $\cReg^2$, which are information-complete and can carry three bits of information. In this special case, the general machinery reduces to its classical counterpart.

On the \textit{indirect quantization path}, the maps simplify significantly: the interaction map $\overline{\sharp}$ becomes the identity map, $(\id,\id)$, and the quantization map $\overline{\lfloor\cdot\rfloor}$ reduces to the simple set comprehension operator $\{\cdot\}$. This signifies that the classical 3-bit pattern is mapped to a singleton set representing a single vertex on the information cube.

On the \textit{direct semantic path}, all three elementary questions ($Q_\emptyset?, Q_1?, Q_2?$) can be answered. The combination of these simultaneous questions is purely cumulative, which is signified by the operator $\odot$. The combined question $(Q_\emptyset?\odot Q_1?\odot Q_2?)(r_1,r_2)_o$ is abbreviated in index notation as $\QQ_{\emptyset12}?(r_1,r_2)_o$, with its answer being a 3-tuple $\A^{\a\b\c}$. The maximum supporting \istate\ is formed by the intersection of the three individual supporting states: $\lceil\QQ_{\emptyset12}?\!=\!\A^{\a\b\c}\rceil = \PL_\emptyset^\a \cap \PL_1^\b \cap \PL_2^\c = \PL_{\emptyset12}^{\a\b\c}$.

Consequently, the general QIS equation simplifies for crisp systems to:
$$(\{\cdot\}\circ\rccfive\circ(\id,\id))(r_1,r_2) =\lceil\QQ_{\emptyset12}?=\A^{\a\b\c}\rceil(r_1,r_2)_o$$
This confirms that classical mereology is recovered as the information-complete, 3-bit limit of the quantized framework.

\subsection{Compositionality and Generalization to Vague Systems}
To discuss the QIS framework for arbitrary elementary systems in $\OReg^2\supset\cReg^2$, one must introduce general statements that link the ways of combining elementary questions, answers, and \istates. This is achieved via index-based equations of the form $\lceil\QQ_{\i\j\k}?=\A^{\alpha\beta\gamma}\rceil(r_1,r_2) = \PL_{\,\i\,\j\,\k}^{\alpha\beta\gamma}$, which syntactically connect answered questions to their corresponding partition cells/\istates.

Information from elementary questions can be combined in two ways \citep{Bittner:VMQI}: \textit{Cumulative Combination ($\odot$):} This is an additive combination of information, where asking two questions provides two independent bits of information. This corresponds to the intersection ($\cap$) of the individual \istates. \textit{Non-Cumulative Combination ($\oplus$):} This combines information in a non-additive, correlational way, corresponding to the `exclusive or' (\xor) logical operation. This is used to model entanglement and corresponds to the symmetric difference ($\vartriangle$) of the individual \istates. (More details in Sec.~\ref{sec:CombiningQuestions}.)

The rules governing the index notation for these combinations are formalized as follows:
\begin{definition}[Index Substitution Rules]
\label{definition:IndexSubstitutionRules}
An equation of the form $\lrceil{\QQ_{\i\j\k}?(r_1,r_2)=\A^{\alpha\beta\gamma}} = \PL_{\,\i\,\j\,\k}^{\alpha\beta\gamma}$ is syntactically well-formed if and only if the following constraints are satisfied: (i) a permutation of the set $\{\alpha,\beta,\gamma\}$ occurs as upper indexes and the order of upper indexes for the answer pattern and for \istates\ is identical. (ii) A permutation of the set $\{\i,\j,\k\}=\IdxSet\equiv\{\emptyset,1,2\}$ occurs as lower indexes and their order is identical for the question pattern and for \istates. (iii) Substitution rules for upper indexes are: $\alpha \rightsquigarrow x, x\in\{\a,\overline{\a},-\}$; $\beta\rightsquigarrow y, y\in\{\b,\overline{\b},-\}$ and $\gamma\rightsquigarrow z, z\in\{\c,\overline{\c},-\}$ with the correspondences $\Y \leftrightarrow \T$ and $\N \leftrightarrow \F$, where the resulting pattern is one of the index patterns $\alpha\beta\gamma\in\Lambda_\nu$ of Tab.~\ref{tab:VectorsInIndexNotation} (overlines occur in pairs or as a triple, never singly). Substitutions of an upper index pattern $\alpha\beta\gamma\rightsquigarrow\overline{\a\b\c}$ are mirrored by a substitution of lower indexes of the form \emph{$\i\j\k\rightsquigarrow\i\hat{\j}\k$}. (iv) Lower indexes of questions are over-lined or replaced by dashes to match the pattern of the upper indexes for answers.
\end{definition}

For the QIS principle to hold consistently, the amount of information carried by an \istate\ must match the information capacity of the correlated elementary system. This is ensured by the following type-matching constraint, where $\nu\in\IdxSetX$ indexes the information forms  (e.g. $\IdxOneBitPR$, $\IdxOneBitET$ etc.):
\begin{equation}
\label{eq:ORegNuVSORccFiveNu}
((\overline{\lfloor\cdot\rfloor}\circ\overline{\rccfive}\circ\overline{\sharp})(r_1,r_2)_o =  \PL_{\,\i\,\j\,\k}^{\alpha\beta\gamma}  = \lceil\QQ_{\i\j\k}?(r_1,r_2)_o=\A^{\alpha\beta\gamma}\rceil) \in \ORccFive_\nu\ \text{iff}\ (r_1,r_2)_o\in\OReg^2_\nu \quad 
\end{equation}
\noindent This constraint prevents  mappings that are nonsensical in the context of the QIS principle, such as a 1-bit system being mapped to a 3-bit information state. The index patterns themselves serve as type declarations for the maps, as summarized in Table \ref{tab:enumerationOfmapsFromQaPattren}. This table provides a roadmap for the detailed axioms presented in Sections \ref{sec:InteractionMaps} and \ref{sec:QuantizationMap}, showing how specific classes of elementary systems ($\OReg^2_\nu$) are mapped to their corresponding classes of information states ($\ORccFive_\nu$).

The QIS principle is therefore not just an equation but a comprehensive semantic framework, whose interpretation is formalized by the following postulate.
\begin{postulate}[Interpretation of QIS]
\label{postulate:MeaningOfMereoQuantification}
(i) According to the left side of QIS, partition cells are interpreted as quantum information states, and according to the right side, the same cells are interpreted as maximum supporting information states for answered questions in a classical support semantics. (ii) The correlation between the information state and the relational elementary system arises from interactions through the answers to elementary questions. (iii) Elementary questions are related to statements of classical mereology through the interaction map ($\overline{\sharp}$), such that the truth values of these statements are correlated with the yes/no answers. (iv) Classical truth values are mapped to partition cells (understood as qubits) by the quantization map ($\overline{\lfloor\cdot\rfloor}$). (v) The ways in which elementary questions and answers are combined to form complex question/answer patterns structurally mirror the ways in which information states are formed.
\end{postulate}

\begin{example}[The QIS Principle for an Entangled System]\normalfont
\label{ex:QISforEntanglement}
Consider the QIS principle in the context of the Remote Sensing Specialist ($O_2$) from the running example. Their inquiry establishes a 1-bit entangled state, and the framework must consistently account for it via both pathways.

\textit{The Direct Semantic Path:}
The specialist's informal question is, "Are the parthood relations the same?" assuming that $(AR,FH)_{O_2}$ has an information capacity of 1 bit.
To model this correlation, the framework uses the non-cumulative question $(Q_1? \oplus Q_2?)$, written in index notation as $\QQ_{-\overline{12}}?$.
A "Yes" answer to this informal question means the formal answers to $Q_1?$ and $Q_2?$ must be the same (e.g., $\b=\c$), which corresponds to a "No" result for the formal `\xor` operation.
The formal answer pattern is therefore $\A^{-\overline{\b\b}}$.
Applying the direct semantic path mapping from the QIS principle yields the state:
$$
\istate_{O_2}(AR,FH) = \lceil\QQ_{-\overline{12}}?(AR,FH)_{O_2}=\A^{-\overline{\b\b}}\rceil = \PL_{\,\emptyset 1 2}^{-\overline{\b\b}}
$$
This state is the symmetric difference $\PL_1^\b \vartriangle \PL_2^{\sim\b}$ (or more precisely $(\PL_\emptyset^\F\cup\PL_\emptyset^\T)\cap(\PL_1^\b \vartriangle \PL_2^{\sim\b})$, see Tab.~\ref{tab:VectorsInIndexNotation}), which yields the set of four mereological possibilities $\{\DR, \PO, \botOneOne, \EQ\}$.

\textit{The Indirect Quantization Path:} For a 1-bit entangled system like $(\AR,\FH)_{O_2}\equiv(O_2, (\AR,\FH)) \in \RegOneBitET$, the  interactions for the simple questions $Q_i?$ are all indeterminate: $\neg\sharpDet(\llbracket\sharp^i_{O_2,r}\rrbracket)$ for all $i \in \{\emptyset, 1, 2\}$, where the predicate $\sharpDet$ characterizes the determinacy of interaction maps  in the formalism. 
This results in the trivial (but determinate) state $\ORccFive$ for each of these probes.
 However, the interaction for the correlational question $Q_{\times \emptyset}?$ (i.e., $Q_1? \oplus Q_2?$) is determinate.
The map $\overline{\sharp}$ yields a family of crisp proxies $(\sharp^{\times \emptyset}_{O_2,r})$ which, after evaluation by $\overline{\rccfive}$, produces a family of classical patterns which are flagged as the outcome of a determinate interaction $\sharpDet(\llbracket\sharp^{\times \emptyset}_{O_2,r}\rrbracket)$.
 The quantization map $\overline{\lfloor\cdot\rfloor}$ takes this determinate result and maps it to the 1-bit entangled state $\PL_{\emptyset 1 2}^{-\overline{\b\b}}$.

Both paths lead to the same result, $\PL_{\emptyset 1 2}^{-\overline{\b\b}}$, confirming the QIS principle for non-orthogonal vagueness and demonstrating how the specialist's finding is rigorously mapped to a precise set of classical outcomes.
\qed
\end{example}

\begin{table}
\caption{\label{tab:enumerationOfmapsFromQaPattren} 
(above dashed lines) Enumeration of substitution instances of type declarations of the form  $\lfloor\QQ_{\i\j\k}?=\A^{\alpha\beta\gamma}\rfloor\in\lfloor\QQ_{\nu}?=\A^{\nu}\rfloor\ =\ \OReg^2_\nu \to \ORccFive_\nu$ (details in Sec. \ref{sec:InteractionMaps}, Lemmata \ref{lemma:InteractionQuantizationType}, \ref{lemma:QuantificationEqSupportLeft});
(below dashed lines) Enumeration of substitution instances of equations of the form $ \lfloor\QQ_{\i\j\k}?=\A^{\alpha\beta\gamma}\rfloor(r_1,r_2)_o = \PL_{\,\i\,\j\,\k}^{\alpha\beta\gamma}$ understood as value assignments that define maps of the signature/type $\lfloor\QQ_{\i\j\k}?=\A^{\alpha\beta\gamma}\rfloor\ :\ (r_1,r_2)_o \in \OReg^2_\nu \mapsto \PL_{\,\i\,\j\,\k}^{\alpha\beta\gamma}\in\{\PL_{\,\i\,\j\,\k}^{\alpha\beta\gamma}\}\subseteq\ORccFive_\nu$ (details in Sec. \ref{sec:QuantizationMap}, Lemmata \ref{lemma:InteractionQuantizationType}, \ref{lemma:QuantificationEqSupportLeft}); (below the double line) Pointer to the discussion and enumeration of substitution instances of equations of the form $ \lceil\QQ_{\i\j\k}?=\A^{\alpha\beta\gamma}\rceil(r_1,r_2)_o = \PL_{\,\i\,\j\,\k}^{\alpha\beta\gamma}$ Lemma \ref{lemma:QuantificationEqSupportRight} in Sec. \ref{sec:CombiningQuestions}.}
$$
\arraycolsep=2pt
\begin{array}{c|c|rclccccll}
 \nu & \txt{Sec. \ref{sec:InteractionMaps},\\Sec. \ref{sec:QuantizationMap},\\Tab. \ref{tab:InteractionPattern}} & \multicolumn{3}{c}{\txt{$\lfloor\QQ_{\i\j\k}?  =  \A^{\alpha\beta\gamma}\rfloor$\\$\in$\\$\lfloor\QQ_{\nu}?  =  \A^{\nu}\rfloor$}} & = & \multicolumn{2}{c}{\txt{\\$\OReg^2_\nu$}}  & \txt{\\$\to$} & \multicolumn{2}{c}{ \txt{\\$\ORccFive_\nu$}}\\
\hline
 & \text{Tab. \ref{tab:InteractionPattern}} & \multicolumn{9}{c}{\lfloor\QQ_{\i\j\k}?  =  \A^{\alpha\beta\gamma}\rfloor\quad\in\quad\lfloor\QQ_\IdxOneBitPR?  =  \A^\IdxOneBitPR \rfloor \quad = \quad \OReg^2_\IdxOneBitPR   \to  \ORccFive_\IdxOneBitPR}\\\cdashline{2-11}
 &  & \lfloor \QQ_{\i--}?&=&\A^{\a--}\rfloor & : & \OReg^2_\IdxOneBitPR && \to & \{\PL_{\,\i\ \j\ \k}^{\a--}\} \\
\IdxOneBitPR & \text{Tab. \ref{tab:Quantization}} & \lfloor \QQ_{-\j-}?&=&\A^{-\b-}\rfloor & : & \OReg^2_\IdxOneBitPR & \subseteq \ESPR & \to & \{\PL_{\,\i\ \j\ \k}^{-\b-}\} & \subset  \ORccFive_\IdxOneBitPR\\
 &  & \lfloor \QQ_{--\k}?&=&\A^{--\c}\rfloor & : & \OReg^2_\IdxOneBitPR & & \to & \{\PL_{\,\i\ \j\ \k}^{--\c}\} \\
 \hline
 & \text{Tab. \ref{tab:InteractionPattern}} & \multicolumn{9}{c}{\lfloor\QQ_{\i\j\k}?  =  \A^{\alpha\beta\gamma}\rfloor\quad\in\quad\lfloor\QQ_\IdxTwoBitPR?  =  \A^\IdxTwoBitPR \rfloor \quad = \quad \OReg^2_\IdxTwoBitPR   \to  \ORccFive_\IdxTwoBitPR}\\ \cdashline{2-11}
  &  & \lfloor \QQ_{\i\j-}?&=&\A^{\a\b-}\rfloor & : & \OReg^2_\IdxTwoBitPR &  & \to & \{\PL_{\,\i\ \j\ \k}^{\a\b-}\} \\
\IdxTwoBitPR & \text{Tab. \ref{tab:Quantization}} & \lfloor \QQ_{\i-\k}?&=&\A^{\a-\c}\rfloor & : & \OReg^2_\IdxTwoBitPR & \subseteq \ESPR & \to & \{\PL_{\,\i\ \j\ \k}^{\a-\c}\} & \subset  \ORccFive_\IdxTwoBitPR\\
  &  & \lfloor \QQ_{-\j\k}?&=&\A^{-\b\c}\rfloor & : & \OReg^2_\IdxTwoBitPR &  & \to & \{\PL_{\,\i\ \j\ \k}^{-\b\c}\} \\\hline
 & \text{Tab. \ref{tab:InteractionPattern}} & \multicolumn{9}{c}{\lfloor\QQ_{\i\j\k}?  =  \A^{\a\b\c}\rfloor\quad\in\quad\lfloor\QQ_\IdxThreeBitPR?  =  \A^\IdxThreeBitPR \rfloor \quad = \quad \OReg^2_\IdxThreeBitPR   \to  \ORccFive_\IdxThreeBitPR}\\  \cdashline{2-11}
  \IdxThreeBitPR & \text{Tab. \ref{tab:Quantization}} &  \lfloor \QQ_{\i\j\k}?&=&\A^{\a\b\c}\rfloor & : &\OReg^2_\IdxThreeBitPR & \subseteq \ESPR & \to & \{\PL_{\,\i\ \j\ \k}^{\a\b\c}\} & =  \ORccFive_\IdxThreeBitPR\\\hline
 & \text{Tab. \ref{tab:InteractionPattern}} & \multicolumn{9}{c}{\lfloor\QQ_{\i\j\k}?  =  \A^{\alpha\beta\gamma}\rfloor\quad\in\quad\lfloor\QQ_\IdxOneBitET?  =  \A^\IdxOneBitET \rfloor \quad = \quad \OReg^2_\IdxOneBitET   \to  \ORccFive_\IdxOneBitET}\\\cdashline{2-11}
 &  & \lfloor \QQ_{\overline{\i\j}-}?&=&\A^{\overline{\a\b}-}\rfloor & : & \OReg^2_\IdxOneBitET &  & \to & \{\PL_{\,\i\ \j\ \k}^{\overline{\a\b}-}\} \\
\IdxOneBitET & \text{Tab. \ref{tab:enumBoxPlusX}} & \lfloor \QQ_{\overline{\i}-\overline{\k}}?&=&\A^{\overline{\a}-\overline{\c}}\rfloor & : & \OReg^2_\IdxOneBitET & \subseteq \ESET & \to & \{\PL_{\,\i\ \j\ \k}^{\overline{\a}-\overline{\c}}\} & \subset  \ORccFive_\IdxOneBitET\\ 
  &  & \lfloor \QQ_{-\overline{\j\k}}?&=&\A^{-\overline{\b\c}}\rfloor & : & \OReg^2_\IdxOneBitET &  & \to & \{\PL_{\,\i\ \j\ \k}^{-\overline{\b\c}}\} \\\hline
 & \text{Tab. \ref{tab:InteractionPattern}} & \multicolumn{9}{c}{\lfloor\QQ_{\i\j\k}?  =  \A^{\alpha\beta\gamma}\rfloor\quad\in\quad\lfloor\QQ_\IdxTwoBitET?  =  \A^\IdxTwoBitET \rfloor \quad = \quad \OReg^2_\IdxTwoBitET   \to  \ORccFive_\IdxTwoBitET}\\  \cdashline{2-11}
  &  & \lfloor \QQ_{\overline{\i\j}\k}?&=&\A^{\overline{\a\b}\c}\rfloor & : & \OReg^2_\IdxTwoBitET &  & \to & \{\PL_{\,\i\, \j\, \k}^{\overline{\a\b}\c}\} \\
\IdxTwoBitET & \text{Tab. \ref{tab:enumBoxPlusX}} & \lfloor \QQ_{\overline{\i}\j\overline{\k}}?&=&\A^{\overline{\a}\b\overline{\c}}\rfloor & : & \OReg^2_\IdxTwoBitET & \subseteq \ESET & \to & \{\PL_{\,\i\, \j\, \k}^{\overline{\a}\b\overline{\c}}\} & \subset  \ORccFive_\IdxTwoBitET\\ 
  &  & \lfloor \QQ_{\i\overline{\j\k}}?&=&\A^{\a\overline{\b\c}}\rfloor & : & \OReg^2_\IdxTwoBitET &  & \to & \{\PL_{\,\i\, \j\, \k}^{\a\overline{\b\c}}\} \\\hline
 & \text{Tab. \ref{tab:InteractionPattern}} & \multicolumn{9}{c}{\lfloor\QQ_{\i\j\k}?  =  \A^{\alpha\beta\gamma}\rfloor\ \in\ \lfloor\QQ_\IdxThreeBitET?  =  \A^\IdxThreeBitET \rfloor \ = \ \OReg^2_\IdxThreeBitET   \to  \ORccFive_\IdxThreeBitET}\\  \cdashline{2-11}
  \IdxThreeBitET & \text{Tab. \ref{tab:Quantization}} & \lfloor \QQ_{\overline{\i\j\k}}?&=&\A^{\overline{\a\b\c}}\rfloor & : & \OReg^2_\IdxThreeBitET & \subseteq \ESET & \to & \{\PL_{\,\i\, \j\, \k}^{\overline{\a\b\c}}\} & =  \ORccFive_\IdxThreeBitET\\\hline
 \nu & \text{Sec. \ref{sec:QuantizationMap}} & \lfloor\QQ_{\i\j\k}?& = & \A^{\alpha\beta\gamma}\rfloor& : & \txt{$(r_1,r_2)_o$\\$\in$\\$\OReg^2_\nu$} & \subset \OO\times\OReg^2 & \mapsto & \txt{$\PL_{\,\i\,\j\,\k}^{\alpha\beta\gamma}$\\$\in$\\$\{\PL_{\,\i\,\j\,\k}^{\alpha\beta\gamma}\}$} & \subset \ORccFive_\nu\\\hline\hline
  \nu & \txt{Sec. \ref{sec:CombiningQuestions},\\Tab. \ref{tab:defQuestionPattern}} & \lceil\QQ_{\i\j\k}?& = & \A^{\alpha\beta\gamma}\rceil& : & \txt{$(r_1,r_2)_o$\\$\in$\\$\OReg^2_\nu$} & \subset \OO\times\OReg^2 & \mapsto & \txt{$\PL_{\,\i\,\j\,\k}^{\alpha\beta\gamma}$\\$\in$\\$\{\PL_{\,\i\,\j\,\k}^{\alpha\beta\gamma}\}$} & \subset \ORccFive_\nu\\

  \end{array}
$$
\end{table}

\section{Interaction Maps}
\label{sec:InteractionMaps}

In what follows, the notation $\lfloor\QQ_{\i\j\k}?=\A^{\alpha\beta\gamma}\rfloor(r_o)=\PL_{\,\i\,\j\,\k}^{\alpha\beta\gamma}$ serves as a shorthand for maps of the form $(\overline{\lfloor\cdot\rfloor}\circ\overline{\rccfive}\circ\overline{\sharp})(r_o)=\PL_{\,\i\,\j\,\k}^{\alpha\beta\gamma}$ in situations where the focus is on the fact that the partition cell $\PL_{\,\i\,\j\,\k}^{\alpha\beta\gamma}$ represents qubits that result from the quantization of the (up to) three classical bits that are correlated with the (up to) three yes/no answer(s) $\A^{\alpha\beta\gamma}$ to the question $\QQ_{\i\j\k}?$.

The fact that \istates\ / partition cells of the form $\PL_{\,\i\,\j\,\k}^{\alpha\beta\gamma}$ fall into families $\ORccFive_\nu$ with $\nu\in\IdxSetX$ has led to type constraints of the form $\lfloor\QQ_\nu?=\A^\nu\rfloor =  \OReg^2_\nu \to \ORccFive_\nu$ (Eq. \ref{eq:ORegNuVSORccFiveNu}) to signify families of mappings whose members are such that for index pattern of the \iform\ $\alpha\beta\gamma\in\Lambda_\nu$ the statement $\lfloor\QQ_{\i\j\k}?  =  \A^{\alpha\beta\gamma}\rfloor \in \lfloor\QQ_{\nu}?  =  \A^{\nu}\rfloor$ signifies that  the map $\lfloor\QQ_{\i\j\k}?  = \A^{\alpha\beta\gamma}\rfloor\equiv(\overline{\lfloor\cdot\rfloor}\circ\overline{\rccfive}\circ\overline{\sharp})$ is of type $\OReg^2_\nu \to \ORccFive_\nu$.

Since the map $(\overline{\lfloor\cdot\rfloor}\circ\overline{\rccfive}\circ\overline{\sharp}):\OReg^2_\nu\to\ORccFive_\nu$ has a composite structure, it follows that its composite maps satisfy the following type constraints: $\overline{\sharp}\in \OReg^2_\nu \to X$, $\overline{\rccfive}\in X \to Y$ and $\overline{\lfloor\cdot\rfloor}\in Y \to \ORccFive_\nu$, where $X$ and $Y$ are sets whose extensions are listed in the right diagram of Fig. \ref{fig:MereologyVSinformation}, but are yet to be discussed. Maps of the form $(\overline{\rccfive}\circ\overline{\sharp}):\OReg^2_\nu \to Y$ are studied in the remainder of this section and maps of the form $\overline{\lfloor\cdot\rfloor}: Y \to \ORccFive_\nu$ are investigated in Sec. \ref{sec:QuantizationMap}.

\subsection{Constraints on the Interaction Map ($\sharp$)}

Tab. \ref{tab:OverviewSupportSemantics}(2) (details in \cite{bittner:InformationMereologyAndVagueness} and below) and Eqs.  \ref{eq:InteractionConstraints} and \ref{eq:InteractionConstraintsXX} (discussed below) list constraints on the interaction map $\sharp_o$ where $o\in\OO$ keeps track of the observer. The constraints distinguish between regions that are mereologically crisp (i.e. members of \cReg), regions that are void of mereological information (i.e. members of \eReg) and regions that are neither crisp nor void of information (i.e., members of \vReg\ and \zReg\ with members of \vReg\ less vague than members of \zReg). The elementary structures in $\OReg^2$ with $\OReg=\cReg\cup\vReg\cup\eReg\cup\zReg$ are distinguished according to which subset  $\xyReg \equiv \xReg \times \yReg \subset \OReg^2$ with \textsf{x,y} $\in \{\textsf{c},\textsf{v},\textsf{e},\textsf{z}\}$ they belong to. These distinctions are used to impose constraints on the domain and co-domain of interaction maps $\sharp_o^\alpha:\OReg^2 \to \Reg^2$ such that sub-maps of the form $(\sharp_o^\alpha)$ with $(r_1,r_2)\in\xyReg \mapsto \sharp_o^\alpha(r_1,r_2) \in \Reg^2$ arise as illustrated in Table \ref{tab:OverviewSupportSemantics}(2). The map $\sharp_o^\alpha$ in conjunction with information about its (in)determinacy induce families/types of maps of the form $\overline{\sharp}:\OReg^2_\nu \to (\Reg^2_\alpha)_{\alpha\in\IdxSetNu}$ where $\IdxSetNu\supseteq\IdxSet$ is a $\nu$-dependent index set (see below in Eq. \ref{eq:IdxSetNu}) and $\nu \in \IdxSetX$ as illustrated in Fig.~\ref{fig:MereologyVSinformation}. 

The various cases that occur when using the subsets  $\xyReg \subset \OReg^2$  to impose constraints on the domain and co-domain of interaction maps $\sharp_o^\alpha:\OReg^2 \to \Reg^2$ are enumerated in the upper part of Tab. \ref{tab:sharpInducesIstate} (above the dashed lines). In general, four cases are distinguished:

\begin{table}
\caption{\label{tab:sharpInducesIstate}Axiomatic constraints on the domains and co-domains of interaction maps of the form $\sharp_o^\alpha,\sharp\sharp_o^\alpha:\xyReg \to \Reg^2$ with $\alpha \in \{\emptyset,1,2,\times \emptyset, \times1, \times2 \}$, $\xyReg \equiv \xReg \times \yReg \subset \OReg^2$, \textsf{x,y} $\in \{\textsf{c},\textsf{v},\textsf{e},\textsf{z}\}$ and $\{i,j,k\} = \IdxSet = \{\emptyset,1,2\}$. The map $\sharp\sharp_o^\alpha$ is a `benchmark' that is used to distinguish situations in which $\sharp_o^\alpha$ is determinate from those where it is indeterminate (cf. Sec.~\ref{sec:DeterminismAndInformation}). }
$$ 
\begin{array}{c||c|c|c|c|c|c}
& \multicolumn{6}{c}{\text{Axiomatic constraints on the domain and co-domain of the  maps}}\\
 & \multicolumn{6}{c}{\sharp_o^\alpha,\sharp\sharp_o^\alpha:\ r\in\xyReg \mapsto \sharp_o^\alpha(r)\in\Reg^2; r\in\xyReg \mapsto \sharp\sharp_o^\alpha(r)\in\Reg^2}\\  &\multicolumn{6}{c}{\text{(from Tab. \ref{tab:OverviewSupportSemantics}(2) and Eqs. \ref{eq:InteractionConstraints}, \ref{eq:InteractionConstraintsXX})}}\\\hline
 \txt{$r$\\$\in$\\$\OReg^2$} &\multicolumn{2}{c|}{\txt{$r$\\$\in$\\$\cvReg\cup\vcReg$\\$\cup$\\$\czReg\cup\zcReg$}}& \txt{$r$\\$\in$\\$\starEReg$\\$\cup$\\$\EstarReg$\\$\cup$\\$\eeReg$\\$\cup$\\$\ccReg$} & \multicolumn{2}{c|}{\txt{$r$\\$\in$\\$\zvReg\cup\vzReg$\\$\cup\, \zzReg$}} & \txt{$r$\\$\in$\\$\vvReg$} \\
 \txt{\rotatebox[origin=c]{-90}{$\mapsto$}} & \multicolumn{2}{c|}{\txt{\rotatebox[origin=c]{-90}{$\mapsto$}}} & \txt{\rotatebox[origin=c]{-90}{$\mapsto$}}   & \multicolumn{2}{c|}{\txt{\rotatebox[origin=c]{-90}{$\mapsto$}}}& \txt{\rotatebox[origin=c]{-90}{$\mapsto$}} \\
\txt{$\sharp_o^\alpha(r),$\\$\sharp\sharp_o^\alpha(r)$\\$\in$\\$\Reg^2$} & \multicolumn{2}{c|}{\txt{$\sharp_o^\alpha(r),\sharp\sharp_o^\alpha(r)$\\$\in$\\$[\bot_\ORccFive]\cup [\top_\ORccFive]$}} & \txt{$\sharp_o^\alpha(r),$\\$\sharp\sharp_o^\alpha(r)$\\$\in$\\$\Reg^2$} & \multicolumn{2}{c|}{ \begin{array}{l}
\txt{$\sharp_o^\alpha(r),\sharp\sharp_o^\alpha(r)$\\$\in$}\\
\text{[}\bot_i^{\downarrow}]\cup [\top_i^{\uparrow}]\  \text{or}\\
                     \text{[}\bot_i^{\uparrow}]\cup [\top_i^\downarrow]\  \text{or}\\
                     \text{[}\bot_{\ORccFive}] \cup [\bot_i]\ \text{or} \\
                      \text{[}\top_{\ORccFive}] \cup [\top_i]\\ 
                      \qquad \text{if}\  \alpha = \times i;\\
                      \text{[}\bot_\ORccFive]\cup [\top_\ORccFive]\\ 
                      \qquad \text{if}\ \alpha \in\IdxSet.\\
                \end{array}} &  \begin{array}{l}
\qquad \txt{$\sharp_o^\alpha(r),\sharp\sharp_o^\alpha(r)$\\$\in$}\\
\qquad \text{[}\bot_k^{\downarrow}]\cup [\top_k^{\uparrow}]\  \text{or}\\
                   \qquad   \text{[}\bot_k^{\uparrow}]\cup [\top_k^\downarrow]\  \text{or}\\
                 \qquad     \text{[}\bot_{\ORccFive}] \cup [\bot_k]\ \text{or} \\
                 \qquad      \text{[}\top_{\ORccFive}] \cup [\top_k]\ \text{with}\\           \qquad             k=i\ \text{if}\ \alpha = \times i\\  
             \qquad          k=j\ \text{if}\ \alpha = \times j;\\  
               \qquad        \text{[}\bot_\ORccFive]\cup [\top_\ORccFive]\ \text{if}\ \alpha \in \IdxSet;\\

               \qquad \quad       \text{and}\\
             \quad         \txt{$\{\llbracket \sharp_o^{\times i}(r)\rrbracket, \llbracket\sharp_o^{\times j}(r)\rrbracket\}$\\$\in \DIA_k$}\\

                \end{array} \\
                \hdashline 
\txt{$\sharp_o^\alpha(r),$\\$\sharp\sharp_o^\alpha(r)$\\defined\\for $\alpha \in$} & \multicolumn{3}{c|}{\alpha \in\IdxSet} &  \multicolumn{2}{c|}{\alpha \in \IdxSet \cup\{\times i\},\ i \in \IdxSet} & \txt{$\alpha \in\IdxSet\cup\{\times i, \times j\},$\\$ \{i,j,k\} = \IdxSet$}\\ 
\end{array}
$$
\end{table}

\paragraph{The case $\ccReg$}

Consider Tab. \ref{tab:OverviewSupportSemantics}(2). For crisp elementary systems in $\ccReg$ ($\cReg \times \cReg$), the operator $\sharp_o^i$ is the identity map $\id\times\id$ for all $i\in\IdxSet$, and maps of the form $(\rccfive\circ\sharp_o^i) (r_1,r_2)$ identify \rccFive relations using the \ThreeBit patterns in the set $\RccFive$ as enumerated in Tab.~\ref{tab:rccFive}.  The result of $(\sharp_o^i r)$ for elementary systems $r \in \ccReg$ is not affected by the parameter $i\in \IdxSet$.

\paragraph{The case $\ceReg\ \cup\ \ecReg\ \cup\ \evReg\ \cup\ \veReg\ \cup\ \ezReg\ \cup\ \zeReg$}

These cases capture elementary systems involving regions from $\eReg$, which are void of mereological information. For a system $r=(r_1,r_2)$ in $\starEReg \equiv \ceReg\cup\veReg\cup\zeReg$ (where $r_2 \in \eReg$), the interaction map's co-domain is restricted to the equivalence class $[\botZeroOne]$, meaning the second component of the output is always the empty region ($\pr_2(\sharp_o^i r)=\emptyset$). Symmetrically, if $r \in \EstarReg \equiv \ecReg\cup\evReg\cup\ezReg$, the output is constrained to $[\botOneZero]$, and if $r \in \eeReg$, the output is fixed to $(\emptyset,\emptyset)$ in the class $[\botOneOne]$. For any system $r \in \EstarReg\cup\starEReg\cup\eeReg$, the outcome of the interaction $(\sharp_o^i r)$ is invariant with respect to the index $i \in \IdxSet$; the specific question asked does not alter the resulting proxy when at least one region is void of mereological features.

\paragraph{The case $\cvReg\ \cup\ \vcReg \cup \czReg \cup \zcReg $}

As specified in Tab. \ref{tab:OverviewSupportSemantics}(2), for elementary systems $r\in\cvReg\ \cup\ \vcReg \cup\czReg\ \cup\ \zcReg$ the interaction map $\sharp_o^i$ is subject to the constraint $(\sharp_o^i r)\in [\bot_{\ORccFive}] \cup [\top_{\ORccFive}]\subset \Reg^2$  with $\bot_{\ORccFive}=\DR$ and $\top_{\ORccFive}=\EQ$ as specified in Fig. \ref{fig:isuccBot_iprecTop} (top) and $[\bot_{\ORccFive}]\cap [\top_{\ORccFive}]=\{\}$. That is, within the constraint $(\sharp_o^i r)\in [\bot_{\ORccFive}] \cup [\top_{\ORccFive}]$ for every $i\in\IdxSet$ there is a certain degree of freedom in that, depending on the value of the argument $i\in\IdxSet$, a total of $2^3$ different outcomes are combinatorially possible, including $(\sharp_o^\emptyset r)\in [\bot_{\ORccFive}]$, $(\sharp_o^1 r)\in [\top_{\ORccFive}]$, $(\sharp_o^2 r)\in [\bot_{\ORccFive}]$, etc.

\paragraph{The case $\vzReg\ \cup\ \zvReg\ \cup\ \zzReg$}

Elementary systems  $r\in \vzReg\cup\zvReg\cup\zzReg$ are taken to $\Reg^2$ by the map $(\sharp_o^{\times i}\ r)$ in ways that satisfy the following constraints: 
\begin{equation}
\label{eq:InteractionConstraints}
    \begin{array}{l}
\ \, \text{for one index}\ i\in\IdxSet\ \text{(bottom row of Tab.~\ref{tab:sharpInducesIstate}) the map}\ \sharp_o^{\times i}: \vzReg\cup\zvReg\cup\zzReg \to  \Reg^2\ \text{is defined, such that}\\
   (\sharp_o^{\times i}\ r) \in 
                     \text{[}\bot_i^{\downarrow}] \cup [\top_i^{\uparrow}]\  \text{or}\ (\sharp_o^{\times i}\ r) \in \text{[}\bot_i^{\uparrow}] \cup [\top_i^{\downarrow}] \  \text{or}\
                     (\sharp_o^{\times i}\ r) \in \text{[}\bot_{\ORccFive}] \cup [\bot_i]\ \text{or} \ (\sharp_o^{\times i}\ r) \in \text{[}\top_{\ORccFive}] \cup [\top_i]. \\             
    \end{array}
\end{equation}
\noindent In Equation~\ref{eq:InteractionConstraints}  the conventions of Fig. \ref{fig:isuccBot_iprecTop} (top) are used for the symbols $\bot_i^{\downarrow}, \top_i^{\uparrow}$ etc. to signify members of the set $\ORccFive$ of \ThreeBit patterns for $\rccFivePlus$ relations depending on the index $i\in\IdxSet$. One can verify in Fig. \ref{fig:isuccBot_iprecTop} (middle right) that $\bot_i^{\downarrow} \ll_i \top_i^{\uparrow}$ and $\bot_i^{\uparrow} \ll_i \top_i^{\downarrow}$ are the arrows in the graph of the \ThreeBit lattice whose vertices jointly constitute the interior plane  $\PL_{\times \i}^\T$. Similarly, the vertices of the edges $\bot_{\ORccFive} \ll_i \bot_i$ and $\top_i \ll_i \top_{\ORccFive}$ of the graph of the \ThreeBit lattice jointly constitute the interior plane $\PL_{\times \i}^\F$. The four alternatives partition $\Reg^2$ (every pattern of $\ORccFive$ lies in exactly one of them), so the constraint excludes no outcome by itself. Its role is to fix which of the two interior planes $\PL_{\times \i}^{\T}$ and $\PL_{\times \i}^{\F}$ an outcome selects, which the predicate $\sharpDet$ and the map $\boxtimes$ use. 

All the constraints in Eq. \ref{eq:InteractionConstraints} have in common that they only depend on the argument $i\in\IdxSet$ and the sets $[\bot_i^{\downarrow}] \cup [\top_i^{\uparrow}]$, $[\bot_i^{\uparrow}] \cup [\top_i^{\downarrow}]$, $[\bot_{\ORccFive}] \cup [\bot_i]$, $[\top_{\ORccFive}] \cup [\top_i]$ have different extensions in $\Reg^2$ for every index $i\in\IdxSet$. As specified in Tab. \ref{tab:OverviewSupportSemantics}(2) the constraint $(\sharp_o^i r)\in [\bot_{\ORccFive}] \cup [\top_{\ORccFive}]$ also holds for every $i\in\IdxSet$. These specific constraints on the outcomes are crucial, as they ensure that interactions with entangled systems produce results that can be uniquely quantized into the corresponding non-orthogonal information states (the interior planes of the information cube), a process detailed in Section \ref{sec:QuantizationMap}.

\paragraph{The case $\vvReg$}

Elementary systems  $r\in \vvReg$ are taken to $\Reg^2$ by the maps $(\sharp_o^{\times i}\ r)$ and $(\sharp_o^{\times j}\ r)$ in ways that satisfy the following constraints: 
\begin{equation}
\label{eq:InteractionConstraintsXX}
    \begin{array}{l}
\text{for two distinct indexes}\ i,j\in\IdxSet\ \text{with}\ k\ \text{the remaining index, the maps}\ \sharp_o^{\times i},\sharp_o^{\times j}: \vvReg \to  \Reg^2\ \text{are defined, such that}\\
 \quad \text{for}\ l\in\{i,j\}:\ \big((\sharp_o^{\times l}\ r) \in \text{[}\bot_l^{\downarrow}] \cup [\top_l^{\uparrow}]\  \text{or}\ (\sharp_o^{\times l}\ r) \in \text{[}\bot_l^{\uparrow}] \cup [\top_l^{\downarrow}] \  \text{or}\\
 \qquad\qquad\quad (\sharp_o^{\times l}\ r) \in \text{[}\bot_{\ORccFive}] \cup [\bot_l]\ \text{or} \ (\sharp_o^{\times l}\ r) \in \text{[}\top_{\ORccFive}] \cup [\top_l]\big)\ \text{and} \\
   \qquad \{\llbracket\sharp_o^{\times i}\ r\rrbracket,  \llbracket\sharp_o^{\times j}\ r\rrbracket\} \in \DIA_k\ \text{where}\ \llbracket x \rrbracket \equiv (\rccfive\ x)\ \text{see Tab.~\ref{tab:rccFive}}.\\
    \end{array}
\end{equation}
\noindent That is, for an elementary system $r\in \vvReg$ and its two defined indexes $i$ and $j$, the outcomes of the interaction maps of the form $(\sharp_o^{\times i}\ r)$ and $(\sharp_o^{\times j}\ r)$ pick out diagonals across the interior of the \ThreeBit cube (Fig. \ref{fig:isuccBot_iprecTop} (bottom)) via  constraints of the form $\{\llbracket\sharp_o^{\times i}\ r\rrbracket,  \llbracket\sharp_o^{\times j}\ r\rrbracket\} \in\DIA_k$. As specified in Tab. \ref{tab:OverviewSupportSemantics}(2), the constraint $(\sharp_o^i r)\in [\bot_{\ORccFive}] \cup [\top_{\ORccFive}]$ also holds for every $i\in\IdxSet$. 

In what follows, the notation  $(\sharp_o^\alpha)$ abbreviates either $(\sharp_o^i)$ or $(\sharp_o^{\times i})$ for $i\in\IdxSet$ such that $\alpha\in\IdxSetNu$. The constraints on the interaction map $\sharp_o$  in Tab. \ref{tab:OverviewSupportSemantics}(2), Eq. \ref{eq:InteractionConstraints} and Eq. \ref{eq:InteractionConstraintsXX} and their summary in Tab. \ref{tab:sharpInducesIstate} are axioms in the formalized theory.

\subsection{Determinism and Information}
\label{sec:DeterminismAndInformation}

It follows from the constraints in Tab. \ref{tab:OverviewSupportSemantics}(2),  Eq. \ref{eq:InteractionConstraints} and Eq.
\ref{eq:InteractionConstraintsXX} that relational elementary systems $(o,r)\in(\OO\times\xyReg)$ must carry one bit of information, the bit extracted for the index $\alpha$, for maps of the form $r\in\xyReg\mapsto \llbracket\sharp_o^\alpha r\rrbracket \in \ORccFive$ with $\llbracket\sharp_o^\alpha r\rrbracket\equiv((\rccfive\circ\sharp_o^\alpha)\ r)$ to be determinate for a given index $\alpha\in\IdxSetNu$. A map is determinate if and only if it always yields the same output for the same input. Let $r\in\xyReg$ and $\alpha \in \IdxSetNu$ be fixed and assume that $(\sharp_o^\alpha r) \in S \cup T\subset \Reg^2$ with $S\neq \{\}, T \neq \{\},  S\cap T = \{\}$. If the relational elementary system does not carry the requisite bit of information, its state is one of ontological indeterminacy between the possibilities $S$ and $T$. An interaction will actualize one of these possibilities, but the formalism must be able to distinguish this case from one where the outcome was necessitated from the start.

To formally capture this distinction, the existence or non-existence of  a benchmark map, $\sharp\sharp_o^\alpha$, that processes the same information for $r$  as $\sharp_o^\alpha$ is probed  and a determinacy predicate for $\sharp_o^\alpha r$ is introduced:

\begin{definition}
\label{def:det_sharp}
For a given index $\alpha$ and relational elementary system $(o,r)$ with a potential outcome set $S\cup T$, an interaction $(\sharp_o^\alpha r)\in S\cup T$ is \textit{determinate} if there does not exist a valid benchmark map $(\sharp\sharp_o^\alpha r)\in S\cup T$ such that the outcomes of the two maps are incompatible: 
$$
\sharpDet(\sharp_o^\alpha r) = \left\{\begin{array}{cl} \alpha\in\IdxSet & \neg \exists (\sharp\sharp_o^i r).\ \llbracket\sharp_o^i r\rrbracket \neq \llbracket\sharp\sharp_o^i r\rrbracket\\
\left(\begin{array}{c}\alpha = \times i,\\ i\in\IdxSet\end{array}\right) &  \neg \exists (\sharp\sharp_o^{\times i}r).\  \llbracket\sharp_o^{\times i} r\rrbracket \neq \llbracket\sharp\sharp_o^{\times i}r\rrbracket\ \&\  \neg(\llbracket\sharp_o^{\times i} r\rrbracket \ll\gg_i \llbracket\sharp\sharp_o^{\times i}r\rrbracket) \end{array}\right.
$$
\noindent   The interaction is \textit{indeterminate} otherwise.
\end{definition}
\noindent This predicate is the formal switch that governs the subsequent quantization process (Sec. \ref{sec:QuantizationMap}): a determinate result will be mapped to a specific planar feature of the \ThreeBit information cube, while an indeterminate result will be mapped to cube as a whole ($\ORccFive$).

The predicate $\sharpDet$ evaluates the determinacy of the map $(\sharp_o^\alpha r)$ based on the (non-) existence of a benchmark map $(\sharp\sharp_o^\alpha r)$ that matches conditions depending on $\alpha$.
That is, an interaction (modeled by $(\sharp_o^\alpha r)$) is determinate if its outcome is necessary and fixed (with respect to the $\alpha$-based conditions), meaning no valid benchmark map could produce a different result and it is indeterminate otherwise. The purpose of this definition is to formally identify the cases where there is indeterminacy due to the absence of information, i.e. information not carried by the elementary system $r\in\xyReg$ on which the interaction map $\sharp_o^\alpha$ acts.

\begin{example}[Illustrating the \sharpDet\ Predicate]\normalfont
\label{ex:SharpDetPredicate}
The abstract determinacy predicate from Definition \ref{def:det_sharp} can be clarified by contrasting two scenarios: \textit{An Indeterminate Interaction:}
Consider a crisp-vague system $(o,r) \in \OO\times\cvReg$, which has an information form of $\IdxTwoBitPR$. The axioms state its interaction $(\sharp_o^i r)$ can result in a crisp proxy pair related by either $\DR$ or $\EQ$. Because two outcomes are possible, a valid benchmark map $(\sharp\sharp_o^i r)$ could exist that produces the alternate outcome (e.g., a pair related by $\EQ$ if $\sharp$ produced one related by $\DR$). Since such a conflicting benchmark map can exist, the condition for determinacy is not met, and the pattern postulated for this class states that the interaction is indeterminate ($\neg\sharpDet(\sharp_o^i r)$ for one $i \in \IdxSet$).
\textit{A Determinate Interaction:}
In contrast, consider a crisp system $(o,r) \in \OO\times\ccReg$, which has an information form of $\IdxThreeBitPR$. For such a system, the interaction map $\sharp_o^i$ is the identity map, and the outcome is fixed and necessary. No valid benchmark map $\sharp\sharp_o^i$ could exist that produces a different result. Therefore, the condition $\llbracket\sharp_o^i r\rrbracket \neq \llbracket\sharp\sharp_o^i r\rrbracket$ can never be met, and the pattern postulated for this class states that the interaction is determinate ($\sharpDet(\sharp_o^i r)$ for all $i \in \IdxSet$). \qed
\end{example}

\subsection{Classifying Elementary Systems Using the Interaction Maps}
\label{sec:RegOneBitAndRegTwoBitAndRegThreeBit}

This information-theoretic framework examines the vagueness of entities/regions indirectly, based on the form of information (\iform) carried by relational elementary systems. Information forms (\iforms) are the members of $\IdxSetX$ as enumerated in Tab.~\ref{tab:VectorsInIndexNotation} and Fig. \ref{fig:RccFivePartitionLattice}. The framework rests on an axiomatic structure that distinguishes between foundational postulates, formal definitions, and their logical consequences.

The core of this structure is the postulate that a relational elementary system's capacity to carry information is determined by the nature of its constituent regions (i.e., whether they are members of $\cReg$, $\vReg$, $\zReg$ or $\eReg$). This capacity is captured by its information form, $\iform(o,r)\in\IdxSetX$, which serves as the axiomatic basis for the theory. From this, formal definitions for the classes of relational elementary systems arise. The interaction behaviors of these systems, specifically the determinacy patterns of their interaction maps are then read off the postulated interaction patterns and stated as theorems.

This logical progression is detailed in Tables \ref{tab:ClassesDefAxThPR} and \ref{tab:ClassesDefAxThET}. These tables are organized into three columns to make the axiomatic structure explicit: (i) Classes of relational elementary systems: This column provides the formal \textit{definitions} of the various classes of systems $\OReg^2_\nu$ (e.g., $\OReg^2_\IdxTwoBitPR$). A system's membership in a class is defined by the information form, $\iform(o,r)=\IdxTwoBitPR$, that it carries. (ii) Postulate of  information form: This column states the foundational \textit{axioms} of the theory. Each postulate links the physical type of the constituent regions $\xyReg$ (e.g. members of $\cvReg \cup \vcReg \cup \czReg \cup \zcReg$) to the specific information form(s) that relational elementary systems composed of such regions can carry. Each postulate also fixes, through the column of Tab.~\ref{tab:InteractionPattern} (top) that bears the name of the form, the determinacy pattern of the interaction maps of such systems. (iii) Theorems stating the  entailed interaction properties: This column lists the logical \textit{consequences} that are read off the postulated patterns. These theorems describe the necessary and observable behaviors of the interaction maps for any system belonging to a given class, expressed using the determinacy predicate $\sharpDet$ from Def.~\ref{def:det_sharp}.

This structure ensures that the determinacy or indeterminacy of an interaction is not a brute fact, but a consequence of the interaction pattern that is postulated with a system's information form. The information form, $\iform(o,r)$, serves as the crucial link, connecting the axiomatic properties of a system to its derived, observable behaviors. For example, if a system's regions dictate that its $\iform(o,r) = \IdxOneBitPR$, it is by definition a member of the class $\OReg^2_\IdxOneBitPR$, and it is a theorem that exactly one of its simple interactions must be determinate.

\begin{remark}\normalfont
\label{remark:CategoriesOfRegs}
It is important to clarify the logical status of the components presented in Tables \ref{tab:ClassesDefAxThPR} and \ref{tab:ClassesDefAxThET}. The information forms (\iforms) and the resulting classification of relational elementary systems (e.g., $\OReg^2_\IdxOneBitPR$) are considered fundamental to the axiomatic structure of the theory. The categories of regions \cReg, \vReg, \zReg\ and \eReg, while essential to the model, serve three specific and distinct roles: First, they provide the necessary components to build models that ground the framework in classical mereology, demonstrating that Quantized Mereology is a consistent generalization.  Second, they illustrate that metaphysical consequences can be drawn from the formalism, showing how an underlying ontology of regions can give rise to the specific information-theoretic limits the theory describes. Third, the actual choice of \cReg, \vReg, \zReg, \eReg\ and their specific interpretation are intended as a coherent proof-of-concept; they are not meant to carry a final or exclusive metaphysical commitment. The formalism is, in principle, open to other physical groundings that could entail the same information forms. A deeper investigation into the definitive ontology of vague regions remains an important area for future work. \qed
\end{remark}

The sets $\RegOneBitPR$, $\RegTwoBitPR$, $\RegOneBitET$, $\RegTwoBitET$, $\RegThreeBitET$ and $\RegThreeBit$ are mutually exclusive. Jointly, all the relational elementary systems form the set $\OReg_\IdxOneToThreeBit^2$. The relatively fine-grained classification of relational elementary systems according to their information form can be coarsened by considering only their information capacity of one, two or three bits or their entanglement status.
\begin{equation}
\label{eq:AxiomORegExhaustive}
\begin{array}{rccl}
 \text{information} & \RegOneBit & =& \RegOneBitPR \cup \RegOneBitET\\
 \text{capacity} &  \RegTwoBit & = &\RegTwoBitPR \cup \RegTwoBitET\cup \RegThreeBitET\\
 & \OReg^2_\IdxThreeBit &  = & \OReg^2_\IdxThreeBitPR \\
 \hline
    \txt{entangle-} &  \ESPR & = & \RegOneBitPR \cup \RegTwoBitPR \cup \RegThreeBit \\
    \text{ment} & \ESET & = & \RegOneBitET \cup \RegTwoBitET \cup \RegThreeBitET \\\hline
    
\text{relational} & \OReg_\IdxOneToThreeBit^2 &=& \RegOneBit \cup \RegTwoBit \cup \OReg^2_\IdxThreeBitPR \\
\txt{elementary\\\quad systems} && = & \ESPR  \cup \ESET\\
\end{array}
\end{equation}
\noindent The classes that partition $\OReg_\IdxOneToThreeBit^2$ can be refined further by equivalence relations between elementary systems with identical maximum non-empty \istates\ \cite{bittner:InformationMereologyAndVagueness}.
The equivalence relations are indicated by the symbol $\simNu$ with $\nu\in\IdxSetX$ and occupy the fourth slot of mereological structures of the form $(\OReg\cup\{\emptyset\},\sim,\sharp,\simNu,\natural)$.

\begin{table}[h!]
\caption{\label{tab:ClassesDefAxThPR}Integrated View of Axiomatic System Classifications for non-entangled systems, Postulated Capacities, and Entailed Properties. (The entailed properties are read off the interaction patterns of Tab.~\ref{tab:InteractionPattern} (top), which are postulated with the information forms.)}
\centering\footnotesize
\begin{tabular}{|p{4.6cm}|p{4.6cm}|p{5.2cm}|}
\hline
\textbf{Classes of relational elementary systems} & \textbf{Postulate: Information Capacity} & \textbf{Theorem: Entailed Interaction Properties} \\
\hline \hline

\begin{center}$\begin{array}{l} \OReg^2_\IdxThreeBitPR = \{(o,r)\in \OO\times\\ \quad (\ccReg \cup \starEReg \cup \EstarReg \cup \eeReg).\\ \quad \iform(o,r)=\IdxThreeBitPR \}\end{array}$\end{center}  & 
Elementary systems in $ \ccReg \cup \starEReg \cup \EstarReg \cup \eeReg$ form  relational systems that carry information of the form $\IdxThreeBitPR$.
\begin{center}
$\begin{array}{l}\forall o \in \OO, \forall r \in \ccReg \cup \starEReg\\ \qquad \cup\ \EstarReg \cup \eeReg:\\ \qquad \iform(o,r) = \IdxThreeBitPR\end{array} $\end{center} & 
For any system in $\OReg^2_\IdxThreeBitPR$, all three simple interactions must be determinate.
\vspace{0.5pc}\begin{center}$ \forall (o,r) \in \OReg^2_\IdxThreeBitPR, \forall i \in \IdxSet: \sharpDet(\sharp_o^i r) $ \end{center}\\ 
\hline

\begin{center}$\begin{array}{l} \OReg^2_\IdxTwoBitPR = \{(o,r)\in \OO\times\\ \quad (\cvReg \cup \vcReg).\\ \quad \iform(o,r)=\IdxTwoBitPR \}\\\\
\OReg^2_\IdxOneBitPR = \{(o,r)\in \OO\times\\
\quad (\czReg \cup \zcReg).\\ \quad \iform(o,r)=\IdxOneBitPR\} \\
\end{array}$\end{center}  & 
Elementary systems in $\cvReg \cup \vcReg$ and $\czReg \cup \zcReg$ form  relational systems that respectively carry information of the form $\IdxTwoBitPR$ or $\IdxOneBitPR$.
\vspace{0.5pc}
\begin{center}
$\begin{array}{l}\forall o \in \OO, \forall r \in  \cvReg \cup \vcReg:\\
\qquad\iform(o,r) =\IdxTwoBitPR\\
\\
\forall o \in \OO, \forall r \in  \czReg \cup \zcReg:\\
\qquad\iform(o,r) =\IdxOneBitPR\end{array} $ \end{center}& 
For any system in $\OReg^2_\IdxTwoBitPR$, exactly one interaction must be indeterminate.
\begin{center}$ \forall (o,r) \in \OReg^2_\IdxTwoBitPR, \exists! i \in \IdxSet: \neg\sharpDet(\sharp_o^i r) $\end{center}
For any system in $\OReg^2_\IdxOneBitPR$, exactly one simple interaction must be determinate.
\begin{center}$ \forall (o,r) \in \RegOneBitPR, \exists! i \in \IdxSet: \sharpDet(\sharp_o^i r)  $\end{center}\\

\hline
\end{tabular}
\end{table}

\begin{table}[h!]
\caption{\label{tab:ClassesDefAxThET}Integrated View of Axiomatic System Classifications for Entangled Systems, Postulated Capacities, and Entailed Properties. (The entailed properties are read off the interaction patterns of Tab.~\ref{tab:InteractionPattern} (top), which are postulated with the information forms.)}
\centering\footnotesize
\begin{tabular}{|p{4.6cm}|p{4.4cm}|p{5.4cm}|}
\hline
\textbf{Classes of relational elementary systems} & \textbf{Postulate: Information Capacity} & \textbf{Theorem: Entailed Interaction Properties} \\
\hline \hline

\begin{center}$\begin{array}{l} \OReg^2_\IdxThreeBitET = \{(o,r)\in \OO\times\vvReg .\\ \quad \iform(o,r)=\IdxThreeBitET \}\end{array}$\end{center}  & 
Elementary systems in $ \vvReg $ form  relational systems that carry information of the form $\IdxThreeBitET$.
\begin{center}$\begin{array}{l}\forall o \in \OO, \forall r \in \vvReg:\\ \qquad\qquad \iform(o,r) = \IdxThreeBitET\end{array} $\end{center} & 
For any system in $\OReg^2_\IdxThreeBitET$, all simple interactions must be indeterminate, and for exactly one axis $k$ the two correlational interactions $\times i$ and $\times j$ with $\{i,j\}=\IdxSet\setminus\{k\}$ are determinate and their outcome patterns form a member of $\DIA_k$.
\begin{center}$ \begin{array}{l}\forall (o,r) \in \RegThreeBitET: (\forall i \in \IdxSet: \neg\sharpDet(\sharp_o^i r))\\
\land (\exists! k \in \IdxSet,\ \exists i,j \in \IdxSet:\\ \qquad \{i,j\} = \IdxSet\setminus\{k\}\ \land\\ \qquad \sharpDet(\sharp_o^{\times i} r) \land \sharpDet(\sharp_o^{\times j} r)\\ \qquad \land \{\llbracket \sharp^{\times i}_{o,r}\rrbracket, \llbracket \sharp^{\times j}_{o,r}\rrbracket\} \in \DIA_k) \end{array}$\end{center} \\
\hline
\begin{center}$\begin{array}{l} \OReg^2_\IdxTwoBitET = \{(o,r)\in \OO\times\\ \quad (\zvReg \cup \vzReg).\\ \quad \iform(o,r)=\IdxTwoBitET \}\\\\
\OReg^2_\IdxOneBitET = \{(o,r)\in \OO\times\\ \quad (\zvReg \cup \vzReg\\ \qquad \cup\ \zzReg).\\
\quad \iform(o,r)=\IdxOneBitET \} \\
\end{array}$\end{center}  & 
Elementary systems in $\zvReg \cup \vzReg \cup \zzReg$ form  relational systems that carry information of the form $\IdxTwoBitET$ or $\IdxOneBitET$.
\begin{center}$\begin{array}{l}\forall o \in \OO, \forall r \in  \zvReg \cup \vzReg:\\ \quad\iform(o,r) \in \{\IdxTwoBitET,\IdxOneBitET\}\\\\
\forall o \in \OO, \forall r \in  \zzReg:\\ \qquad\iform(o,r) = \IdxOneBitET \end{array} $\end{center} & 
For any system in $\OReg^2_\IdxTwoBitET$, exactly one simple interaction and exactly one correlational interaction must be determinate.
\begin{center}$ \begin{array}{l}\forall (o,r) \in \RegTwoBitET: (\exists! i \in \IdxSet: \sharpDet(\sharp_o^i r))\\
\qquad \land (\exists! k \in \IdxSet: \sharpDet(\sharp_o^{\times k} r)) \end{array}$\end{center}

For any system in $\OReg^2_\IdxOneBitET$,  all simple interactions must be indeterminate, and exactly one correlational interaction must be determinate.
\begin{center}$ \begin{array}{l} \forall (o,r) \in \RegOneBitET: (\forall i \in \IdxSet: \neg\sharpDet(\sharp_o^i r))\\ \qquad \land (\exists! k \in \IdxSet: \sharpDet(\sharp_o^{\times k} r))\end{array}  $\end{center}\\
\hline
\end{tabular}
\end{table}

\begin{example}[Classifying the elementary Systems of the Running Example]\normalfont
\label{ex:ClassifyingSystems}
The abstract classes of elementary systems defined in this section can be made concrete by applying them to the scenarios from the running example (Example~\ref{ex:RunningExample}).

\textit{The Soil Chemist's System ($\in \RegTwoBitET$):} With $\AR\in\zReg$ (a stipulation for this example) and $\FH\in\vReg$ (Example~\ref{ex:SharpMap}), the pair $(\AR,\FH)$ is a member of $\zvReg$, so by the postulate of Tab.~\ref{tab:ClassesDefAxThET} its information form is $\IdxTwoBitET$ or $\IdxOneBitET$, depending on the observer. After the Soil Chemist $O_1$ determines that the regions overlap, the elementary system $(\AR, \FH)$ is in the state $\PL_{\emptyset12}^{\T--}$ of orthogonal vagueness. The ordered pair $(O_1,(\AR,\FH))$ belongs to the set $\RegTwoBitET$, the only form of a $\zvReg$ pair in which a simple interaction is determinate. Formally, for $i=\emptyset$, the interaction is determinate ($\sharpDet(\llbracket\sharp^i_{O_1,r}\rrbracket)$), while for the other two questions ($j=1, k=2$), the interactions are indeterminate ($\neg\sharpDet(\llbracket\sharp^j_{O_1,r}\rrbracket)$ and $\neg\sharpDet(\llbracket\sharp^k_{O_1,r}\rrbracket)$); the correlational interaction $\sharp^{\times\emptyset}_{O_1,r}$ is determinate as well (Tab.~\ref{tab:InteractionPattern}), but $O_1$ did not probe it. The state $\PL_{\emptyset12}^{\T--}$ is the maximum supporting state of the overlap answer alone (Example~\ref{ex:SupportRelation}); the maximum \istate\ of the system, $\PL_\emptyset^\T\cap\PL_{\times\emptyset}^{\b}$ with $\b$ the parity extracted by $\sharp^{\times\emptyset}_{O_1,r}$, carries 2 bits.

\textit{The Remote Sensing Specialist's System ($\in \RegOneBitET$):} After the specialist $O_2$ determines that the parthood relations are the same, the system $(\AR, \FH)$ is in a state of non-orthogonal vagueness. The ordered pair $(O_2,(\AR,\FH))$ belongs to the set $\RegOneBitET$. In this case, no single elementary question can yield a determinate answer ($\neg\sharpDet(\llbracket\sharp^i_{O_2,r}\rrbracket)$ for all $i \in \{\emptyset, 1, 2\}$). However, one bit of information is obtained from a non-cumulative, correlational question (e.g., $Q_{\times k}?$), meaning its corresponding interaction map is determinate ($\sharpDet(\llbracket \sharp^{\times k}_{O_2,r}\rrbracket)$ behaves as described in Tab.~\ref{tab:ClassesDefAxThET}).

\textit{A Crisp System ($\in \ccReg$):} For comparison, an elementary system composed of two crisply-defined land parcels observed by $O$ would belong to the set $\RegThreeBit$. For such a system, all three elementary questions have determinate answers, meaning $\sharpDet(\llbracket\sharp^i_{O,r}\rrbracket)$ for all $i \in \{\emptyset, 1, 2\}$.
\qed
\end{example}

\begin{table}[h!]
\caption{\label{tab:InteractionPattern} Interaction pattern (top) and their encoding in type  declarations (bottom). The interaction patterns are postulated with the information forms (second column of Tabs.~\ref{tab:ClassesDefAxThPR}~and~\ref{tab:ClassesDefAxThET}). The entailed properties in the third column of those tables are read off them.
}
$$
\begin{array}{c|c|c|c|c|c|c}
& \multicolumn{6}{c}{\text{Interaction patterns}}\\
\sharpDet(\llbracket \sharp^\alpha_{o,r}\rrbracket)  & \IdxOneBitPR & \IdxTwoBitPR & \IdxThreeBitPR & \IdxOneBitET & \IdxTwoBitET & \IdxThreeBitET\\\hline
\sharpDet(\llbracket \sharp^i_{o,r}\rrbracket) & \T & \T & \T & \F & \T & \F \\
\sharpDet(\llbracket \sharp^j_{o,r}\rrbracket)  &  \F & \T & \T & \F & \F & \F \\
\sharpDet(\llbracket \sharp^k_{o,r}\rrbracket)  &  \F & \F & \T & \F & \F & \F \\
\sharpDet(\llbracket \sharp^{\times i}_{o,r}\rrbracket)  &  - & - & - & \T & \T & \T \\
\sharpDet(\llbracket \sharp^{\times j}_{o,r}\rrbracket)   &  - & - & -& - & - & \T \\
 \txt{$\{\llbracket \sharp^{\times i}_{o,r}\rrbracket, \llbracket \sharp^{\times j}_{o,r}\rrbracket\}$\\$ \in \DIA_k$}&  - & - & -& - & - & \T \\
\hline
& \multicolumn{6}{c}{\text{Type of (quantization map $\circ$ interaction map) (i.e. $\overline{\lfloor \cdot\rfloor}\circ(\overline{\rccfive}\circ\overline{\sharp})$)}}\\
r_o \in \OReg^2_\nu & \OReg^2_\IdxOneBitPR & \OReg^2_\IdxTwoBitPR & \OReg^2_\IdxThreeBitPR & \OReg^2_\IdxOneBitET & \OReg^2_\IdxTwoBitET & \OReg^2_\IdxThreeBitET\\
\txt{\rotatebox[origin=c]{-90}{$\to$}} & \txt{\rotatebox[origin=c]{-90}{$\to$}} & \txt{\rotatebox[origin=c]{-90}{$\to$}} & \txt{\rotatebox[origin=c]{-90}{$\to$}} & \txt{\rotatebox[origin=c]{-90}{$\to$}} & \txt{\rotatebox[origin=c]{-90}{$\to$}} & \txt{\rotatebox[origin=c]{-90}{$\to$}} \\
\txt{$\overline{\rccfive}\circ\overline{\sharp}(r_o)\in$\\$((\ORccFive\times\TF)_\alpha)_{\alpha\in\IdxSetNu}$} & \multicolumn{3}{c|}{((\ORccFive\times\TF)_\alpha)_{\alpha\in\IdxSet}} & \multicolumn{2}{c|}{((\ORccFive\times\TF)_\alpha)_{\alpha\in\IdxSet\cup\{\times i\}}} & \txt{$((\ORccFive\times\TF)_\alpha)$\\$\scriptstyle{\alpha\in\IdxSet\cup\{\times i,\times j\}}$}\\
\txt{\rotatebox[origin=c]{-90}{$\to$}} & \txt{\rotatebox[origin=c]{-90}{$\to$}} & \txt{\rotatebox[origin=c]{-90}{$\to$}} & \txt{\rotatebox[origin=c]{-90}{$\to$}} & \txt{\rotatebox[origin=c]{-90}{$\to$}} & \txt{\rotatebox[origin=c]{-90}{$\to$}} & \txt{\rotatebox[origin=c]{-90}{$\to$}} \\
\txt{$\overline{\lfloor \cdot\rfloor}\circ\overline{\rccfive}\circ\overline{\sharp}(r_o)$\\$ \in \ORccFive_\nu$} & \ORccFive_\IdxOneBitPR & \ORccFive_\IdxTwoBitPR & \ORccFive_\IdxThreeBitPR & \ORccFive_\IdxOneBitET & \ORccFive_\IdxTwoBitET & \ORccFive_\IdxThreeBitET\\\cline{2-7}
& \multicolumn{6}{c}{\overline{\lfloor \cdot\rfloor}\circ(\overline{\rccfive}\circ\overline{\sharp})\in\lfloor\QQ_{\nu}?  =  \A^{\nu}\rfloor =  \OReg^2_\nu\to\ORccFive_\nu\quad \text{(cf. Tabs. \ref{tab:enumerationOfmapsFromQaPattren}, \ref{tab:Quantization})}}\\
\end{array}
$$
\end{table}

\subsection{Map Types Revisited}

This concludes the discussion of how the map $\sharp^\alpha$   induces families/types of maps of the form $\overline{\sharp}:\OReg^2_\nu \to ((\Reg^2)_\alpha)_{\alpha\in\IdxSetNu}$ with $\nu \in \IdxSetX$ (\iforms) and $\nu$-dependent index sets $\IdxSetNu$ of the form
\begin{equation}
\label{eq:IdxSetNu}
\begin{array}{c||c|c|c|c|c|c}
\nu & \IdxOneBitPR & \IdxTwoBitPR & \IdxThreeBitPR & \IdxOneBitET & \IdxTwoBitET & \IdxThreeBitET\\\hline
\IdxSetNu & \IdxSet & \IdxSet & \IdxSet & \IdxSet\cup\{\times i\} &  \IdxSet\cup\{\times i\} & \IdxSet\cup\{\times i, \times j\} 
\end{array}\quad \text{with}\ i\neq j \in \IdxSet
\end{equation}
\noindent corresponding to the bottom row of Tab. \ref{tab:sharpInducesIstate}.
In the joint context of the axioms in Tabs.~\ref{tab:sharpInducesIstate},~\ref{tab:ClassesDefAxThPR} and~\ref{tab:ClassesDefAxThET} then maps of the form $(\overline{\rccfive}\circ\overline{\sharp}):\OReg^2_\nu \to ((\ORccFive\times\TF)_\alpha)_{\alpha\in\IdxSetNu}$  arise as enumerated in the bottom of Tab. \ref{tab:InteractionPattern}. Each map is characterized by one of the \textit{interaction patterns} that are enumerated in the top of Tab. \ref{tab:InteractionPattern}. The top part of the table shows that interaction patterns arise jointly from the constraints on the interaction maps (Tab. \ref{tab:sharpInducesIstate}) and the interrelations between determinism and information that were captured in the truth and falsehood of expressions of the form $\sharpDet(\llbracket \sharp^{\alpha}_{o,r}\rrbracket)$ for the indexes $\alpha\in\IdxSetNu$ (Eq. \ref{eq:IdxSetNu}) in the definitions of information-theoretic kinds of elementary systems.

Maps of the form $(\overline{\rccfive}\circ\overline{\sharp})$ are parts of maps of the form $\overline{\lfloor \cdot\rfloor}\circ(\overline{\rccfive}\circ\overline{\sharp})\equiv \lfloor\Q_{\i\j\k}?=\A^{\alpha\beta\gamma}\rfloor\in\lfloor\Q_\nu?=\A^\nu\rfloor=\OReg^2_\nu\to\ORccFive_\nu$ (cf. Tab. \ref{tab:enumerationOfmapsFromQaPattren}). These maps constitute the left part of the equation that expresses the QIS-principle of Eq. \ref{eq:MeaningOfMereoQuantification} which was visualized using the diagram in Fig. \ref{fig:MereologyVSinformation} (right). The aim in what follows is to establish  a one-one correspondence between information-theoretic kinds of elementary systems ($\OReg^2_\nu$), the interaction patterns that are enumerated in the top of Tab. \ref{tab:InteractionPattern} and types of maps of the form $(\OReg^2_\nu \to ((\ORccFive\times\TF)_\alpha)_{\alpha\in\IdxSetNu})\to \ORccFive_\nu$  that are enumerated in the bottom of Tab. \ref{tab:InteractionPattern}. The maps arise as the composition of interaction maps of the form $(\overline{\rccfive}\circ\overline{\sharp})$ that were discussed in this section and the quantization maps signified by $\overline{\lfloor\cdot\rfloor}:((\ORccFive\times\TF)_\alpha)_{\alpha\in\IdxSetNu}\to \ORccFive_\nu$ that are discussed in the next section.

\section{The Quantization Map}
\label{sec:QuantizationMap}

For relational elementary systems with pairs of regions in $\OReg^2\supset \cReg^2$ the quantization map $\overline{\lfloor\cdot\rfloor}$ takes configurations of classical bits in $((\ORccFive\times\TF)_i)_{i\in\IdxSetNu}$ to qubits that are represented as cells in the set $\ORccFive_\nu\subset\bigcup \Pi(\ORccFive)$ that arises from the \ThreeBit partition lattice $(\Pi(\ORccFive),\sqcap,\sqcup)$. This indexed family of classical bit patterns and their determinacy flags is the direct output of the preceding $(\overline{\rccfive}\circ\overline{\sharp})$ stage of the indirect quantization path. As illustrated on the right of Fig. \ref{fig:MereologyVSinformation} the quantization map $\overline{\lfloor\cdot\rfloor}$ has a compositional structure that is signified as $(\boxplus\circ\boxtimes): ((\ORccFive \times \TF)_i)_{i\in\IdxSetNu} \to \ORccFive_\nu$. To analyze this map, three cases are distinguished: quantization processes that result in representations of non-entangled information states (which include the classical \istates), quantization processes that result in representations of partially entangled \istates, and quantization processes that result in representations of maximally entangled \istates.

The sub-map $\boxplus$  completes the quantization process. Its function is uniform across all cases: it takes the family of non-trivial 1-bit partition cells (of the form $\PL^\a_\alpha$ with $\alpha\in\{\times\i, \i\}$) and trivial 0-bit partition cells (\ORccFive) (i.e. the output of the preceding $\boxtimes$ map) and combines them using set intersection ($\cap$). This projection merges the individual pieces of information, whether determinate or correlational or indeterminate, into a single, final partition cell that represents the system's complete \istate. The operator patterns in Table~\ref{tab:VectorsInIndexNotation} provide a formal verification of this construction; they demonstrate that applying the universal principle of intersection to the intermediate states correctly yields the final \istates\ as defined, thus ensuring the QIS principle's consistency. The following subsections detail how this principle applies to construct non-entangled and entangled states from their respective intermediate components.

\subsection{Non-entangled Quantum \istates}

The notation $\boxtimes$ is used to signify maps of the form
\begin{equation}
\label{eq:IndexToPlanePR}\arraycolsep=1pt
\begin{array}{rcl}
\boxtimes : ((\ORccFive\times\TF)_i)_{i\in\IdxSet} & \to & ((\ORccFive_{\IdxOneBitPR}\cup\{\ORccFive\})_i)_{i\in\IdxSet} \\
    (\llbracket\sharp^i_{o,r}\rrbracket,\sharpDet(\llbracket\sharp^i_{o,r}\rrbracket))_i\in \ORccFive\times\TF & \mapsto & \left\{ \begin{array}{cl} \PL_\i^\F \cup \PL_\i^\T & \text{if}\ \neg\sharpDet(\llbracket\sharp^i_{o,r}\rrbracket) \\
   \PL_\i^\F & \text{if}\ \sharpDet(\llbracket\sharp^i_{o,r}\rrbracket)\ \text{and}\ \llbracket\sharp^i_{o,r}\rrbracket \in \PL_\i^\F \\
   \PL_\i^\T & \text{if}\ \sharpDet(\llbracket\sharp^i_{o,r}\rrbracket)\ \text{and}\ \llbracket\sharp^i_{o,r}\rrbracket \in \PL_\i^\T\\
    \end{array} \right. \\
\end{array}
\end{equation}
That is, the sub-map $\boxtimes$ of the quantization map $(\boxplus\circ\boxtimes)$ takes indexed sets from $((\ORccFive\times\TF)_i)_{i\in\IdxSet}$ to cells of the \ThreeBit partition lattice in $\ORccFive_\IdxOneBitPR\cup\{\ORccFive\}$. One can verify that if $\sharpDet(\llbracket\sharp^i_{o,r}\rrbracket)$, then the cells $((\boxtimes\circ\rccfive_i)\, (o,r))\in \ORccFive_\IdxOneBitPR$ are equal to the truth sets of formulas of the form $|\oQ_i(r_o)|$ as specified in Tabs.~\ref{tab:OverviewSupportSemantics}(4),\ref{tab:OverviewSupportSemanticsTwo}(4). In virtue of being cells in the \ThreeBit partition lattice, members of $\ORccFive_{\IdxOneBitPR}$ carry one bit of information. Similarly, if $\neg\sharpDet(\llbracket\sharp^i_{o,r}\rrbracket)$, then the partition cell $((\boxtimes\circ\rccfive_i)\, (o,r))=\ORccFive$ carries zero bits of information. See Tab. \ref{tab:Quantization}.

\begin{example}[Quantizing a 1-Bit Product State]\normalfont
\label{ex:QuantizingOneBitPR}
To illustrate the quantization process, consider a relational elementary system $(o,r) \in \OReg^2_\IdxOneBitPR$. As per the theorems in Table~\ref{tab:ClassesDefAxThPR}, an interaction with such a system yields one determinate and two indeterminate outcomes. Assume the classical results and determinacy flags are $(\llbracket\sharp^i_{o,r}\rrbracket, \T)$, $(\llbracket\sharp^j_{o,r}\rrbracket, \F)$, and $(\llbracket\sharp^k_{o,r}\rrbracket, \F)$, where $\llbracket\sharp^i_{o,r}\rrbracket \in \PL_\i^\a$.
The quantization map $\overline{\lfloor\cdot\rfloor} = (\boxplus\circ\boxtimes)$ processes this input as follows:
($\boxtimes$) The map is applied to each pair individually. The determinate result is mapped to its corresponding partition cell $\PL_\i^\a$ (a face), while the indeterminate results are mapped to the  trivial partition cell (\ORccFive): $ (\PL_\i^\a, \ORccFive, \ORccFive) $.
    ($\boxplus$) The map then applies the universal principle of intersection to this family of partial states:
    $\boxplus(\PL_\i^\a, \ORccFive, \ORccFive) = \PL_\i^\a \cap \ORccFive \cap \ORccFive = \PL_{\,\i\,\j\,\k}^{\a--}$.
The final quantized \istate\ is $\PL_{\,\i\,\j\,\k}^{\a--}$. Geometrically, this process corresponds to selecting a single face of the information cube, which correctly represents the one bit of information gained from the determinate interaction. \qed
\end{example}

\subsection{Entangled \istates}

\paragraph{Partially entangled \istates}

The map $(\overline{\rccfive}\circ\overline{\sharp})(o,r)$ with $(o,r)\in \OReg^2_\nu$ and $\nu\in\{\IdxOneBitET,\IdxTwoBitET\}$  is partially enumerated in the left column of Tab. \ref{tab:enumBoxPlusX}. In this table, cells with two entries of the form `$\{(\bot_\ORccFive,\F), (\top_\ORccFive,\F)\}$' signify that either $(\llbracket\sharp^i_{o,r}\rrbracket,\sharpDet(\llbracket\sharp_{o,r}^i\rrbracket))_i=(\bot_\ORccFive,\F)$ or $(\llbracket\sharp^i_{o,r}\rrbracket,\sharpDet(\llbracket\sharp_{o,r}^i\rrbracket))_i=(\top_\ORccFive,\F)$ and that neither possibility can be excluded given the information that is possible in the context of the current row and the information capacity of the relational elementary system $(o,r)\in\OReg^2_\nu$. Similarly, cells with four entries of the form `$\{(\bot_\ORccFive,\T), (\bot_i,\T),(\top_\ORccFive,\T), (\top_i,\T)\}$' signify that $(\llbracket\sharp^{\times i}_{o,r}\rrbracket,\sharpDet(\llbracket\sharp_{o,r}^{\times i}\rrbracket))_{\times i}$ takes the value of one of the four set elements and that (due to entanglement)  none of the four possibilities can be excluded given the information that is possible in the context of the current row and the information capacity of the relational elementary system $(o,r)\in\OReg^2_\nu$.

The results of the remaining sub-map $(\boxplus\circ\boxtimes): ((\ORccFive \times \TF)_\alpha)_{\alpha\in\{i,\times i\}} \to \ORccFive_\nu$ of the quantization map $\overline{\lfloor\cdot\rfloor}$  is partially enumerated in the right most column of Tab. \ref{tab:enumBoxPlusX}. The map $\boxtimes$ is defined as: 
\begin{equation}\arraycolsep=1pt
\label{eq:IndexToPlaneX}
\begin{array}{rcl}
\boxtimes  :  ((\ORccFive \times \TF)_\alpha)_{\alpha\in\{i,\times i\}} & \to &  ((\ORccFive_{\IdxOneBit}\cup\{\ORccFive\})_\alpha)_{\alpha\in\{\times i,i\}}\\
    (\llbracket\sharp^{\times i}_{o,r}\rrbracket,\sharpDet(\llbracket\sharp^{\times i}_{o,r}\rrbracket))_{\times i}\in (\ORccFive \times\{\T\}) & \mapsto & \left\{ \begin{array}{cl}  
   \PL_{\times \i}^\T & \text{if}\ \sharpDet(\llbracket\sharp^{\times i}_{o,r}\rrbracket)\
   \text{and}\\
   & \llbracket\sharp^{\times i}_{o,r}\rrbracket \in  \{\bot_i^{\downarrow}, \top_i^{\uparrow}, \bot_i^{\uparrow}, \top_i^{\downarrow} \}\\
     \PL_{\times \i}^\F & \text{if}\  \sharpDet(\llbracket\sharp^{\times i}_{o,r}\rrbracket)\ \text{and}\\
     & \llbracket\sharp^{\times i}_{o,r}\rrbracket \in  \{\bot_\ORccFive, \bot_i, \top_\ORccFive, \top_i \} \\
    \end{array} \right.  \\
   (\llbracket\sharp^i_{o,r}\rrbracket,\sharpDet(\llbracket\sharp^i_{o,r}\rrbracket))_i\in \{\bot_\ORccFive,\top_\ORccFive\}\times\TF & \mapsto & \left\{ \begin{array}{cl} 
   \PL_\i^\F & \text{if}\ \sharpDet(\llbracket\sharp^i_{o,r}\rrbracket)\ \text{and}\ \llbracket\sharp^i_{o,r}\rrbracket \in  \PL_\i^\F \\
   \PL_\i^\T & \text{if}\  \sharpDet(\llbracket\sharp^i_{o,r}\rrbracket)\ \text{and}\ \llbracket\sharp^i_{o,r}\rrbracket \in  \PL_\i^\T\\
    \PL_\i^\F \cup \PL_\i^\T  & \text{if}\  \neg\sharpDet(\llbracket\sharp^i_{o,r}\rrbracket)  \\
    \end{array} \right.\\
\end{array}
\end{equation}
\noindent The sets  $\PL_{\times \i}^\F=\PL_\j^\b\vartriangle\PL_\k^{\sim\b}=\PL_{\j\k}^{\overline{\b\b}}= \{\bot_\ORccFive, \bot_i, \top_\ORccFive, \top_i \}$ and $\PL_{\times \i}^\T=\PL_\j^\b\vartriangle\PL_\k^{\b}=\PL_{\j\k}^{\overline{\b\,\sim\b}}=\{\bot_i^{\downarrow}, \top_i^{\uparrow}, \bot_i^{\uparrow}, \top_i^{\downarrow} \}$ with $\b\in\TF$ pick out the (pairwise orthogonal) interior planes $\PL_{\times \i}^\T$ and $\PL_{\times \i}^\F$ of the \ThreeBit cube as specified in Fig. \ref{fig:PairsOfPlanes} (bottom). 
These sub-spaces are also partition cells in $\ORccFive_{\IdxOneBitET}$ that carry one bit of information in virtue of being cells in the \ThreeBit partition lattice.

The sub-map $\boxplus$ then projects the cell  $\PL_{\times \i}^{\a}$  onto $\PL_\i^\b$ or $\PL_\i^\a\cup\PL_\i^{\sim\a}$ using the intersection operator $\cap$. That is, the sub-map $\boxplus$ takes indexed sets of partition cells in $((\ORccFive_{\IdxOneBit}\cup\{\ORccFive\})_\alpha)_{\alpha\in\{\times i,i\}}$ to partition cells in $\ORccFive_\nu$ with $\nu\in\{\IdxOneBitET,\IdxTwoBitET\}$. The various patterns in the second right most column of Tab. \ref{tab:enumBoxPlusX} correspond to the ways in which $\boxplus$ corresponds to the definitions of the index notation in Tab \ref{tab:VectorsInIndexNotation}.

\begin{example}[Quantizing a 2-Bit Partially Entangled State]\normalfont
\label{ex:QuantizingTwoBitET}
Consider a relational elementary system $(o,r) \in \OReg^2_\IdxTwoBitET$. According to the theorem in Table~\ref{tab:ClassesDefAxThET}, an interaction with such a system yields exactly one determinate simple outcome (say, for index $\k$) and one determinate correlational outcome (say, for index $\times \k$). The quantization map processes this as follows:
($\boxtimes$) The map is applied to the results from the $(\overline{\rccfive}\circ\overline{\sharp})$ stage. The determinate simple interaction yields a  1-bit partition cell (a face, $\PL_k^\c$), while the determinate correlational interaction yields a  1-bit entangled partition cell (an interior plane, $\PL_{\ \ \times \k}^{\a\xor\b}$). All other indeterminate interactions map to the trivial 0-bit partition cell. The resulting family of intermediate partial states is:
    $ (\PL_\k^\c, \PL_{\ \ \times \k}^{\a\xor\b}, \ORccFive, \dots) $. 
    ($\boxplus$) The map then applies the universal principle of intersection to the  intermediate states generated by the $\boxtimes$ map:
   $\boxplus(\PL_\k^\c, \PL_{\ \ \times \k}^{\a\xor\b},\ORccFive, \dots) = \PL_\k^\c \cap \PL_{\ \ \times \k}^{\a\xor\b} \cap \ORccFive \cap \dots = \PL_{\,\i\,\j\,\k}^{\overline{\a\b}\c}$.
The final quantized \istate\ is $\PL_{\,\i\,\j\,\k}^{\overline{\a\b}\c}$, a 2-bit state that correctly combines one bit of determinate (classical) information with one bit of correlational (entangled) information. \qed
\end{example}

\paragraph{Maximally entangled \istates}

For elementary systems in $\OReg^2_\IdxThreeBitET$ the sub-map  $\boxtimes$ of the quantization map $\overline{\lfloor\cdot\rfloor}=(\boxplus\circ\boxtimes)$ is defined as  
\begin{equation}
\label{eq:IndexToPlaneXX}
\begin{array}{rcl}
\boxtimes  :  ((\ORccFive \times \TF)_\alpha)_{\alpha\in\{\times i,\times j\}} & \to &  ((\ORccFive_{\IdxOneBitET})_\alpha)_{\alpha\in\{\times i,\times j\}}\\
    (\llbracket\sharp^{\alpha}_{o,r}\rrbracket,\sharpDet(\llbracket\sharp^{\alpha}_{o,r}\rrbracket))_{\alpha}\in (\ORccFive \times\{\T\}) & \mapsto & \left\{ \begin{array}{cl}  
   \PL_{\times\k}^\T & \text{if}\ \sharpDet(\llbracket\sharp^{\times k}_{o,r}\rrbracket)\
   \text{and}\ \llbracket\sharp^{\times k}_{o,r}\rrbracket \in  \{\bot_k^{\downarrow}, \top_k^{\uparrow}, \bot_k^{\uparrow}, \top_k^{\downarrow} \}\\
     \PL_{\times\k}^\F & \text{if}\  \sharpDet(\llbracket\sharp^{\times k}_{o,r}\rrbracket)\ \text{and}\ \llbracket\sharp^{\times k}_{o,r}\rrbracket \in  \{\bot_\ORccFive, \bot_k, \top_\ORccFive, \top_k \} \\
    \end{array} \right.  \\
    && \qquad \text{if}\  \alpha = \times k\ \text{with}\ k\in\{i,j\}\\
\end{array}
\end{equation}
\noindent  The  sub-map $\boxplus$ of $(\boxplus\circ\boxtimes)$ takes  the partition cells of the form  $\PL_{\times\k}^{\a\xor\b}=(\PL_\i^\a\vartriangle\PL_\j^{\sim\b})$ and $\PL_{\times\i}^{\b\xor \c}=(\PL_\j^{\b}\vartriangle\PL_\k^{\sim\c})$ in $\ORccFive_\IdxOneBitET$ to the partition cell $(\PL_\i^\a\vartriangle\PL_\j^{\sim\b})\cap (\PL_\j^{\b}\vartriangle\PL_\k^{\sim\c})$ in the set $\ORccFive_\IdxThreeBitET$ in ways that correspond to the definitions of the index notation in Tab. \ref{tab:VectorsInIndexNotation}. 

\begin{example}[Quantizing a 2-Bit Maximally Entangled State]\normalfont
\label{ex:QuantizingThreeBitET}
Consider a system $(o,r) \in \OReg^2_\IdxThreeBitET$. Per the theorem in Table~\ref{tab:ClassesDefAxThET}, this system yields two determinate correlational interactions (say, for indices $\times k$ and $\times i$) and no determinate simple interactions. ($\boxtimes$) The map processes the two determinate correlational results, yielding two distinct 1-bit entangled partition cells (two interior planes, $\PL_{\ \ \times \k}^{\a\xor\b}$ and $\PL_{\ \ \times \i}^{\b\xor\c}$). All indeterminate simple interactions map to the trivial 0-bit partition cell. The intermediate family of states is:
    $ (\PL_{\ \ \times \k}^{\a\xor\b}, \PL_{\ \ \times \i}^{\b\xor\c}, \ORccFive, \dots) $.
    ($\boxplus$) The map intersects these cells, applying the universal principle of intersection:
    $\boxplus(\PL_{\ \ \times \k}^{\a\xor\b}, \PL_{\ \ \times \i}^{\b\xor\c},\ORccFive, \dots) = \PL_{\ \ \times \k}^{\a\xor\b} \cap \PL_{\ \ \times \i}^{\b\xor\c} \cap \ORccFive \cap \dots = \PL_{\,\i\,\hat{\j}\,\k}^{\overline{\a\b\c}}$.
The final result is the 2-bit maximally entangled \istate\ $\PL_{\,\i\,\hat{\j}\,\k}^{\overline{\a\b\c}}$. Geometrically, this corresponds to the intersection of two interior planes, resulting in a diagonal line passing through the cube's interior (e.g., the set $\{\DR, \EQ\}$). \qed
\end{example}

\begin{table}
\caption{\label{tab:Quantization}The quantization map $(\boxplus\circ\boxtimes): ((\ORccFive \times \TF)_\alpha)_{\alpha\in\IdxSetNu} \to \ORccFive_\nu$ for relational elementary systems $(o,r)\in\OReg^2_\nu$ with $\nu\in\IdxSetX$ in the context of the combined interaction- and quantization maps of the form $(\boxplus\circ\boxtimes\circ\rccfive\circ\overline\sharp) \equiv \lfloor\Q_{\i\j\k}?=\A^{\alpha\beta\gamma}\rfloor\in \lfloor\Q_\nu?=\A^\nu\rfloor:\OReg^2_\nu \to \ORccFive_\nu$.
}
$$ \arraycolsep=1pt
\begin{array}{c|c|c|c|c|c|c}
& \multicolumn{6}{c}{r_o \in \OReg^2_\nu}\\
\text{maps} & \OReg_{\IdxOneBitPR}^2 & \OReg_{\IdxTwoBitPR}^2 & \OReg_{\IdxThreeBitPR}^2 & \OReg_{\IdxOneBitET}^2 & \OReg_{\IdxTwoBitET}^2 & \OReg_{\IdxThreeBitET}^2 \\\hline
&&&&\multicolumn{2}{c|}{}\\
(\boxtimes(\llbracket\sharp_{o,r}^i\rrbracket))_i = & \txt{$\PL_\i^\a=$\\$|\oQ_i(r_o)|$} & \txt{$\PL_\i^\a=$\\$|\oQ_i(r_o)|$}  & \txt{$\PL_\i^\a=$\\$|\oQ_i(r_o)|$}  & \multicolumn{2}{c|}{\txt{$\ORccFive =$\\$\PL_\i^\F\cup\PL_\i^\T$}} &  \ORccFive\\
(\boxtimes(\llbracket\sharp_{o,r}^j\rrbracket))_j =& \ORccFive & \txt{$\PL_\j^\b=$\\$|\oQ_j(r_o)|$} & \txt{$\PL_\j^\b=$\\$|\oQ_j(r_o)|$}  & \multicolumn{2}{c|}{\txt{$\ORccFive =$\\$\PL_\j^\F\cup\PL_\j^\T$}} & \ORccFive \\
(\boxtimes(\llbracket\sharp_{o,r}^k\rrbracket))_k = & \ORccFive & \ORccFive & \txt{$\PL_\k^\c=$\\$|\oQ_k(r_o)|$} & \ORccFive & \txt{$\PL_\k^\c=$\\$|\oQ_k(r_o)|$} & \ORccFive \\

(\boxtimes(\llbracket\sharp_{o,r}^{\times k}\rrbracket))_{\times k} \in & \multicolumn{3}{c|}{-} & \multicolumn{3}{c}{\txt{$\{ \{\bot_k^{\downarrow}, \top_k^{\uparrow}, \bot_k^{\uparrow}, \top_k^{\downarrow} \}, \{\bot_\ORccFive, \bot_k, \top_\ORccFive, \top_k \}\}$\\$= \{\PL_{\times\k}^\T,\PL_{\times\k}^{\F}\} =\{\PL_\i^\b\vartriangle\PL_\j^\b,\PL_\i^\b\vartriangle\PL_\j^{\sim\b}\}$}}  \\

(\boxtimes(\llbracket\sharp_{o,r}^{\times i}\rrbracket))_{\times i}\in & \multicolumn{3}{c|}{-} & \multicolumn{2}{c|}{-} & \txt{$\{ \{\bot_i^{\downarrow}, \top_i^{\uparrow}, \bot_i^{\uparrow}, \top_i^{\downarrow} \},$\\$\quad \{\bot_\ORccFive, \bot_i, \top_\ORccFive, \top_i \}\}$\\$= \{\PL_{\times\i}^\T,\PL_{\times\i}^{\F}\} =$\\$ \{\PL_\j^\b\vartriangle\PL_\k^\b,\PL_\j^\b\vartriangle\PL_\k^{\sim\b}\}$}  \\
(\boxtimes(\llbracket\sharp_{o,r}^{\times j}\rrbracket))_{\times j}\in &
\multicolumn{3}{c|}{-} & \multicolumn{2}{c|}{-} & -\\\hline

\txt{$(\boxplus\circ\boxtimes)$\\$(\llbracket\sharp_{o,r}^\alpha\rrbracket )_{\alpha\in\IdxSetNu}=$} & \txt{$\quad\PL_\i^\a$\\$\cap\ \ORccFive$\\$\cap\ \ORccFive$} & \txt{$\quad\PL_\i^\a$\\$\cap\ \PL_\j^\b$\\$\cap\ \ \ORccFive$} & \txt{$\quad\PL_\i^\a$\\$\cap\ \PL_\j^\b$\\$\cap\  \PL_\k^\c$} &  \txt{$\qquad\PL_{\times\k}^{\a\xor\b}$\\$\cap\ \ORccFive$} & \txt{$\quad\, \ \PL_{\times\k}^{\a\xor\b}$\\$\cap\ \PL_\k^\c$} & \txt{$(\PL_\i^\a\vartriangle\PL_\j^{\sim\b})$\\$\cap$\\$(\PL_\j^{\b}\vartriangle\PL_\k^{\sim\c})$} \\
& = & = & = & = & = & = \\
& \PL_{\i\ \j\ \k}^{\a--} & \PL_{\i\, \j\ \k}^{\a\b-} & \PL_{\i\, \j\, \k}^{\a\b\c} & \PL_{\i\, \j\ \k}^{\overline{\a\b}-} & \PL_{\i\, \j\, \k}^{\overline{\a\b}\c}  & \PL_{\i\, \j\, \k}^{\overline{\a\b\c}} \\
 \hline
\txt{$(\boxplus\circ\boxtimes\circ$\\$\overline{\rccfive}\circ\overline\sharp)$}=  &  \txt{$\lfloor\Q_{\i\j\k}?$\\=\\$\A^{\a--}\rfloor$}  & \txt{$\lfloor\Q_{\i\j\k}?$\\=\\$\A^{\a\b-}\rfloor$} & \txt{$\lfloor\Q_{\i\j\k}?$\\=\\$\A^{\a\b\c}\rfloor$} & \txt{$\lfloor\Q_{\overline{\i\j}\k}?$\\=\\$\A^{\overline{\a\b}-}\rfloor$} &  \txt{$\lfloor\Q_{\overline{\i\j}\k}?$\\=\\$\A^{\overline{\a\b}\c}\rfloor$} &  \lfloor\Q_{\overline{\i\j\k}}?=\A^{\overline{\a\b\c}}\rfloor  \\ 
& \multicolumn{6}{c}{=\ (\overline{\lfloor\cdot\rfloor}\circ \overline{\rccfive} \circ\ \overline{\sharp})\ \in \lfloor\Q_\nu?=\A^\nu\rfloor:\OReg^2_\nu \to \ORccFive_\nu\quad \text{(cf. Tabs. \ref{tab:enumerationOfmapsFromQaPattren}, \ref{tab:InteractionPattern})}}\\
\end{array}
$$
\end{table}

\subsection{The Combined Interaction- and Quantization Maps}

In the context discussion of the left side of the QIS principle of Eq. \ref{eq:MeaningOfMereoQuantification} and its diagram-based representation in Fig. \ref{fig:MereologyVSinformation}, constraints on the domain and co-domains of the interaction maps of the form $\sharp_o^\alpha:\OReg^2 \to \Reg^2$ with $\alpha \in \{\emptyset,1,2,\times \emptyset, \times1, \times2 \}$ (Tabs. \ref{tab:OverviewSupportSemantics}(2), \ref{tab:sharpInducesIstate} and \ref{tab:InteractionPattern}) in conjunction with quantization maps of the form $\overline{\lfloor\cdot\rfloor}=(\boxplus\circ\boxtimes)$ (Tab. \ref{tab:Quantization}) induce constraints on the domains and co-domains of maps of form $(\overline{\lfloor\cdot\rfloor}\circ\overline{\rccfive}\circ\overline{\sharp}):\OReg^2_\IdxOneToThreeBit \to \bigcup\Pi(\ORccFive)$ that can be stated as follows:
\begin{lemma}
\label{lemma:InteractionQuantizationType}
For all $\nu\in\IdxSetX\setminus\{\IdxZeroBit\}$, $o\in\OO$, and $r\in\OReg^2$ with $(o,r)\in\OReg_\nu^2$: $(o,r)\mapsto ((\boxplus\circ\boxtimes) (((\rccfive\circ\sharp_o^\alpha)(r),\sharpDet((\rccfive\circ\sharp_o^\alpha)(r))))_
{\alpha\in\IdxSetNu})\in \ORccFive_\nu$.
\end{lemma}
\begin{proof} By enumeration of cases. For every $\nu\in\IdxSetX$ inspect the respective columns in Tabs. \ref{tab:sharpInducesIstate},  \ref{tab:InteractionPattern} (for $(\overline{\rccfive}\circ\overline{\sharp})$) and Tab. \ref{tab:Quantization} (for $(\boxplus\circ\boxtimes)$). 
\end{proof}

Above, notations of the form $\lfloor\QQ_{\i\j\k}?=\A^{\alpha\beta\gamma}\rfloor$ were introduced as abbreviations for maps of the form $(\overline{\lfloor\cdot\rfloor}\circ\overline{\rccfive}\circ\overline{\sharp})$ where the former term emphasizes the relationships between patterns of upper and lower indexes of questions, answers and \istates\ and signatures/types of the underlying mappings and the second term emphasizes operational and compositional aspects of the maps themselves. This can be spelled out more precisely as follows:
\begin{definition}  \label{definition:InteractionQuantizationMap} For every lower  index pattern $\i\j\k$ with $\{i,j,k\}=\IdxSet$, every upper index pattern $\alpha\beta\gamma\in\Lambda_\nu$ and every relational elementary system $(o,r)\in\OReg^2_\nu$: the term $\lfloor\QQ_{\i\j\k}?=\A^{\alpha\beta\gamma}\rfloor(o,r)$ designates $\{\}$ if $\alpha\beta\gamma$ is not the answer pattern that Eq.~\ref{eq:IdxSetNuX} correlates with the classical bits extracted from $(o,r)$, overlined slots being compared by their parities, and otherwise designates the partition cell in $\ORccFive_\nu$ that results from the quantization of the classical bits in the indexed sets $\llbracket\sharp^\omega_{o,r}\rrbracket$ and corresponding determinacy information $\sharpDet(\llbracket\sharp^\omega_{o,r}\rrbracket)$ with index $\omega\in\IdxSetNu$ that are correlated with the answer pattern $\alpha\beta\gamma\in\Lambda_\nu$ such that $\lfloor\QQ_{\i\j\k}?=\A^{\alpha\beta\gamma}\rfloor(o,r) = ((\boxplus\circ\boxtimes) (((\rccfive\circ\sharp^\omega)(o,r),\sharpDet((\rccfive\circ\sharp^\omega)(o,r))))_
{\omega\in\IdxSetNu})$ and the following correspondences obtain between the  pattern of upper and lower indexes in $\lfloor\QQ_{\i\j\k}?=\A^{\alpha\beta\gamma}\rfloor$ and indexed sets of truth values of the form $\llbracket\sharp^\omega_{o,r}\rrbracket=(\rccfive\circ\sharp_o^\omega)(r)$ and $\sharpDet(\llbracket\sharp^\omega_{o,r}\rrbracket)$:
\begin{equation}\normalfont
\label{eq:IdxSetNuX}\arraycolsep=2pt
\begin{array}{c||ccc|ccc|ccc|cc|cc|cc}
\nu & \multicolumn{3}{c|}{\IdxOneBitPR} & \multicolumn{3}{c|}{\IdxTwoBitPR} & \multicolumn{3}{c|}{\IdxThreeBitPR} & \multicolumn{2}{c|}{\IdxOneBitET} & \multicolumn{2}{c|}{\IdxTwoBitET} & \multicolumn{2}{c}{\IdxThreeBitET}\\\hline
\alpha\beta\gamma\in\Lambda_\nu & \a&-&- & \a & \b & - & \a & \b & \c & - & \overline{\b\c} & \a & \overline{\b\c} & \multicolumn{2}{c}{\overline{\a\b\c}} \\\txt{permutation\\of $\IdxSet$} & \i & \j & \k & \i & \j & \k & \i & \j & \k & \i & \j\k & \i & \j\k & \multicolumn{2}{c}{\i\hat{\j}\k} \\  
\hline 
\txt{bit of\\$\llbracket\sharp^\omega_{o,r}\rrbracket$} & a_i & a_j &a_k & a_i& a_j & a_k & a_i & a_j & a_k & a_i & \txt{$a_j$\\\textsf{xor}\\$a_k$} & a_i & \txt{$a_j$\\\textsf{xor}\\$a_k$} & \txt{$a_i$\\\textsf{xor}\\$a_j$} & \txt{$a_j$\\\textsf{xor}\\$a_k$} \\\hdashline
\sharpDet(\llbracket\sharp^\omega_{o,r}\rrbracket) & \T & \F & \F & \T & \T & \F & \T & \T & \T & \F & \T & \T & \T & \T & \T\\\hline
\omega\in\IdxSetNu & \multicolumn{3}{c|}{\IdxSet} & \multicolumn{3}{c|}{\IdxSet} & \multicolumn{3}{c|}{\IdxSet} & \multicolumn{2}{c|}{\{i\}\cup\{\times i\}} &  \multicolumn{2}{c|}{\{i\}\cup\{\times i\}} & \multicolumn{2}{c}{\{\times k\, , \times i\}}\\ 
\end{array}
\end{equation}
\end{definition}
\noindent One can then state claims about map signatures that were enumerated in Tab. \ref{tab:enumerationOfmapsFromQaPattren} (above the horizontal double line) to motivate specific instances of the left side of the QIS principle formally as:
\begin{lemma}
 \label{lemma:QuantificationEqSupportLeft}
  The left side of the QIS principle is true in quantized mereology, i.e, for all $\nu\in\IdxSetX\setminus\{\IdxZeroBit\}$, $(o,r)\in\OReg^2_\nu$, $\alpha\beta\gamma\in\Lambda_\nu$, and $\i\j\k$ with $\{i,j,k\}=\IdxSet$: if $\lfloor\QQ_{\i\j\k}?=\A^{\alpha\beta\gamma}\rfloor(o,r)\neq\{\}$ then $ \lfloor\QQ_{\i\j\k}?=\A^{\alpha\beta\gamma}\rfloor(o,r)= \PL_{\,\i\,\j\,\k}^{\alpha\beta\gamma}$.
\end{lemma}
\begin{proof} Enumeration of cases that link the index notation $\PL_{\,\i\,\j\,\k}^{\alpha\beta\gamma}$ for \istates\ (defined  in Tab. \ref{tab:VectorsInIndexNotation}) to qubits $\lfloor\QQ_{\i\j\k}?=\A^{\alpha\beta\gamma}\rfloor(o,r)\in\bigcup\Pi(\ORccFive)$ that result via $\lfloor\QQ_{\i\j\k}?=\A^{\alpha\beta\gamma}\rfloor(o,r) = ((\boxplus\circ\boxtimes) (((\rccfive\circ\sharp_o^\omega)(r),\sharpDet((\rccfive\circ\sharp_o^\omega)(r))))_
{\omega\in\IdxSetNu})=(\overline{\lfloor\cdot\rfloor}\circ\overline{\rccfive}\circ\overline{\sharp})(o,r)$ (Def. \ref{definition:InteractionQuantizationMap}) from the combined actions of interaction maps $(\overline{\rccfive}\circ\overline{\sharp})$ (Tabs. \ref{tab:sharpInducesIstate} and \ref{tab:InteractionPattern}) and quantization maps $\overline{\lfloor\cdot\rfloor}$ (Tab. \ref{tab:Quantization}). 
\end{proof}

In summary and as visualized in Fig. \ref{fig:MereologyVSinformation}, from the outcomes $(\overline{\sharp})(o,r)\in ((\Reg^2)_\alpha)_{\alpha\in\IdxSetNu}$ of interaction mappings classical bits of information are extracted via $(\overline{\rccfive}\circ\overline{\sharp})(o,r)\equiv (\llbracket\sharp^\alpha_{o,r}\rrbracket,\sharpDet(\llbracket\sharp^\alpha_{o,r}\rrbracket)) \in((\ORccFive\times\TF)_\alpha)_{\alpha\in\IdxSetNu}$. As illustrated in Tab. \ref{tab:InteractionPattern} the interaction maps explicate certain interaction patterns. The classical bits are quantized by mapping them to cells in the \ThreeBit\ partition lattice and by combining the cells using the set-theoretic operations of intersection and disjunctive union as enumerated in Tab. \ref{tab:Quantization}. Jointly, the various cases that occur are enumerated in Tabs. \ref{tab:OverviewSupportSemantics}(2),  \ref{tab:sharpInducesIstate},  \ref{tab:InteractionPattern} and Tab. \ref{tab:Quantization}. In this context,  Lemmata \ref{lemma:InteractionQuantizationType} and \ref{lemma:QuantificationEqSupportLeft}  prove the left side of Eq. \ref{eq:MeaningOfMereoQuantification}, $(\overline{\lfloor\cdot\rfloor}\circ\overline{\rccfive}\circ\overline{\sharp})(o,r)  \equiv \lfloor\QQ_{\i\j\k}?=\A^{\alpha\beta\gamma}\rfloor(o,r) = \PL_{\, \i\,\j\,\k}^{\alpha\beta\gamma}$ for $(o,r)\in\OReg^2_\nu$ and $\PL_{\, \i\,\j\,\k}^{\alpha\beta\gamma}\in\ORccFive_\nu$ as well as its variation in Eq.~\ref{eq:ORegNuVSORccFiveNu}.

The ways in which the intermediate results are propagated through the sub-maps  of  composite maps of the form $(\boxplus\circ\boxtimes\circ\overline{\rccfive}\circ\overline{\sharp})$ via indexed sets which indexes abstract from specific permutations of the underlying index sets allows the statement of a stronger version of  the left side of Eq. \ref{eq:MeaningOfMereoQuantification}:
\begin{theorem}
\label{theorem:QuantificationEqSupportLeft}
The left side of the QIS principle remains invariant under the action $\rho$ of the symmetric group $\textsf{S}_3$ on the indexes of expressions in index notation i.e., for all $\nu\in\IdxSetX\setminus\{\IdxZeroBit\}$, $(o,r)\in\OReg^2_\nu$, $\alpha\beta\gamma\in\Lambda_\nu$, $\i\j\k$ with $\{i,j,k\}=\IdxSet$, and permutations $\rho\in\textsf{S}_3$: if $\lfloor\QQ_{\i\j\k}?(o,r)=\A^{\alpha\beta\gamma}\rfloor \neq \{\}$ then
$\lfloor\QQ_{\rho(\i\j\k)}?(o,r)=\A^{\rho(\alpha\beta\gamma)}\rfloor =  \PL \rho\left( \binom{\alpha}{\i}\binom{\beta}{\j}\binom{\gamma}{\k} \right)$.
\end{theorem}
\begin{proof}
Sequences of lower indexes the form  $\i\j\k$ (or $\i\hat{\j}\k$) and upper indexes of the form $\alpha\beta\gamma$ characterize relationships between arguments of map declarations of the form $\lfloor\QQ_{\i\j\k}?=\A^{\alpha\beta\gamma}\rfloor(o,r) = \PL_{\, \i\,\j\,\k}^{\alpha\beta\gamma}$. Relationships between map arguments depend only on abstract specifications of matching slots  in index sequences (cf. Lemma \ref{lemma:PermutationUpperAndLower} and Eq. \ref{eq:IdxSetNuX}) across the whole equation. By assumption, $\textsf{S}_3$ acts in the same ways on all index sequences of the equation and thus the relations between matching slots are preserved. By Lemma \ref{lemma:QuantificationEqSupportLeft} the left side QIS principle is true for any permutation of indexes and it remains true for all further permutations $\rho\in\textsf{S}_3$.
\end{proof}

\section{Combining Elementary Questions}
\label{sec:CombiningQuestions}

If the QIS principle $(\overline{\lfloor\cdot\rfloor}\circ\overline{\rccfive}\circ\overline{\sharp})  = \lceil\QQ_{\i\j\k}?=\A^{\alpha\beta\gamma}\rceil$ (Eq. \ref{eq:MeaningOfMereoQuantification}) is true, then the type constraints of Lemma \ref{lemma:InteractionQuantizationType} extend to the right side of the equation such that, for the index pattern $\alpha\beta\gamma\in\Lambda_\nu$, the term $\lceil\QQ_{\i\j\k}?=\A^{\alpha\beta\gamma}\rceil$ designates a map of the type: $\OReg^2_\nu\to\ORccFive_\nu$. In this context typed expressions of the form $(o,r)\in\OReg_\nu^2\mapsto \lfloor\QQ_{\i\j\k}?=\A^{\alpha\beta\gamma}\rfloor(o,r)\in \ORccFive_\nu$ and $(o,r)\in\OReg_\nu^2\mapsto \lceil\QQ_{\i\j\k}?=\A^{\alpha\beta\gamma}\rceil(o,r)\in \ORccFive_\nu$ capture the matching relation between (i) the amount of information that a relational elementary system in the set $\OReg^2_\nu$ carries, (ii) the amount of information that is sought by the observer $o$  by posing elementary questions of the form $\QQ_{\i\j\k}?$ that elicit answers of the form $\A^{\alpha\beta\gamma}$ with $\alpha\beta\gamma\in\Lambda_\nu$, and (iii) the information states in $\ORccFive_\nu$ that arise from the answers to the questions that were posed.\footnote{\label{footnote:IncompatibleAmountOfInfo}The situations in which the amount of information sought by asking questions of the form $\QQ_{\i\j\k}?$ to relational elementary systems $(o,(r_1,r_2))\in\OO\times\OReg^2$ is incompatible with the amount of information carried by $(r_1,r_2)$ are discussed in \cite{Bittner:VMQI}.}

Within the framework of matching amounts of information that is enforced in the formalism by typed expressions of the form $(o,r)\in\OReg_\nu^2\mapsto \lfloor\QQ_{\i\j\k}?=\A^{\alpha\beta\gamma}\rfloor(o,r)\in\ORccFive_\nu$  and $(o,r)\in\OReg_\nu^2\mapsto \lceil\QQ_{\i\j\k}?=\A^{\alpha\beta\gamma}\rceil(o,r)\in\ORccFive_\nu$, the QIS principle,  when stated in index notation $\lfloor\QQ_{\i\j\k}?=\A^{\alpha\beta\gamma}\rfloor(o,r)=\PL_{\,\i\,\j\,\k}^{\alpha\beta\gamma}=\lceil\QQ_{\i\j\k}?=\A^{\alpha\beta\gamma}\rceil(o,r)$ (Eq. \ref{eq:ORegNuVSORccFiveNu}), in addition indicates by matching index patterns that (a) the ways in which elementary questions and answers are combined correspond to (b) the ways in which \istates\ are combined and (c) the ways in which qubits are combined in the quantization process through $\overline{\lfloor\cdot \rfloor}$. While (c) was discussed above, the correspondence between (a) and (b) was discussed by \citet{Bittner:VMQI}. The latter is summarized below to provide a self-contained presentation of the QIS framework for quantized mereology.

\subsection{Cumulative and Non-cumulative Combination of Information}

As summarized in Tab. \ref{tab:OverviewSupportSemantics}, the set of non-empty maximum information states that support for answers to questions of the form $Q_i?$  posed in the context of the relational elementary systems $(r_1,r_2)_o \in \OReg^2_\IdxOneBitPR$ is the set $\RccFiveOneBit$ which consists of  six \istates\ of the form $\PL_\i^\a$ with $i\in\IdxSet$ and $\a\in\TF$. The elementary question $Q_i?$ is related to the partition $\{\PL_\i^\a\}$  by matching lower indices such that the maximum \istate\ supporting the answer $\a\in\YN$ to the question $Q_i?(r_1,r_2)_o$ is the cell $\PL_\i^\a$ of the partition $\{\PL_\i^\a\}$. Again, it is justified to seamlessly switch between the interpretation of the value of the variable $\a$ as a truth value and a yes/no answer, because (by construction) the set $\PL_\i^\a$ is identical to the truth-set of the formula $\oQ_i(r_1,r_2)_o$ and $\PL_\i^\a$ is also identical to the maximum supporting \istate\ $\lceil Q_i?(r_1,r_2)_o=\a\rceil$. In Fig. \ref{fig:MereologyVSinformation}, this was indicated by the two-sided dotted arrow labeled $\asymp$.

In this framework, every \istate\ in $\RccFiveOneBit$ carries one bit of mereological information by selecting one of two cells in one of the three partitions $\{\PL_{\emptyset}^{\a}\}$, $\{\PL_{1}^{\a}\}$, and $\{\PL_{2}^{\a}\}$.  From this and the constraints on elementary systems $(r_1,r_2)_o \in \OReg^2_\IdxOneBitPR$ in Tab.  \ref{tab:sharpInducesIstate} it follows that the information capacity of $(r_1,r_2)_o$ of one bit (definition of $\OReg^2_\IdxOneBitPR$ in Tab. \ref{tab:ClassesDefAxThPR}) is matched by the one bit of information carried by its \istate\ $\PL_\i^\a$ in $\RccFiveOneBit$.  Both $(r_1,r_2)_o$ and $\PL_\i^\a$ are correlated as a result of the interaction $(\sharp_o^i\, (r_1,r_2))$ with $\sharpDet(\sharp_o^i\, (r_1,r_2))$ and the fact that the truth-set for the formula $\oQ_i(r_1,r_2)_o$ is identical to the maximum supporting \istate\ for the answer $\a$ to the question $Q_i?(r_1,r_2)_o$ (Fig. \ref{fig:MereologyVSinformation} and Tab. \ref{tab:OverviewSupportSemantics}(3)).

It is possible to obtain more than one bit of information from an elementary system $(r_1,r_2)_o\in \OReg^2_\IdxOneToThreeBit\setminus \OReg^2_\IdxOneBit$ with an information capacity of $1<n\le 3$ bits. For a system of a product form this can be achieved by simultaneously asking two questions, $Q_i?(r_1,r_2)_o$ and $Q_j?(r_1,r_2)_o$, in  the context of the relational elementary system $(r_1,r_2)_o$ and combining the answers $\a$ and $\b$. The combination of questions and answers is mirrored by the combination of the correlated \istates\ $\lrceil{Q_i?(r_1,r_2)_o=\a}=\PL_\i^{\a}$ and $\lrceil{Q_j?(r_1,r_2)_o=\b}=\PL_\j^{\b}$. Since the information capacity of $(r_1,r_2)_o$ is  $1<n\le 3$ bits, the \istates\ $\PL_\i^{\a}$ and $\PL_\j^{\b}$ can be combined cumulatively. The additive or cumulative combination of information produces more information and therefore \istates\ of smaller cardinality resulting from the intersection/projection of \istates\ as illustrated in the left of Eq. \ref{eq:NonCumulativeDisjointUnion}. 

\begin{equation}
\label{eq:NonCumulativeDisjointUnion}
\begin{array}{ccc}
\overbrace{\underbrace{\PL_\i^\a}_{\txt{1 bit:\\$\a=\Y$ vs. $\a=\N$}} \cap \underbrace{\PL_\j^\b}_{\txt{1 bit:\\$\b=\Y$ vs. $\b=\N$}}}^{\txt{2 bits: $(\a,\b)\in \{\Y\Y,\Y\N,\N\Y,\N\N\}$}} & &
\overbrace{\underbrace{\PL_\i^\a}_{\txt{1 bit:\\$\a=\Y$ vs. $\a=\N$}} \vartriangle \underbrace{\PL_\j^\b}_{\txt{1 bit:\\$\b=\Y$ vs. $\b=\N$}}}^{\txt{1 bit: $\a=\b$ vs. $\a\neq\b$}}	\\\\
\txt{cumulative combination\\ of information} & &  \txt{non-cumulative combination\\ of information}
\end{array}
\end{equation} 

Information can also be combined in non-additive or non-cumulative ways for elementary systems in $\OReg^2_\IdxOneBitET$. Consider the interior planes in Fig. \ref{fig:PairsOfPlanes} (bottom) which are the result of combining the cells selected by $\PL_\i^\a$ and $\PL_\j^\b$ using the operation of disjunctive union $\PL_\i^\a\vartriangle\PL_\j^\b$. The two interior planes that form a partition of a \ThreeBit cube are differentiated by satisfying the requirement $\a = \b$ or the condition $\a \neq \b$. That is, although one bit of information is carried by each of the simple terms that pick out the cells that are combined ($\a\in\YN$ and $\b\in\YN$), the resulting disjunctive union also carries only one bit of information ( $\a = \b$ vs. $\a \neq \b$) as illustrated in the right of Eq. \ref{eq:NonCumulativeDisjointUnion}. Combining information in such a non-additive way is called \textit{non-cumulative} \cite{Bittner:VMQI}.

\subsection{Combining Questions, Answers and \istates}

\paragraph{The operator-based view}

Consider simultaneously asking two questions, $\Q_\delta?(r_1,r_2)$ and $\Q_\epsilon?(r_1,r_2)_o$, in the context of the relational elementary system $(r_1,r_2)_o$ and combining the answers $\a_\delta$ and $\b_\epsilon$. The combination of questions and answers is mirrored by the combination of the correlated \istates\ $\lrceil{\Q_\delta?(r_1,r_2)_o=\a_\delta}=\PL_\delta^{\a_\delta}$ and $\lrceil{\Q_\epsilon?(r_1,r_2)_o=\b_\epsilon}=\PL_\epsilon^{\b_\epsilon}$. Let "combining simultaneous questions" be designated by the symbol "${\bigcirc}$" and "combining simultaneous maximum \istates" be designated by "${\cdot}$" and "combining simultaneous answers" be designated by "${\Diamond}$". Using these conventions, variants of the following schemata for statements can be investigated in the context of the support semantics of Tab. \ref{tab:OverviewSupportSemantics}:
\begin{equation}
\label{eq:CompositionSchemata}
\begin{array}{cl}
(1) & \is \Models  (\Q_\delta?\bigcirc \Q_\epsilon?)(r_1,r_2)_o=(\a_\delta\, \Diamond\ \b_\epsilon)\ \text{if and only if}\ \is \subseteq \PL_\delta^{\a_\delta} {\ \cdot\ }  \PL_\epsilon^{\b_\epsilon}\\
(2) & \text{if}\ \lrceil{(\Q_\delta? \bigcirc \Q_\epsilon?)(r_1,r_2)_o=(\a\ \Diamond\ \b)}\neq \{\}\ \text{then}\\
& \multicolumn{1}{r}{\lrceil{(\Q_\delta? \bigcirc \Q_\epsilon?)(r_1,r_2)_o=(\a_\delta \Diamond\  \b_\epsilon)} = \PL_\delta^{\a_\delta} {\ \cdot\ }  \PL_\epsilon^{\b_\epsilon}} \\\\
\multicolumn{2}{l}{ \txt{with substitutions\\for $\bigcirc$, $\Diamond$ and $\cdot$}\ : \quad {\arraycolsep=2pt \begin{array}{c|c|c|l|c}
\multicolumn{3}{c|}{\text{when combining simultaneous}}\\
 \text{questions} & \text{answers} & \istates  & \text{mode of} & \text{details}\\
 \Q_\delta?\bigcirc\Q_\epsilon? & \a_\delta\, \Diamond\ \b_\epsilon & \PL_\delta^{\a_\delta} {\ \cdot\ }  \PL_\epsilon^{\b_\epsilon} & \text{combination} & \text{in Tabs} \\\hline
 \odot & (\cdot\, ,\cdot) & \cap & \text{cumulative} & \ref{tab:CombinedQuestionsIformsIstatesCum}, \ref{tab:CombinedQuestionsIformsIstatesThree}\\
 \oplus & \xor & \vartriangle & \text{non-cumulative} & \ref {tab:CombinedQuestionsIformsIstatesNonCum}, \ref{tab:CombinedQuestionsIformsIstatesThree}\\
\end{array}}}\\
\end{array}
\end{equation}

The schematic statements in Eq. \ref{eq:CompositionSchemata} (top) offer a framework for the examination of different scenarios of combining questions, answers and \istates\ in the context of a classical support semantics. As visualized in Eq. \ref{eq:CompositionSchemata} (bottom), the meta-variable $\bigcirc$  can be substituted by the operators $\odot$ or $\oplus$ which respectively signify the cumulative and non-cumulative combination of two simultaneous elementary questions. The meta-variable $\Diamond$ can be substituted by the tuple constructor $(\cdot\, ,\cdot)$ or the operator $\xor$ which signify, respectively, the cumulative and non-cumulative combination of answers to (combined) simultaneous elementary questions. Finally, the meta-variable $\cdot$ can be substituted by the set-theoretic operations of intersection and disjunctive union which signify, respectively, the cumulative and non-cumulative combination of \istates. In the non-cumulative case the second cell enters with its exponent negated, $\PL_\delta^{\a_\delta}\vartriangle\PL_\epsilon^{\sim\b_\epsilon}$, because $\PL_\i^\a\vartriangle\PL_\j^{\sim\b}$ is the set of patterns whose two answers have the parity $\a\xor\b$ (Sec.~\ref{sec:IndexNotationGuide}, Tab.~\ref{tab:CombinedQuestionsIformsIstatesNonCum}). The cases where the meta-variables $\Q_\delta?$ and $\Q_\epsilon?$ are substituted by the elementary questions $Q_i?$ and $Q_j?$ are enumerated in Tabs. \ref{tab:CombinedQuestionsIformsIstatesCum} and \ref{tab:CombinedQuestionsIformsIstatesNonCum}.

As illustrated in Tab. \ref{tab:CombinedQuestionsIformsIstatesThree}, the meta-variables $\Q_\delta?$ and $\Q_\epsilon?$ can also be substituted by simultaneously combined elementary questions of the form $(Q_i? \bigcirc Q_j?)$ and/or $(Q_j? \bigcirc Q_k?)$. This represents the various possibilities of combining three simultaneous elementary questions cumulatively and/or non-cumulatively. The various possible ways to combine three elementary questions / answers / \istates\ are enumerated in the table. Again, a detailed discussion can be found in \cite{Bittner:VMQI}.

\paragraph{Answered questions in index notation}

Having reviewed the compositional structure of combined simultaneous elementary questions and their semantics using the operator-based notation of the schemata in Eq. \ref{eq:CompositionSchemata}, the operator-based notation is now related to the more compact index-based notation for questions and their answers.  In analogy to the definition of the index-based notation for \istates\ using the operations of disjunctive union and intersection in Tab. \ref{tab:VectorsInIndexNotation}, index-based notations for answered  questions are  defined using the operations $\odot$ and $\oplus$ to compose questions, as well as the operations of tuple formation and \xor\ to combine the answers summarized in Tab. \ref{tab:defQuestionPattern}.

Although the index notation for \istates\ employs upper and lower indexes, the index notation for combined elementary questions employs only lower indexes. As discussed above, the index-based expression $\QQ_{\emptyset12}?(r_1,r_2)_o$ abbreviates the cumulative combination of the elementary questions $Q_\emptyset?\odot Q_1? \odot Q_2?(r_1,r_2)_o$. An answer to the combined question is of the form $(\a,\b,\c)\in \YN^3$ such that in index notation $\A^{\a\b\c}\equiv(\a,\b,\c)$. The corresponding non-empty maximum supporting \istate\ is of the form $\PL_{\emptyset12}^{\a\b\c}$, that is, $\lceil\QQ_{\emptyset12}?(r_1,r_2)_o=\A^{\a\b\c}\rceil \in \{\PL_{\emptyset12}^{\a\b\c},\{\}\}$.

Two cumulatively combined elementary questions result in index-based expressions with an unoccupied slot such that $\QQ_{\,\i\,\j-}?(r_1,r_2)_o$ abbreviates a cumulative combination of the form $(Q_\i?\odot Q_\j?)(r_1,r_2)_o$ whose answer is of the form $(\a,\b)\equiv\A^{\a\b-}\in \YN^2$. The corresponding non-empty maximum supporting \istate\ is of the form $\PL_{\,\i\,\j\,\k}^{\a\b-}$, that is, $\lceil\QQ_{\,\i\,\j-}?(r_1,r_2)_o=\A^{\a\b-}\rceil \in \{\PL_{\,\i\,\j\,\k}^{\a\b-},\{\}\}$. That there are three ways of combining two out of three questions is represented by distinguishing the additional question/answer patterns and the maximum supporting \istates: $\lceil\QQ_{-\j\,\k}?=\A^{-\b\c}\rceil\in\{\PL_{\,\i\,\j\,\k}^{-\b\c},\{\}\}$ and $\lceil\QQ_{\i-\k}?=\A^{\a-\c}\rceil \in \{\PL_{\,\i\ \j\,\k}^{\a-\c},\{\}\}$. This is summarized in the rows labeled $\QQ_\IdxTwoBitPR?$ of Tab. \ref{tab:defQuestionPattern}. The special case of a single question is captured by index-based expressions with two unoccupied slots of the form $\lceil\QQ_{\i--}?=\A^{\a--}\rceil\in\{\PL_{\,\i\ \j\ \k}^{\a--},\{\}\}$, $\lceil\QQ_{-\j-}?=\A^{-\b-}\rceil\in\{\PL_{\ \i\,\j\,\k}^{-\b-},\{\}\}$ and $\lceil\QQ_{--\k}?=\A^{--\c}\rceil\in\{\PL_{\ \i\ \j\,\k}^{--\c},\{\}\}$ (see the rows labeled $\QQ_\IdxOneBitPR?$ in the table). 

Non-cumulatively combined elementary questions result in index-based expressions with over-lined  indexes of the form $\lceil\QQ_{\overline{\i\,\j}-}?={\A^{\overline{\a\b}-}}\rceil\in\{{\PL_{\,\i\,\j\,\k}^{\overline{\a\b}-}},\{\}\}$. Here, $\QQ_{\overline{\i\,\j}-}?$ abbreviates $Q_i?\oplus Q_j?$ and ${\A^{\overline{\a\b}-}}$ abbreviates $\a \xor \b$ (rows labeled $\QQ_\IdxOneBitET?$). Similarly for $\lceil\QQ_{\overline{\i\,\j}\k}?={\A^{\overline{\a\b}\c}}\rceil\in\{{\PL_{\,\i\,\j\,\k}^{\overline{\a\b}\c}},\{\}\}$ that abbreviates $(Q_i?\oplus Q_j?)\odot Q_k?=(\a \xor \b, \c)$ (rows labeled $\QQ_\IdxTwoBitET?$) and $\lceil\QQ_{\overline{\i\hat{\j}\k}}?={\A^{\overline{\a\b\c}}}\rceil\in\{{\PL_{\,\i\,\hat{\j}\,\k}^{\overline{\a\b\c}}},\{\}\}$ which abbreviates $Q_i?\oplus Q_{\hat{j}}?\oplus Q_k?=(\a \xor \b, \b\xor\c)$ (row labeled $\QQ_\IdxThreeBitET?$)\footnote{\label{footnote:HatSymbol}$\lceil\QQ_{\overline{\i\hat{\j}\k}}?={\A^{\overline{\a\b\c}}}\rceil\in\{{\PL_{\i\,\hat{\j}\,\k}^{\overline{\a\b\c}}},\{\}\}$ in operator notation corresponds to $Q_i?\oplus Q_{\hat{j}}?\oplus Q_k?=(\a \xor \b, \b\xor\c)$ where the hat symbol (here $\hat{\j}$)  indicates the designated index  which occurs in both $\xor$ terms (via the upper index \b\ which corresponds to the lower index \j) and the corresponding discrete unions of partition cells (here $\PL_\j^\b$)  (cf. Lemma \ref{lemma:PermutationUpperAndLower}). When the hat symbol is omitted the designated index is the index in the middle.}. This illustrates that the substitutions for the operator-based view of combining elementary questions/answers/\istates\  that are employed in Tabs. \ref{tab:CombinedQuestionsIformsIstatesCum}, \ref {tab:CombinedQuestionsIformsIstatesNonCum} and \ref{tab:CombinedQuestionsIformsIstatesThree} mirror the substitution rules for the expressions in index notation that were introduced in Def. \ref{definition:IndexSubstitutionRules}.

\subsection{The Right Side of the QIS Principle}

Answered questions of the form $\QQ_{\i\j\k}?$ $(r_1,r_2)_o=\A^{\alpha\beta\gamma}$ are maps with type patterns of the form $\QQ_{\i\j\k}? : \OReg_\nu^2 \to \YN^{\QQ_\nu}$, where the  schematic variable $\QQ_{\i\j\k}?$ can be substituted by  the  index-based expressions within the $\QQ_\nu?$ groupings of  the second column of Tab. \ref{tab:defQuestionPattern}. As shown in the table, the substitution instances of the schematic  variable $\A^{\alpha\beta\gamma}$  correspond in operator notation to tuples of yes/no answers, the interpretation of which is determined by the compositional structure of the simultaneous elementary questions that constitute substitution instances of $\QQ_{\i\j\k}?$ such that either $\A^{\alpha\beta\gamma}$ corresponds to $\YN$ or $\A^{\alpha\beta\gamma}$ corresponds to $\YN^2$ or $\A^{\alpha\beta\gamma}$ corresponds to $\YN^3$. This is indicated as  $\YN^{\QQ_\nu}$ in  declarations of the form $\QQ_{\i\j\k}? : \OReg^2_\nu \to \YN^{\QQ_\nu}$ for all $\QQ_{\i\j\k}?\in\QQ_\nu?$. 

By inspection of the fourth column of Tab. \ref{tab:defQuestionPattern} which enumerates the combination of answers in operator notation to questions of the form $Q_i?\bigcirc Q_j?\bigcirc Q_k? \equiv \Q_{\i\j\k}?$ one can verify that there is a one-one correspondence between the combination of yes/no answers and the combination of truth values obtained via interaction maps of the form $(\overline{\rccfive}\circ\overline{\sharp})$. As enumerated in Eq. \ref{eq:IdxSetNuX}, truth values are combined via $\xor$ operations and the formation of indexed sets that corresponds to the ways in which (in Tab. \ref{tab:defQuestionPattern}) yes/no answers  are combined via $\xor$ operations and tuple formation. 

The absence of information that can arise in the context of interactions via $(\overline{\rccfive}\circ\overline{\sharp})$ is complemented by the absence of a question that would require/force an answer. This mirrors the point that was made in the beginning of this section: the QIS principle requires that the amount of information that an elementary system  carries matches the amount of information that is sought by posing  questions to that system.\footnote{Again,  situations in which the amount of information sought by asking questions  to elementary systems is incompatible with the amount of information carried by that system are discussed in \cite{Bittner:VMQI}.} The absence of a bit of information / a question / a yes/no answer is represented by a dash (signifying an empty slot) in index notation.

One can prove that the right side of the QIS principle is true in quantized mereology:

\begin{lemma}
\label{lemma:QuantificationEqSupportRight}
For all $\nu\in\IdxSetX\setminus\{\IdxZeroBit\}$, $r_o\in\OReg^2_\nu$, $\alpha\beta\gamma\in\Lambda_\nu$, and $\i\j\k$ with $\{i,j,k\}=\IdxSet$: if $\lceil\QQ_{\i\j\k}?(r_o)=\A^{\alpha\beta\gamma}\rceil \neq \{\}$ then
$\lceil\QQ_{\i\j\k}?(r_o)=\A^{\alpha\beta\gamma}\rceil =  \PL_{\,\i\,\j\,\k}^{\alpha\beta\gamma}$. 
\end{lemma}
\begin{proof}
By case patterns: $\lceil\Q_{\i\j\k}?(r_1,r_2)_o=\A^{\a\b\c}\rceil\equiv\lceil(Q_i?\odot Q_j?\odot Q_k?)(r_1,r_2)_o=(\a,\b,\c)\rceil=\PL_{\,\i\j\k}^{\a\b\c}$, a vertex of the cube (Tab.~\ref{tab:rccFive}); $\lceil\Q_{\i\j\k}?(r_1,r_2)_o=\A^{\a\b-}\rceil\equiv\lceil(Q_i?\odot Q_j?) (r_1,r_2)_o=(\a,\b)\rceil=\PL_{\,\i\j\k}^{\a\b-}$ in Fig. \ref{fig:RccFiveEdges};  $\lceil\Q_{\i\overline{\j\k}}?(r_1,r_2)_o = \A^{\a\overline{\b\c}}\rceil\equiv\lceil (Q_{\times i}?\odot Q_i?)(r_1,r_2)_o=(\b_{\times i},\a)\rceil=\PL_{\,\i\,\j\,\k}^{\a\overline{\b\c}}$ with $\b_{\times i}=\b\xor\c$ in Fig. \ref{fig:RccFiveEdgesX}; $\lceil\Q_{\overline{\i\hat{\j}\k}}?(r_1,r_2)_o=\A^{\overline{\a\b\c}}\rceil\equiv\lceil (Q_i?\oplus Q_j?\oplus Q_k?)(r_1,r_2)_o=(\a\xor\b,\b \xor \c)\rceil=\PL_{\,\i\,\hat{\j}\,\k}^{\overline{\a\b\c}}$ in Fig. \ref{fig:RccFiveEdgesXX}; $\lceil\Q_{\overline{\i\j}-}?(r_1,r_2)_o=\A^{\overline{\a\b}-}\rceil\equiv\lceil(Q_i?\oplus Q_j?) (r_1,r_2)_o=(\a \xor\b)\rceil$ in Fig. \ref{fig:PairsOfPlanes} (bottom); $\lceil\Q_{\i--}?(r_1,r_2)_o=\A^{\a--}\rceil\equiv\lceil Q_i?(r_1,r_2)_o=\a\rceil$ in Fig. \ref{fig:PairsOfPlanes} (top).
All figures are displayed in the $\emptyset12$-base with $\i=\emptyset$, $\j=1$ and $\k=2$. 
\end{proof}

 In analogy to Theorem \ref{theorem:QuantificationEqSupportLeft} one can state that the right side of the QIS principle remains invariant under the action $\rho$ of the symmetric group $\textsf{S}_3$:
\begin{theorem} 
\label{theorem:QuantificationEqSupportRight}
 For all $\nu\in\IdxSetX\setminus\{\IdxZeroBit\}$, $r_o\in\OReg^2_\nu$, $\alpha\beta\gamma\in\Lambda_\nu$, $\i\j\k$ with $\{i,j,k\}=\IdxSet$, and permutations $\rho\in\textsf{S}_3$: if $\lceil\QQ_{\i\j\k}?(r_o)=\A^{\alpha\beta\gamma}\rceil \neq \{\}$ then
$\lceil\QQ_{\rho(\i\j\k)}?(r_o)=\A^{\rho(\alpha\beta\gamma)}\rceil =  \PL \rho\left( \binom{\alpha}{\i}\binom{\beta}{\j}\binom{\gamma}{\k} \right)$.
\end{theorem}
\begin{proof} Sequences of the form  $\i\j\k$ ($\i\hat{\j}\k$) and $\alpha\beta\gamma$ are used only to specify matching slots of related questions, answers and \istates\ and their combination across equations of the form $\lceil\QQ_{\i\j\k}?(r_o)=\A^{\alpha\beta\gamma}\rceil =  \PL_{\,\i\,\j\,\k}^{\alpha\beta\gamma}$.  These relations are not affected by imposing a permutation $\rho\in\textsf{S}_3$ on all upper and lower index pattern. By Lemma \ref{lemma:QuantificationEqSupportRight} the right side QIS principle is true for an arbitrary permutation of indexes and it remains true for all further permutations $\rho\in\textsf{S}_3$.

\end{proof}

Jointly, Theorems \ref{theorem:QuantificationEqSupportLeft} and  \ref{theorem:QuantificationEqSupportRight} prove the QIS principle is true and  permutation invariant  in quantized mereology: 
\begin{corollary}
\label{corollary:QuantificationEqSupport}
For all $\nu\in\IdxSetX\setminus\{\IdxZeroBit\}$, $r_o\in\OReg^2_\nu$, $\alpha\beta\gamma\in\Lambda_\nu$, $\i\j\k$ with $\{i,j,k\}=\IdxSet$, and permutations $\rho\in\textsf{S}_3$,
the QIS principle is true in quantized mereology, i.e. if $\lfloor\QQ_{\i\j\k}?(r_o)=\A^{\alpha\beta\gamma}\rfloor \neq \{\}$ and $\lceil\QQ_{\i\j\k}?(r_o)=\A^{\alpha\beta\gamma}\rceil \neq \{\}$ then
$\lfloor\QQ_{\rho(\i\j\k)}?(r_o)=\A^{\rho(\alpha\beta\gamma)}\rfloor =  \PL \rho\left( \binom{\alpha}{\i}\binom{\beta}{\j}\binom{\gamma}{\k} \right)=\lceil\QQ_{\rho(\i\j\k)}?(r_o)=\A^{\rho(\alpha\beta\gamma)}\rceil$.
\end{corollary}
\noindent Under the hypotheses of Corollary~\ref{corollary:QuantificationEqSupport}, the corresponding instances of Eq.~\ref{eq:MeaningOfMereoQuantification} and the direction from system type to state type in Eq.~\ref{eq:ORegNuVSORccFiveNu} hold in quantized mereology.

\begin{table}
\caption{\label{tab:defQuestionPattern}Definition of question/answer patterns (Q/A patterns) and expressions in index-based notation for combined questions and answers as well as their relation to maximum supporting \istates\ in index-based notation such that $\lrceil{\QQ_{\i\j\k}?(r_1,r_2)_o=\A^{\alpha\beta\gamma}}=\RccFiveQA$ if $\lrceil{\QQ_{\i\j\k}?(r_1,r_2)_o=\A^{\alpha\beta\gamma}}\neq\{\}$ for $(r_1,r_2)_o\in\OReg^2_\nu$ and $\nu\in\IdxSetX$. \cite{Bittner:VMQI} }
\small
$$\arraycolsep=0.1pc\begin{array}{c|c|ccccc|ccccc|c|c|c|c}
\txt{Q/A-\\pattern}& \multicolumn{6}{c|}{\text{combined questions in}} & \multicolumn{7}{c|}{\text{answer patterns}} & \multicolumn{2}{c}{\istates} \\
\text{type} & \txt{index\\notation} & \multicolumn{5}{c|}{\text{operator notation}} & \multicolumn{5}{c|}{\text{operator notation}} &  \txt{index\\notation} & \txt{$\QQ_{\i\j\k}?(r_1,r_2)_o$\\$\in$} & \txt{index\\notation} & \text{type} \\\hline
    & \QQ_{\i--}? & Q_i? &  & - &  & - & \a & & - & & - & \A^{\a--} & & \PL_{\i\,\j\,\k}^{\a--} & \\
    \QQ_\IdxOneBitPR? & \QQ_{-\j-}? & - &  & \Q_j? & & - & - & & \b && - & \A^{-\b-} & \YN & \PL_{\i\,\j\,\k}^{-\b-} & \ORccFive_\IdxOneBitPR\\ 
    & \QQ_{--\k}? & - & & - & & Q_k? & - && - && \c & \A^{--\c} & & \PL_{\i\,\j\,\k}^{--\c} \\\hline
    & \QQ_{\i\j-}? & Q_i? & \odot & Q_j? & & - & (\a & , & \b) & & - & \A^{\a\b-} && \PL_{\i\,\j\,\k}^{{\a\b}-}\\ 
     \QQ_\IdxTwoBitPR? & \QQ_{\i-\k}? & Q_i? & & \odot  &  & Q_k? & (\a &  & , & & \c) & \A^{\a-\c} &  \YN \times \YN & \PL_{\i\,\j\,\k}^{{\a}-{\c}} & \ORccFive_\IdxTwoBitPR\\
     & \QQ_{-\j\k}? & - & & Q_j? & \odot & Q_k? & - & & (\b & , & \c) & \A^{-\b\c} & & \PL_{\i\,\j\,\k}^{-{\b\c}} \\\hline
     \QQ_\IdxThreeBitPR? & \QQ_{\i\j\k}? & Q_i? & \odot & Q_j? & \odot & Q_k? & (\a & , & \b & , & \c) & \A^{\a\b\c} & \txt{$\YN \times \YN$\\$\times \YN$} & \PL_{\i\,\j\,\k}^{\a\b\c} & \ORccFive_\IdxThreeBitPR\\\hline
     & \QQ_{\overline{\i\j}-}? & (Q_i? & \oplus & Q_j?) & & - & (\a & \xor & \b) & & - & \A^{\overline{\a\b}-} & & \PL_{\i\,\j\,\k}^{\overline{\a\b}-}\\
     \QQ_\IdxOneBitET? & \QQ_{\overline{\i}-\overline{\k}}? & Q_i? & & \oplus & & Q_k? & (\a & & \xor & & \c) & \A^{\overline{\a}-\overline{\c}} & \YN & \PL_{\i\,\j\,\k}^{\overline{\a}-\overline{\c}} & \ORccFive_\IdxOneBitET\\
     & \QQ_{-\overline{\j\k}}? & - & & Q_j? & \oplus & \Q_k? & - & & (\b & \xor & \c) & \A^{-\overline{\b\c}} & & \PL_{\i\,\j\,\k}^{-\overline{\b\c}}\\\hline
   & \QQ_{\overline{\i\j}\k}? & (Q_i? & \oplus & Q_j?) & \odot & Q_k? & ((\a & \xor & \b) & , & \c) & \A^{\overline{\a\b}\c} &  & \PL_{\i\,\j\,\k}^{\overline{\a\b}\c}\\  
   \QQ_\IdxTwoBitET? & \QQ_{\overline{\i}\j\overline{\k}}? & (Q_i? & \oplus & Q_k?) & \odot & Q_j? & ((\a & \xor & \c) & , & \b) & \A^{\overline{\a}\b\overline{\c}} & \YN \times \YN &  \PL_{\i\,\j\,\k}^{\overline{\a}\b\overline{\c}} & \ORccFive_\IdxTwoBitET \\  
   & \QQ_{\i\overline{\j\k}}? & Q_i? & \odot & (Q_j? & \oplus & Q_k?) & (\a & , & (\b & \xor & \c)) & \A^{\a\overline{\b\c}} & & \PL_{\i\,\j\,\k}^{\a\overline{\b\c}} \\\hline
    \QQ_\IdxThreeBitET? & \QQ_{\overline{\i\hat{\j}\k}}? & Q_i? & \oplus & Q_j? & \oplus & Q_k? & ((\a & \xor & \b) & , & (\b\xor\c)) & \A^{\overline{\a\b\c}} & \YN \times \YN & \PL_{\i\,\hat{\j}\,\k}^{\overline{\a\b\c}} & \ORccFive_\IdxThreeBitET \\  \hline
\QQ_\nu?    & \QQ_{\i\j\k}? & \multicolumn{5}{c|}{\Q_\delta? \bigcirc\, \Q_\epsilon?} & \multicolumn{5}{c|}{\a_\delta \Diamond\, \b_\epsilon} & \A^{\alpha\beta\gamma} & \YN^{\QQ_{\i\j\k}} & \RccFiveQA  & \ORccFive_\nu\\
     \multicolumn{16}{l}{\text{Symbols that are used in schematic expressions of the form}\ \lrceil{\QQ_{\i\j\k}?(r_1,r_2)_o=\A^{\alpha\beta\gamma}}=\RccFiveQA= \PL_\delta^{\a_\delta} {\ \cdot\ }  \PL_\epsilon^{\b_\epsilon}}\\
    \multicolumn{16}{l}{= \lrceil{(\Q_\delta? \bigcirc \Q_\epsilon?)(r_1,r_2)_o=(\a_\delta \Diamond\  \b_\epsilon)}\  \text{and which substitution instances occur in the respective columns.}}\\
\end{array}$$
\end{table}

\section{Framework Synthesis and the QIS Principle}
\label{sec:Synthesis}

Having established the core components of quantized mereology in the preceding sections (from the 3-bit information cube and the partition lattice to the interaction and quantization maps) these elements can now be synthesized to present a complete, synoptic view of the framework. The relationships between these components, developed throughout this paper, are summarized in Figure \ref{fig:CoarseOverview}. This diagram illustrates the full quantization process, showing how the logic of vague relations is formally derived from, and relates back to, its classical mereological foundation. The entire formal structure is unified by the Quantized Information Semantics (QIS) principle, which is now presented in its complete context.

The series of downward pointing arrows on the left side of Figure \ref{fig:CoarseOverview} represents the classical mereology for crisp regions. This is mirrored on the right by the structure for potentially vague regions, which maps ordered pairs to qubits and their corresponding equivalence classes. The arrows bridging from left to right symbolize the quantization process, while those pointing from right to left denote projection operations that relate the quantized framework back to its classical counterpart.

\paragraph{(i) The Partition-Theoretic Representation of Qubits}
A foundational methodological choice of this framework, following \citet{Ellerman2022}, was the representation of qubit states not as vectors in a Hilbert space, but as cells in partitions of the set of classical bit patterns, $\ORccFive$. Of the 4140 partitions of the 8-element set $\ORccFive$, the 15 that are generated by the faces and interior planes of the 3-bit information cube were selected as the mereologically meaningful ones; they form the carrier set of the non-distributive lattice shown in Figure \ref{fig:RccFivePartitionLattice}. The cells of these specific partitions constitute the set of all \istates\ of a relational elementary system.

\paragraph{(ii) The Unifying Role of the QIS Principle}
The framework's structure is unified by the Quantized Information Semantics (QIS) principle, first stated in Eq. \ref{eq:MeaningOfMereoQuantification}. This principle establishes the formal equivalence between two distinct pathways for determining an information state: The \textit{direct semantic path}, where an answered question directly identifies a maximum supporting \istate\ within a classical support semantics. This was denoted by  $\lceil\QQ_{\i\j\k}?=\A^{\alpha\beta\gamma}\rceil$. The \textit{indirect quantization path}, denoted $(\overline{\lfloor\cdot\rfloor}\circ\overline{\rccfive}\circ\overline{\sharp})$, which constructs the very same \istate\ by formally quantizing the underlying classical mereological information obtained via an interaction.

As established in Section~\ref{sec:CompositionalityRoadmap}, the equivalence of these paths rests on the dual role played by the interaction map, $\sharp$, which serves as both the constructive engine for the indirect path and the definitional anchor for the direct path. The proof of this equivalence, culminating in Corollary \ref{corollary:QuantificationEqSupport}, confirms the internal coherence of the framework for every relational elementary system and matching question/answer pattern satisfying the hypotheses of that corollary.

\paragraph{(iii) Classical Mereology as a Limiting Case}
Finally, the framework successfully demonstrates that classical mereological relations are preserved as a special, information-complete case. As shown in Section \ref{sec:QIS_cReg} and visualized on the left of Figure \ref{fig:MereologyVSinformation}, when the formalism is applied to crisp regions in $\cReg^2$, the system simplifies completely. The interaction map $\overline{\sharp}$ becomes the identity; the quantization map $\overline{\lfloor\cdot\rfloor}$ reduces to simple set formation; and the system operates exclusively on the most refined partition of the lattice, corresponding to the classical 3-bit states. This confirms that quantized mereology serves as a true generalization, extending rather than replacing the classical formalism.

 \begin{figure}
	$$
	\xymatrixcolsep{2pc}
	\xymatrix{
		\txt{crisp\\regions}	& (\Reg\setminus \{\emptyset\})^2 \ar@{<-}[d]_{\scriptstyle\id\times\id}\\
		&	\overline{\Reg^2} \ar[d]_{\txt{$\overline{\rccfive}$\\\tiny (Tab.~\ref{tab:rccFive})}} \ar@{}[rrd]|*+[F]{\txt{$\scriptstyle(\overline{\lfloor\cdot\rfloor}\circ\overline{\rccfive}\circ\overline{\sharp})=\lceil\QQ_{\i\j\k}?=\A^{\alpha\beta\gamma}\rceil$\\\tiny  Fig. \ref{fig:MereologyVSinformation}}}& &\OReg^2 \ar@{-->}@(ur,dr)[]^{\txt{$\natural$\\\tiny (B2)}}	\ar[d]^{\txt{$\scriptstyle \lceil\QQ_{\i\j\k}?=\A^{\alpha\beta\gamma}\rceil$\\\tiny (Tabs.  \ref{tab:OverviewSupportSemantics}, \ref{tab:CombinedQuestionsIformsIstatesCum},\\ \tiny\ref{tab:CombinedQuestionsIformsIstatesNonCum}, \ref{tab:CombinedQuestionsIformsIstatesThree}, \ref{tab:defQuestionPattern})}} \ar@{-->}[ll]^{\overline{\sharp}\ \txt{\tiny (Sec. \ref{sec:RegOneBitAndRegTwoBitAndRegThreeBit}, Eq. \ref{eq:InteractionConstraints})}} \ar@{-->}[llu]_{\txt{$\Jud$ \tiny (B3)}}& \txt{vague\\regions}\\		
		\txt{$\ORccFive=(\TF_i)_{1\le i \le 3}$\\$(\ORccFive,\wedge,\vee)$\\dist. lattice\\\tiny (Tab.~\ref{tab:rccFive})}&	\overline{\ORccFive\times\TF} \ar[d]_{\txt{$\scriptstyle([\cdot]_\sim)$\\\tiny (Tab.~\ref{tab:rccFive})}} \ar[rr]^{\overline{\lfloor\cdot\rfloor}}_{\txt{\tiny (Tabs.  \ref{tab:sharpInducesIstate}, \ref{tab:Quantization})}}&& 				
{\scriptstyle \bigcup\Pi(\ORccFive)} \ar[d]^{\txt{$([\cdot]_{\scriptstyle\simNu})$\\\tiny (B2)}} 
\ar@{-->}@(ur,dr)[]^{\txt{$\scriptstyle\lceil\cdot | \cdot\rceil$\\\tiny (B2)}}
		& \txt{$(\Pi(\ORccFive),\sqcap,\sqcup)$\\ non-distr.\\lattice {\tiny (Fig.\ref{fig:RccFivePartitionLattice})}} 				
		\\
		\txt{Information\\complete\\elem. syst.}& \Reg^2/_\sim  \ar[rr]_{\id\times\id} &  & {\displaystyle \OReg^2/}_{\scriptscriptstyle\simNu} \ar@{-->}@(ld,rd)[]_{\txt{$[\natural]$\\\tiny (B2)}} \ar@{-->}[lld]^{\txt{$[\Jud]$ \tiny (B3)}}& \txt{Information\\incomplete\\elem. syst.}\\
		\txt{Classical\\mereology}	& (\Reg\setminus \{\emptyset\})^2/_\sim \ar[u]_{\id\times\id} & &&  \txt{Quantized\\mereology}\\
	}
	$$
	\caption{\label{fig:CoarseOverview}Overview of the framework of quantized mereology as a diagram with pointers to the equations, tables and figures where the symbolic expressions that are used in the diagram are introduced. $A \to B$ signifies a map that takes members of the set $A$ to members of the set $B$. The over-lined symbols indicate different interpretations depending on whether ordered pairs of crisp entities/regions or ordered pairs of vague entities/regions are considered.  Dashed arrows signify maps that are subject to indeterminacy. A more detailed diagram can be found in Fig. \ref{fig:MereologyVSinformation}. Related work in \cite{Bittner:VMQI} and \cite{bittner:Hilbert} that is not reviewed here is signified as B2 and B3 respectively.} 
\end{figure}

\section{Conclusion}
\label{sec:Conclusion}

This paper addressed the long-standing challenge of applying classical mereology to entities with vague or indeterminate boundaries. The central thesis was that mereological vagueness can be productively reframed as an objective, information-theoretic limit on a system, a perspective that motivates a move from the classical bits of traditional formalisms to the richer structure of quantum bits, or qubits. 

The framework of Quantized Mereology was constructed through a series of methodical steps. First, the eight classical relations of the $\rccFivePlus$ formalism were mapped to the vertices of a 3-bit information cube, establishing a geometric framework for mereological information. The key step was to map mereological vagueness and indeterminacy to  higher-dimensional geometric features of this cube: its edges, faces, and interior planes.  These features were formalized as cells in a non-distributive partition lattice, providing an accessible, set-theoretic representation of qubits that avoids the Hilbert-space formalism of physics. 

The internal coherence of this entire structure was illustrated through the Quantized Information Semantics (QIS) principle.  This principle establishes a formal equivalence between two distinct pathways: the intuitive, top-down direct semantic path, which derives an information state from the meaning of an answered question, and the formal, bottom-up indirect quantization path, which constructs the same state from the underlying classical bits.  The proof of this equivalence confirms that the framework is both formally coherent and semantically meaningful.

The resulting framework makes several key contributions. It provides a true generalization of classical mereology, which is recovered as a special, information-complete case where all three bits are determinate.  It offers a novel and precise language for describing different kinds of vagueness, distinguishing between simple superposition (orthogonal vagueness) and correlation (non-orthogonal vagueness, or entanglement). Finally, the methodology employed, quantizing a classical qualitative reasoning formalism via a partition-based approach, offers a promising heuristic for other domains in AI and Qualitative Reasoning that must grapple with the irreducible indeterminacy and vagueness of the real world. 

Ultimately, this work presents a formal structure for vagueness based on abstract information forms (\iforms), using a specific ontology of regions as a powerful illustrative model. The choice of four distinct classes, regions that are crisp (\cReg), vague to two different degrees (\vReg, \zReg), and void of information (\eReg), is not arbitrary (cf. Remarks~\ref{remark:VoidOfInformation}~and~\ref{remark:CategoriesOfRegs}): the four sorts admit non-trivial examples for each of the framework's six information forms of positive capacity, ensure symmetric interaction behaviors, and account for the specific case of entities void of mereological features. Whether a smaller grounding would do is left open.

While this demonstrates the model's formal adequacy and internal consistency, it also highlights the distinction between formal machinery and \textit{philosophical grounding}. The four-part classification is stipulated axiomatically to satisfy the needs of the information-theoretic framework rather than being derived from an independent, pre-formal theory of vagueness. This leaves open several promising avenues for future research: What, if any, is the independent metaphysical criterion that distinguishes a \vReg\ from a \zReg? And is there a deeper reason to believe that objective vagueness is itself ``quantized'' into such discrete levels? Exploring alternative physical or metaphysical groundings that can be integrated into this information-theoretic framework will be crucial for clarifying the connection between the structure of reality and the information it can convey.

\section*{AI Disclosure}

Large language models were used in the preparation of this paper. Google Gemini
2.5 Pro was used in the revision of the manuscript, following this methodology:
(1) the model was asked to perform a critical but helpful review of the author's
manuscript using its deep research tool, (2) the model was asked to develop a
detailed revision plan that addresses the problems raised in that review, (3)
for each item of the plan the model was asked to generate text snippets for
inclusion in the manuscript, (4) the generated text was edited by the author and
the manuscript was updated, and (5) the model was asked to criticize the edited
text. The last two steps were repeated until agreement was reached or the model
was overruled. In this way the structure of the paper, the argument of the
paper, and the presentation of the definitions and proofs were improved.

In the later revision stages, Anthropic Claude (Opus) and OpenAI Codex were used
to locate relevant results in the literature, to check the definitions,
derivations, tables, and cross-references of the manuscript against each other,
to simulate peer review, and to copy-edit prose. The findings of each model were
reviewed independently of the other model, and no finding was adopted before it
had been re-derived from the manuscript.

The author reviewed and edited all generated material, independently checked
every citation, definition, derivation, and claim, and takes full responsibility
for the content of this publication. The conversation logs can be obtained from
the author upon request.

\printbibliography

\appendix

\section{Further Tables and Figures}
\label{app:tables}

\begin{figure}[H]
$$
\xymatrixcolsep{1pc}\xymatrix{
&&&& \ORccFive_\IdxTwoBit  \ar@{<..}[llld]_\subset \ar@{<..}[d]_\subset \ar@{<..}[rd]^\subset& &\\
& \RccFiveTwoBit \ar@{<..}[ddr]_<<<<<\subset \ar@{<..}[dd]_<<<<<\subset \ar@{<..}[ddl]_<<<<<\subset & &  \txt{$\ORccFive_\IdxThreeBit=$\\$\RccFiveThreeBit$} \ar@{<..}[d]_<<<<<= & \RccFiveTwoBitXX \ar@{<..}[ddl]_<<<<<=& \RccFiveTwoBitX \ar@{<..}[ddr]_<<<<<\subset \ar@{<..}[dd]_<<<<<\subset \ar@{<..}[ddl]_<<<<<\subset &\\
&  & &  \{\PL_{\i\, \j\, \k}^{\a\b\c}\} & &  &\\
\{\PL_{\i\ \j\ \k}^{\a\b-}\} \ar[urrr] & \{\PL_{\i\ \j\ \k}^{\a-\c}\} \ar[urr] & \{\PL_{\i\ \j\ \k}^{-\b\c}\} \ar[ur] & \{\PL_{\i\,\hat{\j}\, \k}^{\overline{\a\b\c}}\} \ar[u] & \{\PL_{\i\, \j\, \k}^{\overline{\a\b}\c}\} \ar[ul] & \{\PL_{\i\, \j\, \k}^{\overline{\a}\b\overline{\c}}\} \ar[ull] & \{\PL_{\i\, \j\, \k}^{\a\overline{\b\c}}\} \ar[ulll] \\
\{\PL_{\i\ \j\ \k}^{\a--}\} \ar[u] \ar[ur]  \ar[urrrrrr]  & \{\PL_{\i\ \j\ \k}^{-\b-}\} \ar[ul] \ar[ur]  \ar[urrrr] & \{\PL_{\i\ \j\ \k}^{--\c}\} \ar[ul] \ar[u] \ar[urr] & & \ar[ul] \{\PL_{\i\ \j\ \k}^{\overline{\a\b}-}\} \ar[u] & \ar[ull] \{\PL_{\i\ \j\ \k}^{\overline{\a}-\overline{\c}}\} \ar[u] & \ar[ulll] \{\PL_{\i\ \j\ \k}^{-\overline{\b\c}}\} \ar[u]\\
& \RccFiveOneBit \ar@{<..}[ur]_<<<<<\subset \ar@{<..}[u]_<<<<<\subset \ar@{<..}[ul]_<<<<<\subset \ar@{..>}[rrd]_\subset & & \ar[ulll] \ar[ull] \ar[ul] \{\PL_{\i\ \j\ \k}^{---}\} \ar[ur] \ar[urr] \ar[urrr] \ar@{..>}[r]_{=} & \RccFiveZeroBit & \RccFiveOneBitX \ar@{<..}[ur]_<<<<<\subset \ar@{<..}[u]_<<<<<\subset \ar@{<..}[ul]_<<<<<\subset  \ar@{..>}[lld]^\subset &\\
&&& \ORccFive_\IdxOneBit & & &\\
}$$
$$
   \begin{array}{c|c|c|c|c|c|c|c}
        \ORccFive_\nu  & \ORccFive_{\IdxOneBitPR} & \ORccFive_{\IdxOneBitET} & \ORccFive_{\IdxTwoBitPR} & \ORccFive_{\IdxThreeBitPR} & \ORccFive_{\IdxTwoBitET} & \ORccFive_{\IdxThreeBitET} & \RccFiveZeroBit\\  \hline
\txt{depicted in} & \txt{Fig. \ref{fig:PairsOfPlanes}(t)} & \txt{Fig. \ref{fig:PairsOfPlanes}(b)} & \txt{Fig. \ref{fig:RccFiveEdges}} & \txt{Fig.~\ref{fig:rccFive}(l)}  & \txt{Fig. \ref{fig:RccFiveEdgesX}} & \txt{Fig. \ref{fig:RccFiveEdgesXX}} & \txt{Fig.~\ref{fig:rccFive}(r)} \\
\end{array}$$
    \caption{\label{fig:RccFivePartitionLattice}The \ThreeBit\ partition lattice $(\Pi(\ORccFive),\sqcap,\sqcup,\{\PL_
    {\,\i\ \j\ \k}^{---}\},\{\PL_
    {\,\i\, \j\,\k}^{\a\b\c}\})$ of partitions formed by intersections and discrete unions of planar partitions of the \ThreeBit\ cube \cite{Bittner:VMQI}. Solid lines indicate cover relations of the refinement order, from which joins are read off as least upper bounds; the three refinements $\{\PL_{\i\,\j\,\k}^{\overline{\a\b}-}\}\lesssim\{\PL_{\i\,\j\,\k}^{\a\b-}\}$, $\{\PL_{\i\,\j\,\k}^{\overline{\a}-\overline{\c}}\}\lesssim\{\PL_{\i\,\j\,\k}^{\a-\c}\}$, and $\{\PL_{\i\,\j\,\k}^{-\overline{\b\c}}\}\lesssim\{\PL_{\i\,\j\,\k}^{-\b\c}\}$ are not drawn. The subset relations between partitions and \iform-indexed sets $\ORccFive_\nu$ with $\nu\in\IdxSetX$ are indicated by dotted lines. (See  Tab. \ref{tab:VectorsInIndexNotation} for definitions. An arbitrary but fixed assignment of the indexes \i\j\k\ with some permutation of members of $\IdxSet$ is assumed.)}

\end{figure}

\begin{table}[H]
\caption{\label{tab:OverviewSupportSemanticsTwo}Overview of elementary questions and the associated support semantics. (Adapted from \citet{bittner:InformationMereologyAndVagueness,Bittner:VMQI}) and the support semantics of \citet{Ciardelli2018}).}
	\begin{minipage}{12cm}
	$$\arraycolsep=2pt \begin{array}{l|lc} 
		    \text{`empty' region} & \emptyset,\ \text{where}\ (\Reg,\sqcap,\sqcup, \emptyset,\bigsqcup\Reg)\ \text{is a complete lattice}  \\
    \text{crisp regions} & \cReg = \Reg \setminus \{\emptyset\}\\
    & \text{$\Reg$ contains enough regions for each of the $\rccFive$ relations}\\
   \text{enough regions} &  \qquad \text{to be non-empty}:\ \forall r \in \RccFive.\  \exists r_1,r_2 \in\Reg.\ \rccfive(r_1,r_2) = r\\
    \text{vague regions} &  \vReg \cup \zReg \cup \eReg,\ \vReg \neq \{\},\ \zReg \neq \{\},\ \eReg\neq\{\}. & (1)\\
    \text{regions} & \OReg =  \cReg \cup \vReg \cup \zReg \cup\eReg; &  \\
    & \xReg \cap \yReg = \{\}\  \text{for}\ \textsf{x}\neq\textsf{y}\in\{\textsf{c},\textsf{v},\textsf{z},\textsf{e}\} & \\
& \xyReg\equiv\xReg\times\yReg\subset\OReg^2\ \text{with}\ \textsf{x},\textsf{y}\in \{\textsf{c},\textsf{v},\textsf{z},\textsf{e}\} \\   
    \hline
\txt{interaction\\map} & \parbox[t]{10cm}{The map $\sharp: \IdxSet \to \OO \to \OReg^2 \to \Reg^2$ provides a crisp proxy for a potentially vague pair. Its behavior is axiomatically constrained based on the types of constituent regions (\cReg, \vReg, \zReg, \eReg). If $r \in \vzReg \cup \zvReg \cup \vvReg\cup \zzReg$ then an additional map $(\sharp_o^{\times i}\, r)$ is defined for one index $i$, and two such maps if $r\in\vvReg$ (cf. Eqs.~\ref{eq:InteractionConstraints},\ref{eq:InteractionConstraintsXX} and Tab.~\ref{tab:sharpInducesIstate}). The notation $\sharp_o^\alpha:\OReg^2\to\Reg^2$ is used to abbreviate  $\sharp_o^i,\sharp_o^{\times i}$ with $\alpha \in \{\emptyset,1,2,\times \emptyset, \times 1, \times 2 \},\ i\in\IdxSet$.
 } & (2)\\
\hline
\text{elementary} &	Q? : \IdxSet \to \OO\times\OReg^2 \to \YN\ \ \equiv\ \ Q_i?: \OO\times\OReg^2 \to \YN\quad  \text{with}\ i\in\IdxSet \\
\text{questions} & \quad \text{where each elem. question $Q_i?$ has an associated proposition:}\\
&	\quad  Q_\emptyset(r_1,r_2)_o \equiv \pr_1(\sharp_o^\emptyset (r_1,r_2)) \meet \pr_2(\sharp_o^\emptyset (r_1,r_2)) \neq \emptyset\\ 
& \qquad = (\rccfive\,(\sharp_o^\emptyset(r_1,r_2))\, \emptyset);\\ 
&	\quad  Q_1(r_1,r_2)_o \equiv \pr_1(\sharp_o^1 (r_1,r_2)) \meet \pr_2(\sharp_o^1 (r_1,r_2)) = \pr_1(\sharp_o^1 (r_1,r_2))\\
& \qquad = (\rccfive\,(\sharp_o^1(r_1,r_2))\, 1); & (3)\\
&	\quad  Q_2 (r_1,r_2)_o \equiv \pr_1(\sharp_o^2 (r_1,r_2)) \meet \pr_2(\sharp_o^2 (r_1,r_2)) = \pr_2(\sharp_o^2 (r_1,r_2))\\
& \qquad  = (\rccfive\,(\sharp_o^2(r_1,r_2))\, 2)\quad \text{with}\ \pr_l(r_1,r_2)=r_l\ \text{for}\ l\in\{1,2\}. \\
&		Q_i?(r_1,r_2)_o= \left\{ \begin{array}{lll}\Y & \text{if} & \sharp_o^i(r_1,r_2) \models Q_i(r_1,r_2)_o\\  
		 \N &   \text{if} & \sharp_o^i(r_1,r_2) \models \neg Q_i(r_1,r_2)_o\\ \end{array}\right. \\\hline
		 
\text{truthsets} &	|\oQ | : \IdxSet \to \OO\times\OReg^2 \to (\ORccFive_\IdxOneBitPR \cup \{\{\}\})\ \text{such that}\\ 
&	|\oQ_i(r_1,r_2)_o|  = \{ \rccfive(s_1,s_2)\in\ORccFive \mid  \sharp_o^i(r_1,r_2) \models \oQ_i(r_1,r_2)_o\ \text{and} & (4) \\
&	 \qquad (s_1,s_2) \models \oQ_i(s_1,s_2)_o\ \text{and}\ (s_1,s_2)\in\Reg^2 \}\  \text{with}\ \oQ_i \in \{Q_i,\neg Q_i\} &  \\
& \text{if}\ |Q_i(r_1,r_2)_o| \neq \{\}\ \text{then}\  |Q_i(r_1,r_2)_o| = \PL_{\,\i}^\T \\
&  \text{if}\ |\neg Q_i(r_1,r_2)_o| \neq \{\}\ \text{then}\  |\neg Q_i(r_1,r_2)_o| = \PL_{\,\i}^\F. \\
 \hline
\text{support} &  \is \Models (Q_i?(r_1,r_2)_o=\a)  \equiv \text{if}\ \a=\Y\ \text{then}\ \is \subseteq |Q_i(r_1,r_2)_o|\\
& \qquad\qquad\text{else}\ \is \subseteq |\neg Q_i(r_1,r_2)_o|	 & (5)\\
\text{from \istates} & \quad \text{such that}\ \is\subseteq \PL_\i^\a \ \text{with the one-one correspondence}\ \Y \leftrightarrow \T,\ \N \leftrightarrow \F \\
\hline
 &  \lrceil{Q?} : \IdxSet \to \OO\times\OReg^2 \to \YN\to (\ORccFive_{\IdxOneBitPR} \cup \{\{\}\})\\
\text{max. \istates}& \lrceil{Q_i?(r_1,r_2)_o=\a} = \bigcup\{ \is \subseteq\ORccFive \mid  \ \is \Models Q_i?(r_1,r_2)_o=\a \} &  (6) \\
&\text{if}\ \lrceil{Q_i? (r_1,r_2)_o=\a} \neq\{\}\ \text{then}\ \lrceil{Q_i? (r_1,r_2)_o=\a} = \PL_\i^\a \\
\end{array}$$
	\end{minipage}	
\end{table}

\def\QOneOne{{\tiny\xymatrixcolsep{1.5pc}\xymatrixrowsep{1.5pc}\xymatrix@!0{
			& \bot \ar@{..}[rr]\ar@{..}'[d][dd]
			& & \EQ \ar@{..}[dd]
			\\
			\bot \ar@{..}[ur]\ar@{..}[rr]\ar@{..}[dd] 
			& & \PPi \ar@{-}[ur]\ar@{..}[dd] 
			\\
			& \bot \ar@{..}'[r][rr]
			& & \PP
			\\
			\DR\ar@{..}[rr]\ar@{..}[ur]
			& & \PO \ar@{..}[ur]
}}}

\def\QOneTwo{{\tiny\xymatrixcolsep{1.5pc}\xymatrixrowsep{1.5pc}\xymatrix@!0{
			& \bot \ar@{..}[rr]\ar@{..}'[d][dd]
			& & \EQ \ar@{..}[dd]
			\\
			\bot \ar@{..}[ur]\ar@{..}[rr]\ar@{..}[dd] 
			& & \PPi \ar@{..}[ur]\ar@{..}[dd] 
			\\
			& \bot \ar@{..}'[r][rr]
			& & \PP
			\\
			\DR\ar@{..}[rr]\ar@{..}[ur]
			& & \PO \ar@{-}[ur]
}}}

\def\QOneThree{{\tiny\xymatrixcolsep{1.5pc}\xymatrixrowsep{1.5pc}\xymatrix@!0{
			& \bot \ar@{..}[rr]\ar@{..}'[d][dd]
			& & \EQ \ar@{..}[dd]
			\\
			\bot \ar@{-}[ur]\ar@{..}[rr]\ar@{..}[dd] 
			& & \PPi \ar@{..}[ur]\ar@{..}[dd] 
			\\
			& \bot \ar@{..}'[r][rr]
			& & \PP
			\\
			\DR\ar@{..}[rr]\ar@{..}[ur]
			& & \PO \ar@{..}[ur]
}}}

\def\QOneFour{{\tiny\xymatrixcolsep{1.5pc}\xymatrixrowsep{1.5pc}\xymatrix@!0{
			& \bot \ar@{..}[rr]\ar@{..}'[d][dd]
			& & \EQ \ar@{..}[dd]
			\\
			\bot \ar@{..}[ur]\ar@{..}[rr]\ar@{..}[dd] 
			& & \PPi \ar@{..}[ur]\ar@{..}[dd] 
			\\
			& \bot \ar@{..}'[r][rr]
			& & \PP
			\\
			\DR\ar@{..}[rr]\ar@{-}[ur]
			& & \PO \ar@{..}[ur]
}}}

\def\QZeroOne{{\tiny\xymatrixcolsep{1.5pc}\xymatrixrowsep{1.5pc}\xymatrix@!0{
			& \bot \ar@{..}[rr]\ar@{..}'[d][dd]
			& & \EQ \ar@{-}[dd]
			\\
			\bot \ar@{..}[ur]\ar@{..}[rr]\ar@{..}[dd] 
			& & \PPi \ar@{..}[ur]\ar@{..}[dd] 
			\\
			& \bot \ar@{..}'[r][rr]
			& & \PP
			\\
			\DR\ar@{..}[rr]\ar@{..}[ur]
			& & \PO \ar@{..}[ur]
}}}

\def\QZeroTwo{{\tiny\xymatrixcolsep{1.5pc}\xymatrixrowsep{1.5pc}\xymatrix@!0{
			& \bot \ar@{..}[rr]\ar@{..}'[d][dd]
			& & \EQ \ar@{..}[dd]
			\\
			\bot \ar@{..}[ur]\ar@{..}[rr]\ar@{..}[dd] 
			& & \PPi \ar@{..}[ur]\ar@{-}[dd] 
			\\
			& \bot \ar@{..}'[r][rr]
			& & \PP
			\\
			\DR\ar@{..}[rr]\ar@{..}[ur]
			& & \PO \ar@{..}[ur]
}}}

\def\QZeroThree{{\tiny\xymatrixcolsep{1.5pc}\xymatrixrowsep{1.5pc}\xymatrix@!0{
			& \bot \ar@{..}[rr]\ar@{-}'[d][dd]
			& & \EQ \ar@{..}[dd]
			\\
			\bot \ar@{..}[ur]\ar@{..}[rr]\ar@{..}[dd] 
			& & \PPi \ar@{..}[ur]\ar@{..}[dd] 
			\\
			& \bot \ar@{..}'[r][rr]
			& & \PP
			\\
			\DR\ar@{..}[rr]\ar@{..}[ur]
			& & \PO \ar@{..}[ur]
}}}

\def\QZeroFour{{\tiny\xymatrixcolsep{1.5pc}\xymatrixrowsep{1.5pc}\xymatrix@!0{
			& \bot \ar@{..}[rr]\ar@{..}'[d][dd]
			& & \EQ \ar@{..}[dd]
			\\
			\bot \ar@{..}[ur]\ar@{..}[rr]\ar@{-}[dd] 
			& & \PPi \ar@{..}[ur]\ar@{..}[dd] 
			\\
			& \bot \ar@{..}'[r][rr]
			& & \PP
			\\
			\DR\ar@{..}[rr]\ar@{..}[ur]
			& & \PO \ar@{..}[ur]
}}}

\def\QThreeOne{{\tiny\xymatrixcolsep{1.5pc}\xymatrixrowsep{1.5pc}\xymatrix@!0{
			& \bot \ar@{-}[rr]\ar@{..}'[d][dd]
			& & \EQ \ar@{..}[dd]
			\\
			\bot \ar@{..}[ur]\ar@{..}[rr]\ar@{..}[dd] 
			& & \PPi \ar@{..}[ur]\ar@{..}[dd] 
			\\
			& \bot \ar@{..}'[r][rr]
			& & \PP
			\\
			\DR\ar@{..}[rr]\ar@{..}[ur]
			& & \PO \ar@{..}[ur]
}}}

\def\QThreeTwo{{\tiny\xymatrixcolsep{1.5pc}\xymatrixrowsep{1.5pc}\xymatrix@!0{
			& \bot \ar@{..}[rr]\ar@{..}'[d][dd]
			& & \EQ \ar@{..}[dd]
			\\
			\bot \ar@{..}[ur]\ar@{..}[rr]\ar@{..}[dd] 
			& & \PPi \ar@{..}[ur]\ar@{..}[dd] 
			\\
			& \bot \ar@{-}'[r][rr]
			& & \PP
			\\
			\DR\ar@{..}[rr]\ar@{..}[ur]
			& & \PO \ar@{..}[ur]
}}}

\def\QThreeThree{{\tiny\xymatrixcolsep{1.5pc}\xymatrixrowsep{1.5pc}\xymatrix@!0{
			& \bot \ar@{..}[rr]\ar@{..}'[d][dd]
			& & \EQ \ar@{..}[dd]
			\\
			\bot \ar@{..}[ur]\ar@{-}[rr]\ar@{..}[dd] 
			& & \PPi \ar@{..}[ur]\ar@{..}[dd] 
			\\
			& \bot \ar@{..}'[r][rr]
			& & \PP
			\\
			\DR\ar@{..}[rr]\ar@{..}[ur]
			& & \PO \ar@{..}[ur]
}}}

\def\QThreeFour{{\tiny\xymatrixcolsep{1.5pc}\xymatrixrowsep{1.5pc}\xymatrix@!0{
			& \bot \ar@{..}[rr]\ar@{..}'[d][dd]
			& & \EQ \ar@{..}[dd]
			\\
			\bot \ar@{..}[ur]\ar@{..}[rr]\ar@{..}[dd] 
			& & \PPi \ar@{..}[ur]\ar@{..}[dd] 
			\\
			& \bot \ar@{..}'[r][rr]
			& & \PP
			\\
			\DR\ar@{-}[rr]\ar@{..}[ur]
			& & \PO \ar@{..}[ur]
}}}

 \begin{figure}[H]\small
 	\centering
 	\small$$\arraycolsep=0.6pt\begin{array}{c|cccc}
 		& \multicolumn{4}{c}{\PL_\i^\a \cap \PL_\j^\b\Models(Q_i?\odot Q_j?)(r_1,r_2)=(\a,\b)} \\
 		\a\b & \Y\Y & \Y\N & \N\Y & \N\N\\ \hline 		
   \{\tensor*{\PL}{^\a_\emptyset^\b_1^-_2}\} & \QZeroOne & \QZeroTwo & \QZeroThree & \QZeroFour\\
 		& \PL_{\emptyset 1}^{\T\T}=\PL_\emptyset^\T\cap\PL_1^\T & \PL_{\emptyset 1}^{\T\F}=\PL_\emptyset^\T\cap\PL_1^\F & \PL_{\emptyset 1}^{\F\T}=\PL_\emptyset^\F\cap\PL_1^\T & \PL_{\emptyset 1}^{\F\F}=\PL_\emptyset^\F\cap\PL_1^\F  \\
 		\{\tensor*{\PL}{^\a_\emptyset^-_1^\b_2}\} &  \QOneOne & \QOneTwo& \QOneThree& \QOneFour\\
 		& \PL_{\emptyset 2}^{\T\T}=\PL_\emptyset^\T\cap\PL_2^\T& \PL_{\emptyset 2}^{\T\F}=\PL_\emptyset^\T\cap\PL_2^\F & \PL_{\emptyset 2}^{\F\T}=\PL_\emptyset^\F\cap\PL_2^\T & \PL_{\emptyset 2}^{\F\F}=\PL_\emptyset^\F\cap\PL_2^\F \\
 		\{\tensor*{\PL}{_{\emptyset 1 2}^{-\a\b}}\} & \QThreeOne & \QThreeTwo & \QThreeThree& \QThreeFour\\
 		& \PL_{12}^{\T\T}=\PL_1^\T\cap\PL_2^\T &  \PL_{12}^{\T\F}=\PL_1^\T\cap\PL_2^\F & \PL_{12}^{\F\T}=\PL_1^\F\cap\PL_2^\T & \PL_{12}^{\F\F}=\PL_1^\F\cap\PL_2^\F\\
 	\end{array}$$ 		
 	\caption{\label{fig:RccFiveEdges}Geometric illustration (in the $\emptyset12$ base) of maximum information states of the form  $\PL_\i^\a \cap \PL_\j^\b$ that  support answers to questions of the form $(Q_i?\odot Q_j?) (r_1,r_2)=(\a,\b)$. Supporting maximum information states are signified as solid edges in the \rccFive\ cube \cite{bittner:InformationMereologyAndVagueness}. The rows of the table form partitions of the \rccFive\ cube, and the three partitions form the family $ \RccFiveTwoBitFms$ of partitions formed by edges of the \rccFive\ cube.}
 \end{figure}

\def\QXOneOne{{\tiny\xymatrixcolsep{1.5pc}\xymatrixrowsep{1.5pc}\xymatrix@!0{
			& \bot \ar@{..}[rr]\ar@{..}'[d][dd]
			& & \EQ \ar@{..}[dd]
			\\
			\bot \ar@{..}[ur]\ar@{..}[rr]\ar@{..}[dd] 
			& & \PPi \ar@{..}[ur]\ar@{..}[dd] \ar@{-}[ld]
			\\
			& \bot \ar@{..}'[r][rr]
			& & \PP
			\\
			\DR\ar@{..}[rr]\ar@{..}[ur]
			& & \PO \ar@{..}[ur]
}}}

\def\QXOneTwo{{\tiny\xymatrixcolsep{1.5pc}\xymatrixrowsep{1.5pc}\xymatrix@!0{
			& \bot \ar@{..}[rr]\ar@{..}'[d][dd]
			& & \EQ \ar@{..}[dd]
			\\
			\bot \ar@{..}[ur]\ar@{..}[rr]\ar@{..}[dd] \ar@{-}[drrr]
			& & \PPi \ar@{..}[ur]\ar@{..}[dd] 
			\\
			& \bot \ar@{..}'[r][rr]
			& & \PP
			\\
			\DR\ar@{..}[rr]\ar@{..}[ur]
			& & \PO \ar@{..}[ur]
}}}

\def\QXOneThree{{\tiny\xymatrixcolsep{1.5pc}\xymatrixrowsep{1.5pc}\xymatrix@!0{
			& \bot \ar@{..}[rr]\ar@{..}'[d][dd]
			& & \EQ \ar@{..}[dd]
			\\
			\bot \ar@{..}[ur]\ar@{..}[rr]\ar@{..}[dd] 
			& & \PPi \ar@{..}[ur]\ar@{..}[dd] 
			\\
			& \bot \ar@{..}'[r][rr]
			& & \PP
			\\
			\DR\ar@{..}[rr]\ar@{..}[ur]
			& & \PO \ar@{..}[ur]\ar@{-}[uuul]
}}}

\def\QXOneFour{{\tiny\xymatrixcolsep{1.5pc}\xymatrixrowsep{1.5pc}\xymatrix@!0{
			& \bot \ar@{..}[rr]\ar@{..}'[d][dd]
			& & \EQ \ar@{..}[dd]
			\\
			\bot \ar@{..}[ur]\ar@{..}[rr]\ar@{..}[dd] 
			& & \PPi \ar@{..}[ur]\ar@{..}[dd] 
			\\
			& \bot \ar@{..}'[r][rr]
			& & \PP
			\\
			\DR\ar@{..}[rr]\ar@{..}[ur]\ar@{-}[rrruuu]
			& & \PO \ar@{..}[ur]
}}}

\def\QXTwoOne{{\tiny\xymatrixcolsep{1.5pc}\xymatrixrowsep{1.5pc}\xymatrix@!0{
			& \bot \ar@{..}[rr]\ar@{..}'[d][dd]
			& & \EQ \ar@{..}[dd]
			\\
			\bot \ar@{..}[ur]\ar@{..}[rr]\ar@{..}[dd] 
			& & \PPi \ar@{..}[ur]\ar@{..}[dd] 
			\\
			& \bot \ar@{..}'[r][rr]
			& & \PP \ar@{-}[lllu]
			\\
			\DR\ar@{..}[rr]\ar@{..}[ur]
			& & \PO \ar@{..}[ur]
}}}

\def\QXTwoTwo{{\tiny\xymatrixcolsep{1.5pc}\xymatrixrowsep{1.5pc}\xymatrix@!0{
			& \bot \ar@{..}[rr]\ar@{..}'[d][dd]
			& & \EQ \ar@{..}[dd]
			\\
			\bot \ar@{..}[ur]\ar@{..}[rr]\ar@{..}[dd] 
			& & \PPi \ar@{..}[ur]\ar@{..}[dd] \ar@{-}[dl]
			\\
			& \bot \ar@{..}'[r][rr]
			& & \PP
			\\
			\DR\ar@{..}[rr]\ar@{..}[ur]
			& & \PO \ar@{..}[ur]
}}}

\def\QXTwoThree{{\tiny\xymatrixcolsep{1.5pc}\xymatrixrowsep{1.5pc}\xymatrix@!0{
			& \bot \ar@{..}[rr]\ar@{..}'[d][dd]
			& & \EQ \ar@{..}[dd]
			\\
			\bot \ar@{..}[ur]\ar@{..}[rr]\ar@{..}[dd] 
			& & \PPi \ar@{..}[ur]\ar@{..}[dd] 
			\\
			& \bot \ar@{..}'[r][rr]
			& & \PP
			\\
			\DR\ar@{..}[rr]\ar@{..}[ur]
			& & \PO \ar@{..}[ur] \ar@{-}[uuul]
}}}

\def\QXTwoFour{{\tiny\xymatrixcolsep{1.5pc}\xymatrixrowsep{1.5pc}\xymatrix@!0{
			& \bot \ar@{..}[rr]\ar@{..}'[d][dd]
			& & \EQ \ar@{..}[dd]
			\\
			\bot \ar@{..}[ur]\ar@{..}[rr]\ar@{..}[dd] 
			& & \PPi \ar@{..}[ur]\ar@{..}[dd] 
			\\
			& \bot \ar@{..}'[r][rr]
			& & \PP
			\\
			\DR\ar@{..}[rr]\ar@{..}[ur]\ar@{-}[rrruuu]
			& & \PO \ar@{..}[ur]
}}}

\def\QXThreeOne{{\tiny\xymatrixcolsep{1.5pc}\xymatrixrowsep{1.5pc}\xymatrix@!0{
			& \bot \ar@{..}[rr]\ar@{..}'[d][dd]\ar@{-}[rd]
			& & \EQ \ar@{..}[dd]
			\\
			\bot \ar@{..}[ur]\ar@{..}[rr]\ar@{..}[dd] 
			& & \PPi \ar@{..}[ur]\ar@{..}[dd] 
			\\
			& \bot \ar@{..}'[r][rr]
			& & \PP
			\\
			\DR\ar@{..}[rr]\ar@{..}[ur]
			& & \PO \ar@{..}[ur]
}}}

\def\QXThreeTwo{{\tiny\xymatrixcolsep{1.5pc}\xymatrixrowsep{1.5pc}\xymatrix@!0{
			& \bot \ar@{..}[rr]\ar@{..}'[d][dd]
			& & \EQ \ar@{..}[dd]
			\\
			\bot \ar@{..}[ur]\ar@{..}[rr]\ar@{..}[dd] 
			& & \PPi \ar@{..}[ur]\ar@{..}[dd] 
			\\
			& \bot \ar@{..}'[r][rr]
			& & \PP
			\\
			\DR\ar@{..}[rr]\ar@{..}[ur]
			& & \PO \ar@{..}[ur] \ar@{-}[ul]
}}}

\def\QXThreeThree{{\tiny\xymatrixcolsep{1.5pc}\xymatrixrowsep{1.5pc}\xymatrix@!0{
			& \bot \ar@{..}[rr]\ar@{..}'[d][dd]
			& & \EQ \ar@{..}[dd]\ar@{-}[dlll]
			\\
			\bot \ar@{..}[ur]\ar@{..}[rr]\ar@{..}[dd] 
			& & \PPi \ar@{..}[ur]\ar@{..}[dd] 
			\\
			& \bot \ar@{..}'[r][rr]
			& & \PP
			\\
			\DR\ar@{..}[rr]\ar@{..}[ur]
			& & \PO \ar@{..}[ur]
}}}

\def\QXThreeFour{{\tiny\xymatrixcolsep{1.5pc}\xymatrixrowsep{1.5pc}\xymatrix@!0{
			& \bot \ar@{..}[rr]\ar@{..}'[d][dd]
			& & \EQ \ar@{..}[dd]
			\\
			\bot \ar@{..}[ur]\ar@{..}[rr]\ar@{..}[dd] 
			& & \PPi \ar@{..}[ur]\ar@{..}[dd] 
			\\
			& \bot \ar@{..}'[r][rr]
			& & \PP
			\\
			\DR\ar@{..}[rr]\ar@{..}[ur]\ar@{-}[urrr]
			& & \PO \ar@{..}[ur]
}}}

\def\QXFourOne{{\tiny\xymatrixcolsep{1.5pc}\xymatrixrowsep{1.5pc}\xymatrix@!0{
			& \bot \ar@{..}[rr]\ar@{..}'[d][dd]
			& & \EQ \ar@{..}[dd]
			\\
			\bot \ar@{..}[ur]\ar@{..}[rr]\ar@{..}[dd] 
			& & \PPi \ar@{..}[ur]\ar@{..}[dd] 
			\\
			& \bot \ar@{..}'[r][rr]
			& & \PP\ar@{-}[lluu]
			\\
			\DR\ar@{..}[rr]\ar@{..}[ur]
			& & \PO \ar@{..}[ur] 
}}}

\def\QXFourTwo{{\tiny\xymatrixcolsep{1.5pc}\xymatrixrowsep{1.5pc}\xymatrix@!0{
			& \bot \ar@{..}[rr]\ar@{..}'[d][dd]
			& & \EQ \ar@{..}[dd] \ar@{-}[lldd]
			\\
			\bot \ar@{..}[ur]\ar@{..}[rr]\ar@{..}[dd] 
			& & \PPi \ar@{..}[ur]\ar@{..}[dd] 
			\\
			& \bot \ar@{..}'[r][rr]
			& & \PP
			\\
			\DR\ar@{..}[rr]\ar@{..}[ur]
			& & \PO \ar@{..}[ur]
}}}

\def\QXFourThree{{\tiny\xymatrixcolsep{1.5pc}\xymatrixrowsep{1.5pc}\xymatrix@!0{
			& \bot \ar@{..}[rr]\ar@{..}'[d][dd]
			& & \EQ \ar@{..}[dd]
			\\
			\bot \ar@{..}[ur]\ar@{..}[rr]\ar@{..}[dd] 
			& & \PPi \ar@{..}[ur]\ar@{..}[dd] 
			\\
			& \bot \ar@{..}'[r][rr]
			& & \PP
			\\
			\DR\ar@{..}[rr]\ar@{..}[ur]
			& & \PO \ar@{..}[ur]\ar@{-}[lluu]
}}}

\def\QXFourFour{{\tiny\xymatrixcolsep{1.5pc}\xymatrixrowsep{1.5pc}\xymatrix@!0{
			& \bot \ar@{..}[rr]\ar@{..}'[d][dd]
			& & \EQ \ar@{..}[dd]
			\\
			\bot \ar@{..}[ur]\ar@{..}[rr]\ar@{..}[dd] 
			& & \PPi \ar@{..}[ur]\ar@{..}[dd] 
			\\
			& \bot \ar@{..}'[r][rr]
			& & \PP
			\\
			\DR\ar@{..}[rr]\ar@{-}[rruu]\ar@{..}[ur]
			& & \PO \ar@{..}[ur]
}}}

\def\QXFiveOne{{\tiny\xymatrixcolsep{1.5pc}\xymatrixrowsep{1.5pc}\xymatrix@!0{
			& \bot \ar@{..}[rr]\ar@{..}'[d][dd]
			& & \EQ \ar@{..}[dd]
			\\
			\bot \ar@{..}[ur]\ar@{..}[rr]\ar@{..}[dd] 
			& & \PPi \ar@{..}[ur]\ar@{..}[dd] 
			\\
			& \bot \ar@{..}'[r][rr]
			& & \PP\ar@{-}[lu]
			\\
			\DR\ar@{..}[rr]\ar@{..}[ur]
			& & \PO \ar@{..}[ur] 
}}}

\def\QXFiveTwo{{\tiny\xymatrixcolsep{1.5pc}\xymatrixrowsep{1.5pc}\xymatrix@!0{
			& \bot \ar@{..}[rr]\ar@{..}'[d][dd]
			& & \EQ \ar@{..}[dd]
			\\
			\bot \ar@{..}[ur]\ar@{..}[rr]\ar@{..}[dd]  \ar@{-}[rd]
			& & \PPi \ar@{..}[ur]\ar@{..}[dd] 
			\\
			& \bot \ar@{..}'[r][rr]
			& & \PP
			\\
			\DR\ar@{..}[rr]\ar@{..}[ur]
			& & \PO \ar@{..}[ur]
}}}

\def\QXFiveThree{{\tiny\xymatrixcolsep{1.5pc}\xymatrixrowsep{1.5pc}\xymatrix@!0{
			& \bot \ar@{..}[rr]\ar@{..}'[d][dd]
			& & \EQ \ar@{..}[dd]
			\\
			\bot \ar@{..}[ur]\ar@{..}[rr]\ar@{..}[dd] 
			& & \PPi \ar@{..}[ur]\ar@{..}[dd] 
			\\
			& \bot \ar@{..}'[r][rr]
			& & \PP
			\\
			\DR\ar@{..}[rr]\ar@{..}[ur]
			& & \PO \ar@{..}[ur]\ar@{-}[ruuu]
}}}

\def\QXFiveFour{{\tiny\xymatrixcolsep{1.5pc}\xymatrixrowsep{1.5pc}\xymatrix@!0{
			& \bot \ar@{..}[rr]\ar@{..}'[d][dd]
			& & \EQ \ar@{..}[dd]
			\\
			\bot \ar@{..}[ur]\ar@{..}[rr]\ar@{..}[dd] 
			& & \PPi \ar@{..}[ur]\ar@{..}[dd] 
			\\
			& \bot \ar@{..}'[r][rr]
			& & \PP
			\\
			\DR\ar@{..}[rr]\ar@{..}[ur]\ar@{-}[uuur]
			& & \PO \ar@{..}[ur] 
}}}

\begin{figure}[H]
\centering
	$$\begin{array}{c|cccc}
 		(\b_{\times i},\a)= & \multicolumn{4}{c}{\tensor*{\PL}{_{\times\i}^{\b_{\times\i}}} \cap \tensor*{\PL}{_{\i}^{\a}}\Models(Q_{\times i}?\odot Q_i?)(r_1,r_2)=(\b_{\times i},\a)} \\
 		(\b\xor\c,\a) & \N\Y & \Y\Y & \N\N & \Y\N\\ \hline
	\txt{\quad\\$i=1$\\$\{\PL_{\emptyset12}^{\overline{\b}\a\overline{\c}}\}$} & \QXFourTwo & \QXFourOne  & \QXFourFour & \QXFourThree \\
	& \tensor*{\PL}{_{\emptyset12}^{\overline{\T}\T\overline{\T}}} = \tensor*{\PL}{_{\emptyset12}^{\overline{\F}\T\overline{\F}}} & \tensor*{\PL}{_{\emptyset12}^{\overline{\T}\T\overline{\F}}} = \tensor*{\PL}{_{\emptyset12}^{\overline{\F}\T\overline{\T}}} &  \tensor*{\PL}{_{\emptyset12}^{\overline{\T}\F\overline{\T}}} = \tensor*{\PL}{_{\emptyset12}^{\overline{\F}\F\overline{\F}}} &  \tensor*{\PL}{_{\emptyset12}^{\overline{\T}\F\overline{\F}}} = \tensor*{\PL}{_{\emptyset12}^{\overline{\F}\F\overline{\T}}} \\	
	\txt{\quad\\$i=2$\\$\{\PL_{\emptyset12}^{\overline{\b\c}\a}\}$} & \QXThreeThree &  \QXThreeOne& \QXThreeFour & \QXThreeTwo\\
	& \tensor*{\PL}{_{\emptyset12}^{\overline{\T\T}\T}} = \tensor*{\PL}{_{\emptyset12}^{\overline{\F\F}\T}}  &  \tensor*{\PL}{_{\emptyset12}^{\overline{\T\F}\T}} = \tensor*{\PL}{_{\emptyset12}^{\overline{\F\T}\T}} &  \tensor*{\PL}{_{\emptyset12}^{\overline{\T\T}\F}} = \tensor*{\PL}{_{\emptyset12}^{\overline{\F\F}\F}} &  \tensor*{\PL}{_{\emptyset12}^{\overline{\T\F}\F}} = \tensor*{\PL}{_{\emptyset12}^{\overline{\F\T}\F}}\\
	\txt{\quad\\$i=\emptyset$\\$\{\PL_{\emptyset12}^{\a\overline{\b\c}}\}$} & \QXFiveThree & \QXFiveOne &  \QXFiveFour &\QXFiveTwo\\
	& \tensor*{\PL}{_{\emptyset12}^{\T\overline{\T\T}}} = \tensor*{\PL}{_{\emptyset12}^{\T\overline{\F\F}}} & \tensor*{\PL}{_{\emptyset12}^{\T\overline{\T\F}}} = \tensor*{\PL}{_{\emptyset12}^{\T\overline{\F\T}}} & \tensor*{\PL}{_{\emptyset12}^{\F\overline{\T\T}}} = \tensor*{\PL}{_{\emptyset12}^{\F\overline{\F\F}}} &  \tensor*{\PL}{_{\emptyset12}^{\F\overline{\T\F}}} = \tensor*{\PL}{_{\emptyset12}^{\F\overline{\F\T}}}\\
	\end{array}$$

	\caption{\label{fig:RccFiveEdgesX}While the states in Figure~\ref{fig:RccFiveEdges} correspond to the cube's edges, these 2-bit entangled states correspond to diagonals across its faces. Interpretation  of    2-bit \istates\ as solid linear features across the $\parallel$planes of the \rccFive\ cube (displayed in the $\emptyset12$ base). The rows of the table form the partitions  $\{\tensor*{\PL}{_{\emptyset 12}^{\a\overline{\b\c}}}\}$, $\{\tensor*{\PL}{_{\emptyset 12}^{\overline{\b}\a\overline{\c}}}\}$ and   $\{\tensor*{\PL}{_{\emptyset 12}^{\overline{\b\c}\a}}\}$ of the \rccFive\ cube. The three partitions form the family $\RccFiveTwoBitFmsX$. \cite{Bittner:VMQI}}
\end{figure}

\begin{figure}[H]
\centering

	$$\def\arraystretch{1.4}
	\begin{array}{c|cccc} 
 	\txt{maximum\\supporting\\\istates}	& \multicolumn{4}{c}{ \left\{\begin{array}{c}\PL_{\i\, \hat{\j}\, \k}^{\overline{\a\b\c}}\\ \PL_{\times\k}^{\a\xor\b}\cap\PL_{\times\i}^{\b\xor\c} \\ (\PL_\i^\a\vartriangle\PL_\j^{\sim\b})\cap(\PL_\j^{\b}\vartriangle\PL_\k^{\sim\c}) \end{array} \right\} \Models 
 	\left\{ \begin{array}{l}(Q_i?\oplus Q_j?\oplus Q_k?)(r_1,r_2)=\\
 	\qquad(\a\xor\b,\b \xor \c)\\ 
 	\end{array} \right\} 
 	} \\\\
\txt{\\displayed\\in $\emptyset12$\\base}	  & \QXOneFour & \QXOneThree & \QXOneTwo & \QXOneOne\\	 
\DIA_j\ni & \{\bot_\ORccFive,\top_\ORccFive\} &  \{\bot_j^\downarrow,\top_j^\downarrow\}  & \{\bot_j^\uparrow, \top_j^\uparrow\}  & \{\bot_j,\top_j\}\\ 
 & = & = & = & = \\
 &\PL_{\times\k}^\F\cap\PL_{\times\i}^\F &\PL_{\times\k}^\T\cap\PL_{\times\i}^\F & \PL_{\times\k}^\F\cap\PL_{\times\i}^\T & \PL_{\times \k}^\T\cap\PL_{\times \i}^\T \\
 & = & = & = & =\\
 & \txt{$\PL_\i^\b \vartriangle\PL_\j^{\sim\b}$\\$ \cap$\\ $\PL_\j^{\b}\vartriangle\PL_\k^{\sim\b}$}&  \txt{$\PL_\i^{\sim\b} \vartriangle\PL_\j^{\sim\b}$\\$ \cap$\\ $\PL_\j^{\b}\vartriangle\PL_\k^{\sim\b}$} & \txt{$\PL_\i^\b \vartriangle\PL_\j^{\sim\b}$\\$ \cap$\\ $\PL_\j^{\b}\vartriangle\PL_\k^{\b}$}  & \txt{$\PL_\i^{\sim\b} \vartriangle\PL_\j^{\sim\b}$\\$ \cap$\\ $\PL_\j^{\b}\vartriangle\PL_\k^{\b}$} \\
& = & = & = & =\\
 \txt{index\\notation} & \PL_{\,\i\,\j\,\k}^{\overline{\F\F\F}}=\PL_{\,\i\,\j\,\k}^{\overline{\T\T\T}} & \PL_{\,\i\,\j\,\k}^{\overline{\T\F\F}}=\PL_{\,\i\,\j\,\k}^{\overline{\F\T\T}}  & \PL_{\,\i\,\j\,\k}^{\overline{\F\F\T}}=\PL_{\,\i\,\j\,\k}^{\overline{\T\T\F}} & \PL_{\,\i\,\j\,\k}^{\overline{\T\F\T}}=\PL_{\,\i\,\j\,\k}^{\overline{\F\T\F}} \\
 \\
 \txt{$(\a\xor\b,$\\$\ \b \xor \c)$} = &  \N\N & \Y\N & \N\Y  & \Y\Y\\
	\end{array}$$

	\caption{\label{fig:RccFiveEdgesXX}Interpretation  of  maximum supporting  2-bit \istates\ for answered questions of the form $(Q_i?\oplus Q_j?\oplus Q_k?)(r_1,r_2)=(\a\xor\b,\b \xor \c)$ with $(r_1,r_2)\in \OReg^2_\IdxThreeBitET$ as solid linear features throughout the interior of the \ThreeBit\ information cube (displayed in the $\emptyset12$ base with $\emptyset12=\i\j\k=\i\hat{\j}\k$).  (Adapted from  \citet{Bittner:VMQI}.)}
\end{figure}

\begin{table}[H]
\caption{\label{tab:enumBoxPlusX}Partial enumeration of the composite map $(\boxplus\circ\boxtimes): ((\ORccFive \times \TF)_\alpha)_{\alpha\in\{\times k,k\}} \to \ORccFive_\nu$ of Fig. \ref{fig:MereologyVSinformation} with $\nu\in\{\IdxOneBitET,\IdxTwoBitET\}$. The first sub-map $\boxtimes$ takes indexed sets of the form $((\ORccFive\times\TF)_\alpha)_{\alpha\in\{\times k,k\}}$ to partition cells in the set $\ORccFive_{\IdxOneBitET}$ (expressed as disjunctive unions) and $\ORccFive_{\IdxOneBitPR}\cup\{\ORccFive\}$; of the three middle columns, the first two expand the interior plane returned by $\boxtimes$ into the faces of Tab.~\ref{tab:VectorsInIndexNotation} whose disjunctive union it is, and the third is the cell returned by $\boxtimes$ for the simple interaction. The second sub-map $\boxplus$ takes partition cells from $\ORccFive_{\IdxOneBitET}$ (expressed as disjunctive unions) and $\ORccFive_{\IdxOneBitPR}\cup\{\ORccFive\}$ to partition cells in $\ORccFive_\nu$ using the projection operator $\cap$. The various patterns in the associated column identify the ways in which $\boxplus$ corresponds to the definitions of the index notation in Tab \ref{tab:VectorsInIndexNotation}.}
\small
$$ \def\arraystretch{1}
\begin{array}{ccc|ccc|cc}
\multicolumn{8}{c}{}\\
\multicolumn{6}{c|}{\text{for}\ (o,r)\in(\OReg^2_\IdxOneBitET\cup\OReg^2_\IdxTwoBitET)\ \text{and}\ \alpha \in \{\times k, k\}} & \multicolumn{2}{c}{\text{assembly}}\\
\multicolumn{3}{c|}{(\overline{\rccfive}\circ\overline{\sharp})(r_o) \in ((\ORccFive\times\TF)_\alpha)_{\alpha\in\{\times k, k\}}} & \multicolumn{3}{c|}{\text{algebraic factors}} & \multicolumn{2}{c}{\text{final cell}}\\
\multicolumn{2}{c}{(\llbracket\sharp^{\times k}_{o,r}\rrbracket,\sharpDet(\llbracket\sharp_{o,r}^{\times k}\rrbracket))_{\times k}} & (\llbracket\sharp^k_{o,r}\rrbracket,\sharpDet(\llbracket\sharp_{o,r}^k\rrbracket))_k & \multicolumn{3}{c|}{\text{for the final intersection}} & \multicolumn{2}{c}{}\\
\multicolumn{2}{c}{\in} & \in & \text{first face} & \text{second face} & \text{simple cell} & \txt{\small(Tab \ref{tab:VectorsInIndexNotation})} & \ORccFive_\nu\\\hline
\{(\bot_\ORccFive,\T),(\bot_k,\T), & (\top_\ORccFive,\T), (\top_k,\T)\} & \{(\bot_\ORccFive,\F),(\top_\ORccFive,\F)\} & \PL_\i^\F & \PL_\j^\T & \PL_\k^\F\cup\PL_\k^\T & \vartriangle\cap\cup & \PL_{\,\i\,\j\,\k}^{\overline{\F\F}-} \\
\multicolumn{2}{c}{=\PL_{\times \k}^\F=\PL_\i^\a\vartriangle\PL_\j^{\sim\a}} & =  \PL_\k^\F\cup\PL_\k^\T& \PL_\i^\T & \PL_\j^\F & \PL_\k^\F\cup\PL_\k^\T & \vartriangle\cap\cup & \PL_{\,\i\,\j\,\k}^{\overline{\T\T}-}  \\
\{(\bot_k^{\downarrow},\T),(\top_k^{\uparrow},\T), & (\bot_k^{\uparrow},\T), (\top_k^{\downarrow},\T)\} & \{(\bot_\ORccFive,\F),(\top_\ORccFive,\F)\} & \PL_\i^\F & \PL_\j^\F & \PL_\k^\F\cup\PL_\k^\T & \vartriangle\cap\cup & \PL_{\,\i\,\j\,\k}^{\overline{\F\T}-} \\
\multicolumn{2}{c}{=\PL_{\times \k}^\T=\PL_\i^\a\vartriangle\PL_\j^{\a}} & =  \PL_\k^\F\cup\PL_\k^\T& \PL_\i^\T & \PL_\j^\T & \PL_\k^\F\cup\PL_\k^\T & \vartriangle\cap\cup & \PL_{\,\i\,\j\,\k}^{\overline{\T\F}-}  \\
\ldots & \ldots &\ldots &\ldots & \ldots & \ldots & \ldots & \ldots\\
\{(\bot_\ORccFive,\T),(\bot_k,\T), & (\top_\ORccFive,\T), (\top_k,\T)\} & \{(\bot_\ORccFive,\T)\} & \PL_\i^\F & \PL_\j^\T & \PL_\k^\F & \vartriangle\cap & \PL_{\,\i\,\j\,\k}^{\overline{\F\F}\F} \\
\multicolumn{2}{c}{=\PL_{\times \k}^\F=\PL_\i^\a\vartriangle\PL_\j^{\sim\a}} & =  \PL_\k^\F& \PL_\i^\T & \PL_\j^\F & \PL_\k^\F & \vartriangle\cap & \PL_{\,\i\,\j\,\k}^{\overline{\T\T}\F}  \\
\ldots & \ldots &\ldots &\ldots & \ldots & \ldots & \ldots & \ldots\\
\{(\bot_k^{\downarrow},\T),(\top_k^{\uparrow},\T), & (\bot_k^{\uparrow},\T), (\top_k^{\downarrow},\T)\} & \{(\top_\ORccFive,\T)\} & \PL_\i^\F & \PL_\j^\F & \PL_\k^\T & \vartriangle\cap & \PL_{\,\i\,\j\,\k}^{\overline{\F\T}\T} \\
\multicolumn{2}{c}{=\PL_{\times \k}^\T=\PL_\i^\a\vartriangle\PL_\j^{\a}} & =  \PL_\k^\T& \PL_\i^\T & \PL_\j^\T & \PL_\k^\T & \vartriangle\cap & \PL_{\,\i\,\j\,\k}^{\overline{\T\F}\T} \\  
\ldots & \ldots &\ldots &\ldots & \ldots & \ldots & \ldots & \ldots\\
\multicolumn{8}{c}{\txt{notational convention: $i = \min (\IdxSet\setminus\{k\}), \IdxSet=\{i,j,k\}$} } \\
\end{array}$$
\end{table}

\begin{table}[H]
\caption{\label{tab:CombinedQuestionsIformsIstatesCum}Cumulatively combined elementary questions and elementary \istates\ and the elementary systems that are in those \istates\ which arise from the following substitutions in Eq. \ref{eq:CompositionSchemata}\ : $\Q_\delta?\rightsquigarrow Q_i?$, $\a_\delta\rightsquigarrow\a$, $\Q_\epsilon?\rightsquigarrow Q_j?$, $\b_\epsilon\rightsquigarrow\b$, $\bigcirc\rightsquigarrow\odot$, $\Diamond\rightsquigarrow(\id,\id)$,  $\PL_\delta^{\a_\delta} \rightsquigarrow \PL_\i^\a$, $\PL_\epsilon^{\b_\epsilon}  \rightsquigarrow \PL_\j^\b$, $\cdot\rightsquigarrow\cap$. \cite{bittner:InformationMereologyAndVagueness,Bittner:VMQI}}
    $$\begin{array}{r|l}
         & \odot : (\IdxSet\to \OReg^2_\nu \to \YN) \to (\IdxSet\to\OReg^2_\nu\to \YN) \to \YN^2\\
  \text{questions} & Q_i?(r_1,r_2)_o \odot Q_j?(r_1^\prime,r_2^\prime)_o = \left\{\begin{array}{l} \textsf{undefined}\quad \text{if}\ (r_1,r_2)_o \neq (r_1^\prime,r_2^\prime)_o\ \text{or}\ i = j\\
  (Q_i?(r_1,r_2)_o,Q_j?(r_1,r_2)_o)\quad \text{otherwise} \\
\end{array} \right.  \\
 & \text{abbrev:}\  ({Q_i?\odot Q_j?})(r_1,r_2)_o  \equiv  Q_i?(r_1,r_2)_o \odot Q_j?(r_1,r_2)_o\quad \text{when}\ \odot\ \text{is defined}.  \\
 \hline
\text{support} & 	s \Models ((Q_i?\odot Q_j?)(r_1,r_2)_o=(\a,\b))  \equiv s = \{\}\ \text{or}\ \exists  s_i,s_j.\ \{\} \subset s_i,s_j \subset \ORccFive\ \text{and}   \\
&	\qquad\qquad s_i \Models (Q_i?(r_1,r_2)_o=\a)\ \ \text{and}\ s_j \Models (Q_j?(r_1,r_2)_o=\b)  \ \text{and}\ s = s_i \cap s_j  \\\hline 
\text{max} & \lrceil{(Q_i?\odot Q_j?)(r_1,r_2)_o=(\a,\b)} \equiv \bigcup \{ s\subseteq \ORccFive.\ s \Models (Q_i?\odot Q_j?)(r_1,r_2)_o=(\a,\b) \}\\
	\istates &	\text{if}\ \lrceil{(Q_i?\odot Q_j?)(r_1,r_2)_o=(\a,\b)} \neq\{\}\ \text{then} \\
	& \qquad \lrceil{(Q_i?\odot Q_j?)(r_1,r_2)_o=(\a,\b)}  = \PL_\i^\a \cap \PL_\j^\b = \tensor*{\PL}{_{\i\ \j}^{\a\b}} = \tensor*{\PL}{_{\i\ \j\ \k}^{\a\b-}}  \\  \hline
\text{partitions} &  \{\PL_{\,\i\, \j\, \k}^{\a\b-} \} \supseteq \{ \lrceil{Q_i?\odot Q_j?(r_1,r_2)_o=(\a,\b)} \neq \{\} \mid (r_1,r_2)_o\in\OReg^2_\IdxTwoBitPR, \a,\b\in\YN\}    \\

	   \end{array}$$
\end{table}

\begin{table}[H]
\caption{\label{tab:CombinedQuestionsIformsIstatesNonCum}Non-cumulatively combined elementary questions and elementary \istates\ and the elementary systems that are in those \istates\ which arise from the following substitutions in Eq. \ref{eq:CompositionSchemata}\ : $\Q_\delta?\rightsquigarrow Q_i?$, $\a_\delta\rightsquigarrow\a$, $\Q_\epsilon?\rightsquigarrow Q_j?$, $\b_\epsilon\rightsquigarrow\b$, $\bigcirc\rightsquigarrow\oplus$, $\diamond\rightsquigarrow\xor$,  $\PL_\delta^{\a_\delta} \rightsquigarrow \PL_\i^\a$, $\PL_\epsilon^{\b_\epsilon}  \rightsquigarrow \PL_\j^{\sim\b}, \cdot\rightsquigarrow\vartriangle$. \cite{bittner:InformationMereologyAndVagueness,Bittner:VMQI}}
    $$\begin{array}{r|l}
\text{questions} &     \begin{array}{l}
    \xor : \YN \to \YN \to \YN; \quad \a \xor \b = \left\{ \begin{array}{ll}
        \Y & \text{if}\ \a \neq \b\\
        \N & \text{if}\ \a = \b\\
    \end{array}\right.\\
    \oplus : (\IdxSet\to \OReg^2_\times\to \YN) \to (\IdxSet\to\OReg^2_\times\to \YN) \to \YN\\
        Q_i?(r_1,r_2)_o \oplus Q_j?(r_1^\prime,r_2^\prime)_o = \left\{\begin{array}{l} \textsf{undefined}\quad \text{if}\ (r_1,r_2)_o \neq (r_1^\prime,r_2^\prime)_o\ \text{or}\ i = j\\
  (Q_i?(\sharp_o^{\times k}(r_1,r_2)) \xor Q_j?(\sharp_o^{\times k}(r_1,r_2)))\\ \qquad \text{with $\{i,j,k\}=\IdxSet$ otherwise} \\
\end{array} \right.\\
    \end{array}
\\
&  \text{abbrev:}\   \begin{array}{rcl}
        	  ({Q_i?\oplus Q_j?})(r_1,r_2)_o  &\equiv & Q_i?(r_1,r_2)_o \oplus Q_j?(r_1,r_2)_o\\
	\qquad\qquad Q_{\times k}?(r_1,r_2)_o &\equiv & ({Q_i?\oplus Q_j?})(r_1,r_2)_o\ \text{with}\ \{i,j\} = \IdxSet\setminus \{k\}\ \\
    \end{array}\\\hline
\text{support} & 
	s \Models ((Q_i?\oplus Q_j?)(r_1,r_2)_o=\a\xor\b)  \equiv s = \{\}\ \text{or}\ \big((Q_i?\oplus Q_j?)(r_1,r_2)_o=\a\xor\b\\
& \qquad\qquad \text{and}\ s \subseteq  \PL_\i^\a \vartriangle \PL_\j^{\sim\b} = \PL_{\i\, \j\, \k}^{\overline{\a\b}-}\big) \\\hline
	\begin{array}{c} \text{max}\\ \istates\end{array} &
	\begin{array}{l}
	\text{if}\ \lrceil{(Q_i?\oplus Q_j?)(r_1,r_2)_o= (\a\xor\b)} \neq \{\}\ \text{then}\\
	\qquad\qquad\qquad \lrceil{(Q_i?\oplus Q_j?)(r_1,r_2)_o= (\a\xor\b)}  =  \PL_\i^\a \vartriangle \PL_\j^{\sim\b} = \PL_{\i\, \j\, \k}^{\overline{\a\b}-}\\
\end{array} \\\hline
	\text{partitions} & \{\PL_{\i\ \j\ \k}^{\overline{\a\b}-}\}  \supseteq  \{ \lrceil{Q_i?\oplus Q_j?(r_1,r_2)_o=(\a\xor\b)} \neq \{\} \mid (r_1,r_2)_o\in\OReg^2_\IdxOneBitET, \\
	 & \qquad\qquad\qquad \qquad\qquad\qquad\qquad \qquad   \a,\b\in\YN\}  \\
	   \end{array}$$
\end{table}

\begin{table}[H]
\caption{\label{tab:CombinedQuestionsIformsIstatesThree}Cumulatively and non-cumulatively combined three elementary questions, three elementary \istates\ and the elementary systems that are in those \istates. \cite{bittner:InformationMereologyAndVagueness,Bittner:VMQI}}
    $$\begin{array}{r|l}
\text{questions} &  \begin{array}{rl}
    (\bigcirc\bigcirc) & ::  (\IdxSet \to \OReg^2_\nu \to \YN) \to (\IdxSet \to \OReg^2_\nu \to \YN) \to\\
    & \qquad\qquad\qquad\qquad\qquad \quad (\IdxSet \to \OReg^2_\nu \to \YN) \to \YN^n\\
 \text{where} &  \bigcirc \in \{\odot,\oplus\}\ \text{and}\ \left\{ {\arraycolsep=2pt \begin{array}{ll}\text{if}\ \bigcirc\bigcirc = \odot\odot\ \text{then} & n = 3\ \text{and}\ \nu = \IdxThreeBitPR\\
\text{if}\ \bigcirc\bigcirc = \oplus\oplus\ \text{then} &  n=2\ \text{and}\ \nu = \IdxThreeBitET\\ 
\text{otherwise} & n=2\ \text{and}\ \nu = \IdxTwoBitET\\
\end{array}} \right.\\
  (\bigcirc\bigcirc)& Q_i?(r_1,r_2)_o\ Q_j?(r_1^\prime,r_2^\prime)_o\ Q_k?(r_1^{\prime\prime},r_2^{\prime\prime})_o\ \text{ is undefined if}\\ 
& \qquad \neg((r_1,r_2)_o = (r_1^\prime,r_2^\prime)_o =  (r_1^{\prime\prime},r_2^{\prime\prime})_o) \
     \text{or}\ i = j\ \text{or}\ i = k\ \text{or}\ k = j\\
\text{abbrev:}\ & ({Q_i?\bigcirc Q_j?\bigcirc Q_k?})(r_o)  \equiv  Q_i?(r_o) \bigcirc Q_j?(r_o) \bigcirc Q_k?(r_o)\\
 & \qquad\qquad \text{whenever}\ (\bigcirc\bigcirc)\ \text{is defined, with}\ r = (r_1,r_2)_o  
    \end{array}  \\

& \begin{array}{l}
 Q_i?\bigcirc Q_j?\bigcirc Q_k? (r_o)  = \left\{\arraycolsep=1pt\begin{array}{ll} 
    (Q_i?\oplus Q_j?(r_o),Q_k?(r_o))&  \text{if}\ \bigcirc\bigcirc = \oplus\odot\\
 (Q_i?(r_o), Q_j?\oplus Q_k?(r_o)) & \text{if}\ \bigcirc\bigcirc = \odot\oplus\\
       (Q_i?\oplus Q_j?(r_o),
      \ Q_j? \oplus Q_k?(r_o)) & \text{if}\ \bigcirc\bigcirc = \oplus\oplus\\
    (Q_i?(r_o), Q_j?(r_o),Q_k?(r_o)) & \text{if}\ \bigcirc\bigcirc = \odot\odot\\
\end{array} \right. \\
    \end{array}  \\\hline
	& \oplus\odot:\ \Q_\delta?\rightsquigarrow (Q_i?\oplus Q_j?), \a_\delta\rightsquigarrow(\a\xor\b), \Q_\epsilon?\rightsquigarrow Q_k?, \b_\epsilon\rightsquigarrow\c, \bigcirc\rightsquigarrow\odot, \\
& \qquad \quad \diamond\rightsquigarrow(\id,\id), \PL_\delta^{\a_\delta} \rightsquigarrow \PL_{\,\i\,\j}^{\overline{\a\b}}, \PL_\epsilon^{\b_\epsilon}  \rightsquigarrow \PL_\k^\c, \cdot\rightsquigarrow\cap  \\ 
	& \odot\oplus:\ \Q_\delta?\rightsquigarrow Q_i?, \a_\delta\rightsquigarrow\a, \Q_\epsilon?\rightsquigarrow (Q_j?\oplus Q_k?), \b_\epsilon\rightsquigarrow(\b\xor\c), \bigcirc\rightsquigarrow\odot, \\
\txt{subst. in\\Eq. \ref{eq:CompositionSchemata}} & \qquad \quad \diamond\rightsquigarrow(\id,\id), \PL_\delta^{\a_\delta} \rightsquigarrow \PL_\i^\a, \PL_\epsilon^{\b_\epsilon}  \rightsquigarrow \PL_{\,\j\,\k}^{\overline{\b\c}}, \cdot\rightsquigarrow\cap \\ 
	& \oplus\oplus:\ \Q_\delta?\rightsquigarrow (Q_i?\oplus Q_j?), \a_\delta\rightsquigarrow(\a\xor\b), \Q_\epsilon?\rightsquigarrow (Q_j?\oplus Q_k?),  \\
& \qquad \quad \b_\epsilon\rightsquigarrow(\b\xor\c), \bigcirc\rightsquigarrow\odot, \diamond\rightsquigarrow(\id,\id), \PL_\delta^{\a_\delta} \rightsquigarrow \PL_{\,\i\, \j}^{\overline{\a\b}}, \PL_\epsilon^{\b_\epsilon}  \rightsquigarrow \PL_{ \j\, \k}^{\overline{\b \c}}, \cdot\rightsquigarrow\cap  \\ 
	& \odot\odot:\ \Q_\delta?\rightsquigarrow (Q_i?\odot Q_j?), \a_\delta\rightsquigarrow(\a,\b), \Q_\epsilon?\rightsquigarrow Q_k?, \b_\epsilon\rightsquigarrow\c, \bigcirc\rightsquigarrow\odot, \diamond\rightsquigarrow(\id,\id),\\
& \qquad \quad \PL_\delta^{\a_\delta} \rightsquigarrow \PL_{\,\i\,\j}^{\a\b}, \PL_\epsilon^{\b_\epsilon}  \rightsquigarrow \PL_\k^\c, \cdot\rightsquigarrow\cap  \\ 

	   \end{array}$$
\end{table}

\begin{table}[H]
\caption{Glossary of Major Notations}
\label{tab:Glossary}
\centering
\begin{tabular}{p{2cm} p{9.5cm} p{2.5cm}}
\hline
\textbf{Symbol} & \textbf{Informal Meaning} & \textbf{Formal Defs} \\ \hline \hline
\multicolumn{3}{l}{\textbf{Foundational Sets and Classes}} \\ \hline
$\cReg, \vReg$, \newline $\zReg$, $\eReg$ & The classification of regions as Crisp regions (\cReg) and Vague regions (\vReg, \zReg, \eReg\ in increasing degree of vagueness). & Sec.~\ref{sec:MereologicalRelations} \\
$\ORccFive$ & The set of 8 classical mereological relations (the vertices of the cube).  & Sec.~\ref{sec:RccFiveRelations} \\
$\OReg^2_\nu$ & A class of relational elementary systems (an observer and a pair of regions) with a specific information capacity type, $\nu$. & Sec.~\ref{sec:RegOneBitAndRegTwoBitAndRegThreeBit} \\
$\ORccFive_\nu$ & A class of $\istates$ (subsets of $\ORccFive$) with a specific information type, $\nu$.  & Sec.~\ref{sec:PartitionLattice} \\
\hline
\multicolumn{3}{l}{\textbf{Core Concepts}} \\ \hline
$\istate$ & An information state; a set of possible classical relations.  & Sec.~\ref{sec:SubspacesIstates} \\
$\iform$ & An information form; an axiomatic type that defines a system's information capacity.  & Sec.~\ref{sec:RegOneBitAndRegTwoBitAndRegThreeBit} \\
\hline
\multicolumn{3}{l}{\textbf{Index Notations}} \\ \hline
$\PL_{\i\j\k}^{\alpha\beta\gamma}$ & The index notation for an $\istate$, representing a cell in a partition. & Sec.~\ref{sec:PartitionLattice}, Tab.~\ref{tab:VectorsInIndexNotation} \\
$\QQ_{\i\j\k}?$ & The index notation for a (potentially complex) mereological question.  & Sec.~\ref{sec:CombiningQuestions}, Tab.~\ref{tab:defQuestionPattern} \\
$\A^{\alpha\beta\gamma}$ & The index notation for an answer pattern, corresponding to a $\QQ$ and a $\PL$. & Sec.~\ref{sec:CombiningQuestions}, Tab.~\ref{tab:defQuestionPattern} \\
\hline
\multicolumn{3}{l}{\textbf{Maps and Operators}} \\ \hline
$\rccfive$ & The classical map from a pair of crisp regions to its 3-bit pattern. & Sec.~\ref{sec:RccFiveRelations} \\
$\sharp, \sharp_o^i$ & The fundamental interaction map ($\sharp,\sharp^\times: \IdxSet \to \OO \to \OReg^2 \to \Reg^2$); $\sharp_o^i$ refers to its application for a single interaction indexed by $i$. & Sec.~\ref{sec:MereologicalRelations} \\
$\overline{\sharp}$ & A description of the output of the $\sharp$ map for all relevant indices; refers to the indexed family of outcomes, $(\sharp_o^\alpha)_{\alpha \in \IdxSetNu}$. & Sec.~\ref{sec:CompositionalityRoadmap} \\
$\sharp\sharp$ & Benchmark map used to formally define the determinacy of an interaction. & Sec.~\ref{sec:DeterminismAndInformation} \\
$\lceil \dots \rceil$ & The maximum supporting $\istate$ operator for the direct semantic path. & Sec.~\ref{sec:LinkingIstates} \\
$\overline{\lfloor\cdot\rfloor}$ & The quantization map for the indirect quantization path. & Sec.~\ref{sec:QuantizationMap} \\\hline
\multicolumn{3}{l}{\textbf{Combination Logic}} \\ \hline
$\odot$ & Cumulative combination of questions/information.  & Sec.~\ref{sec:CompositionalityRoadmap} \\
$\cap$ & Set intersection; the formal operator for cumulative combination ($\odot$). & Sec.~\ref{sec:PartitionLattice}; Sec.~\ref{sec:CompositionalityRoadmap} \\$\oplus$ & Non-cumulative (entangled) combination of questions/information.  & Sec.~\ref{sec:CombiningQuestions} \\
$\vartriangle$ & Symmetric difference; the formal operator for non-cumulative combination ($\oplus$). & Sec.~\ref{sec:PartitionLattice}; Sec.~\ref{sec:CombiningQuestions} \\\hline
\end{tabular}
\end{table}
\end{document}